\documentclass[11pt]{article}

\usepackage[margin=1in]{geometry}
\usepackage{amsmath}
\usepackage{appendix}
\usepackage{amsthm}
\usepackage{bbm}
\usepackage{amssymb}
\usepackage{setspace}
\usepackage{color}
\usepackage{float}
\usepackage{dsfont}
\usepackage{enumerate}
\usepackage{hyperref}
\usepackage[capitalize]{cleveref}
\usepackage{graphicx}
\usepackage{algpseudocode,algorithm,algorithmicx}
\usepackage{tikz}
\usepackage{subcaption} 
\usepackage{wrapfig}
\usepackage{makecell}

\newtheorem{theorem}{Theorem}
\newtheorem{proposition}[theorem]{Proposition}
\newtheorem{lemma}[theorem]{Lemma}
\newtheorem{corollary}[theorem]{Corollary}

\theoremstyle{definition}
\newtheorem{definition}{Definition}
\newtheorem{remark}[definition]{Remark}

\numberwithin{theorem}{section} 
\numberwithin{definition}{section}

\numberwithin{equation}{section}

\usepackage{setspace}

\def\({\Big(}
\def\){\Big)}
\def\C{\mathbb{C}}

\def\E{\mathbb{E}}
\def\bbF{\mathbb{F}}
\def\N{\mathbb{N}}

\def\P{\mathbb{P}}
\def\R{\mathbb{R}}

\def\T{\mathbb{T}}
\def\bbU{\mathbb{U}}
\def\Z{\mathbb{Z}}

\DeclareMathOperator*{\argmin}{arg\,min}

\renewcommand{\subset}{\subseteq}
\renewcommand{\hat}{\widehat}
\renewcommand{\tilde}{\widetilde}
\renewcommand{\epsilon}{\varepsilon}
\renewcommand{\Re}{\text{Re}}

\def\range{{\rm range}}

\def\NSR{{\rm NSR}}

\def\mesh{{\rm mesh}}
\def\tr{{\rm tr}}

\newcommand{\bfa}{{\boldsymbol a}}

\newcommand{\bfc}{{\boldsymbol c}}

\newcommand{\bfe}{{\boldsymbol e}}

\newcommand{\bfu}{{\boldsymbol u}}
\newcommand{\bfv}{{\boldsymbol v}}
\newcommand{\bfw}{{\boldsymbol w}}
\newcommand{\bfx}{{\boldsymbol x}}
\newcommand{\bfy}{{\boldsymbol y}}

\newcommand{\bfA}{{\boldsymbol A}}
\newcommand{\bfB}{{\boldsymbol B}}

\newcommand{\bfG}{{\boldsymbol G}}

\newcommand{\bfI}{{\boldsymbol I}}
\newcommand{\bfM}{{\boldsymbol M}}

\newcommand{\bfQ}{{\boldsymbol Q}}
\newcommand{\bfS}{{\boldsymbol S}}
\newcommand{\bfT}{{\boldsymbol T}}
\newcommand{\bfU}{{\boldsymbol U}}
\newcommand{\bfV}{{\boldsymbol V}}
\newcommand{\bfW}{{\boldsymbol W}}

\newcommand{\bfzero}{{\boldsymbol 0}}

\newcommand{\bfeta}{{\boldsymbol \eta}}
\newcommand{\bfzeta}{{\boldsymbol \zeta}}
\newcommand{\bfmu}{{\boldsymbol \mu}}

\newcommand{\bfphi}{{\boldsymbol \phi}}
\newcommand{\bfPhi}{{\boldsymbol \Phi}}
\newcommand{\bPsi}{{\boldsymbol \Psi}}
\newcommand{\bfSigma}{{\boldsymbol \Sigma}}

\newcommand{\calA}{\mathcal{A}}
\newcommand{\calB}{\mathcal{B}}

\newcommand{\calS}{\mathcal{S}}

\newcommand{\calF}{\mathcal{F}}

\newcommand{\calN}{\mathcal{N}}
\newcommand{\calP}{\mathcal{P}}
\newcommand{\calT}{\mathcal{T}}
\newcommand{\calU}{\mathcal{U}}

\newcommand{\calY}{\mathcal{Y}}

\newcommand{\diag}{{\rm diag}}

\def\wherespace{\quad\text{where}\quad}

\def\foreachspace{\quad\text{for each}\quad}
\def\forallspace{\quad\text{for all}\quad}
\def\ifspace{\quad\text{if}\quad}
\def\andspace{\quad\text{and}\quad}

\numberwithin{equation}{section}

\title{Optimality of Gradient-MUSIC for Spectral Estimation}
\author{Albert Fannjiang\footnote{University of California, Davis. Email: fannjiang@math.ucdavis.edu} \and Weilin Li\footnote{City University of New York, City College. Email: wli6@ccny.cuny.edu} \and Wenjing Liao\footnote{Georgia Institute of Technology. Email: wliao60@gatech.edu}}

\begin{document}
	
	\maketitle
	
	\begin{abstract}		
		We introduce the Gradient-MUSIC algorithm for estimating the unknown frequencies and amplitudes of a nonharmonic signal from noisy time samples. While the classical MUSIC algorithm performs a computationally expensive search over a fine grid, Gradient-MUSIC is significantly more efficient and eliminates the need for discretization over a fine grid by using optimization techniques. It coarsely scans the 1D landscape to find initialization simultaneously for all frequencies followed by parallelizable local refinement via gradient descent. We also analyze its performance when the noise level is sufficiently small and the signal frequencies are separated by at least $8\pi/m$, where $\pi/m$ is the standard resolution of this problem. Even though the 1D landscape is nonconvex, we prove a global convergence result for Gradient-MUSIC: coarse scanning provably finds suitable initialization and gradient descent converges at a linear rate. In addition to convergence results, we also upper bound the error between the true signal frequencies and amplitudes with those found by Gradient-MUSIC. For example, if the noise has $\ell^\infty$ norm at most $\epsilon$, then the frequencies and amplitudes are recovered up to error at most $C\epsilon/m$ and $C\epsilon$ respectively for a universal $C>0$, which are minimax optimal in $m$, $\epsilon$, and number of frequencies. Our theory can also handle stochastic noise with performance guarantees under nonstationary independent Gaussian noise. 
		Our main approach is a comprehensive geometric analysis of the landscape, a perspective that has not been explored before.
	\end{abstract}
	
	\medskip
	\noindent
	{\bf 2020 Math Subject Classification:}  42A05, 42A10, 90C26, 94A12
	
	
	\tableofcontents

	\addtocontents{toc}{\protect\setcounter{tocdepth}{1}}
	
	\section{Introduction}
	
	\subsection{Motivation} 
	
	The spectral estimation problem is to accurately estimate the frequencies and amplitudes $(\bfx,\bfa):=\{(x_j,a_j)\}_{j=1}^s$ of a nonharmonic Fourier sum,
	\begin{equation}
		\label{eq:hfun}
		h(\xi) = \sum_{j=1}^s a_j e^{i x_j \xi},
	\end{equation}
	given noisy measurements of $h$ at a finite collection of $\xi$. After normalization, we consider a canonical setting where $x_j\in \R/2\pi \Z=:\T$, the clean measurements $\bfy := \{h(k)\}_{k=-m+1,\dots,m-1} \in \C^{2m-1}$ are collected at $2m-1$ consecutive integers, and the noisy measurements $\tilde\bfy \in \C^{2m-1}$ are perturbed from the clean measurements by the noise vector $\bfeta \in \C^{2m-1}$:
	\begin{equation}
		\tilde\bfy = \bfy+\bfeta, \andspace \bfy := \{h(k)\}_{k=-m+1,\dots,m-1}.
		\label{eq:noisymeasurement}
	\end{equation}
	The noise vector $\bfeta$ can be deterministic or random. 
	This inverse problem arises in many interesting applications in imaging and signal processing,
	including {\it direction-of-arrival} (DoA) estimation \cite{krim1996two}, inverse source and inverse scattering \cite{kirsch2002music,fannjiang2010compressive,fannjiang2011music}, electroencephalogram (EEG) signal analysis \cite{subha2010eeg}, and nuclear magnetic resonance  spectroscopy \cite{li1998high}.

	In the spectral estimation problem, both the frequencies and amplitudes are unknown, which makes this a nonlinear inverse problem. Without noise, the unknown frequencies can be exactly recovered up to machine precision by $2s$ measurements. In the presence of noise, the main difficulty of spectral estimation lies in estimating the frequencies, after which the amplitudes can be estimated by solving an over-determined linear system of equations. 
	
	Research on spectral estimation has focused on developing robust algorithms and analyzing their recovery errors from a theoretical perspective.
	The first solution to spectral estimation can be traced back to Prony \cite{Prony} in 1795. The Prony's method uses exactly $2s$ measurements and requires finding all roots of a degree $s$ trigonometric polynomial. Subspace methods such as classical MUSIC \cite{schmidt1986multiple} and ESPRIT \cite{kailath1989esprit} were invented for multi-snapshot spectral estimation problem whereby the amplitudes vary randomly and a large number of snapshots of data are collected. These methods can be adapted to tackle the single-snapshot spectral estimation problem (also referred to as the line spectral estimation problem in some papers), which was done in \cite{hua1990matrixpencil,moitra2015matrixpencil,liao2016music,li2020super}. These methods are typically grouped together despite significant differences because they share the same first step of finding an approximation of the ``signal subspace". Convex methods for spectral estimation were developed more recently. Several closely related methods include total variation minimization, atomic norm minimization, and Beurling-LASSO, see \cite{de2012exact,tang2013offgrid,candes2014towards,duval2015exact}. 
	
	Classical MUSIC was among the first stable methods to achieve high-resolution recovery. It creates a {\it 1D landscape function} and finds its $s$ smallest local minima, which act as the estimated signal frequencies. Evaluating the landscape function on a fine grid is computationally expensive, and without additional information, classical MUSIC finds the local minima through an exhaustive grid search. Consequently, its resolution and accuracy are tied to the grid spacing: finer grids improve precision, but incur significantly higher computational cost. While this computational drawback motivated the development of more efficient algorithms, classical MUSIC is especially attractive in practical applications due to its simplicity and generality.
	
	As for theoretical analysis, it is now well-understood that the stability of spectral estimation to noise depends on the minimum separation between the frequencies, denoted $\Delta:=\Delta(\bfx)$, relative to $m$, which represents the aperture in imaging settings \cite{den1997resolution}. The seminal paper \cite{donoho1992superresolution} illustrated that spectral estimation is fundamentally different depending on whether $m\Delta\gg \pi$ (which we call {\it well-separated case}) or $m\Delta\ll \pi$ (which we call {\it super-resolution regime}). Importantly, $\pi/m$ is called the {\it Rayleigh length} and is regarded as the standard resolution of this problem \cite{den1997resolution}. Typically, the well-separated case and super-resolution regime are treated independently in theoretical analysis because they require separate techniques and have different error rates. This paper focuses on theoretical analysis for the well-separated case and super-resolution is beyond the scope of this paper.  

	Motivated by classical MUSIC’s widespread use and its reliance on an exhaustive grid search, we aim to develop a computationally efficient algorithm that reduces its computational cost while providing provable error guarantees. From a computational perspective, the algorithm should be computationally tractable and significantly faster than classical MUSIC. From a theoretical perspective, an ideal algorithm is provably optimal, achieving estimation error that matches the minimax lower bound for the well-separated case, up to logarithmic factors. The most important parameters for this setting are the number of samples $2m-1$, an appropriate notion of noise level, and number of frequencies $s$.
	
	\subsection{Our approach and contributions}
	
	In this work, we introduce Gradient-MUSIC, an algorithm that is significantly more efficient than classical MUSIC and is provably optimal in the well-separated regime under both deterministic and i.i.d. Gaussian noise, specifically when $m\Delta \geq 8\pi$.
	
	{\bf Algorithm Contribution.} We propose an alternative strategy inspired by nonconvex optimization techniques: \textit{coarse scanning followed by local refinement by gradient descent}. The proposed algorithm is called Gradient-MUSIC, which avoids an expensive grid search that classical MUSIC suffers from. Gradient-MUSIC evaluates the landscape function on a coarse grid,  locates a set of suitable simultaneous initialization, and runs gradient descent to produce iterates that converge to the relevant  local minima of the landscape function at a linear rate.  Gradient-MUSIC has several key features: 
	\begin{enumerate}[(a)] \itemsep-2pt
		
		\item \textit{Precise coarse grid for initialization.} To find suitable initialization, the landscape function only needs to be evaluated on a grid of width $1/(2m)$, just a constant factor smaller than the Rayleigh length $\pi/m$. While it is possible that the numerical constant can be improved, the $1/m$ scaling cannot be altered. 
		
		\item 
		{\it Grid-less in outcome.} Our method utilizes a grid only to find good initialization. Afterwards, gradient descent produces iterates that live on
		the continuous domain $\R/2\pi\Z$, so Gradient-MUSIC outputs estimated parameters that do not suffer from discretization, in stark contrast to classical MUSIC. 
		
		\item \textit{Speed without compromise.} Evaluating the landscape function on a coarse grid keeps the computational cost modest, while the local refinement step enjoys fast linear convergence. In addition,  the local refinement step is parallelizable with the {\em simultaneous initialization} on the coarse grid, further speeding up the computation. Importantly, the improved speed of Gradient-MUSIC is not achieved by sacrificing accuracy, since our theoretical results show that it is optimal in many aspects. Under the assumptions of this paper, one should always use Gradient-MUSIC over its classical counterpart.
	\end{enumerate}
	
	\renewcommand{\arraystretch}{1.25} 
	\begin{table}[t]
		\centering
		\begin{tabular}{|c|c|c|c|} \hline
			{\bf Reference} &{\bf Noise assumption} &{\bf Frequency error} &{\bf Amplitude error} \\
			&$\bfeta\in \C^{2m-1}$ &$\max_j|x_j-\hat x_j|$ &$\max_j |a_j-\hat a_j |$ \\ \hline
			\cref{cor:deterministic}  &\thead{$\|\bfeta\|_p\leq c m^{1/p}$ \\ $c>0$ small, $p\in [1,\infty]$}  &$C\|\bfeta\|_p/m^{1+1/p}$ & $C\|\bfeta\|_p/m^{1/p}$ \\ \hline
			\cref{cor:stochastic}  &\thead{$\bfeta\sim \calN(\bfzero,\bfSigma)$, $|r|<1/2$, $\sigma>0$\\ $\bfSigma=\diag\big(\{\sigma^2(1+|k|)^{2r}\}_k\big)$}  &$C\sigma \sqrt{\log(m)}/m^{3/2-r}$ &$C\sigma \sqrt{\log(m)}/m^{1/2-r}$ \\\hline
		\end{tabular}
		\caption{Gradient-MUSIC outputs $(\hat\bfx,\hat\bfa)$ where $\hat\bfx$ has been indexed to best match $\bfx$. The conclusions only show dependence on $m$, $s$, and noise parameters. $C,c>0$ are  universal constants and $\sigma>0$ is arbitrary. For the stochastic model, the error bounds hold with probability at least $1-o(m)$ as $m\to\infty$. }
		\label{tab:mainresults}
	\end{table} 
	
	{\bf Theoretical Contribution.} We provide a stability analysis of Gradient-MUSIC that is highly general, which allows one to treat different types of perturbations in a unified manner. For demonstration, we specialize to two canonical noise models. 
	\begin{enumerate}[(a)]
		\item 
		The first is arbitrary $\bfeta$, i.e. deterministic or adversarial, and the noise level is measured by its $\ell^p$ norm $\|\bfeta\|_p$. 
		\item 
		The second is $\bfeta\sim \calN(\bfzero,\bfSigma)$ for $\bfSigma:=\diag\big(\{\sigma^2(1+|k|)^{2r}\}_k\big)$, where a parameter $r\in \R$ controls whether the noise variance stays constant ($r=0$), grows ($r>0$), or decays ($r<0$). 
	\end{enumerate}
	For these two models, the performance guarantees of Gradient-MUSIC are summarized in \cref{tab:mainresults}. These results are optimal in several aspects, which is discussed in \cref{sec:implications}.
	
	Our main results are proved through \textit{a comprehensive geometric analysis of the landscape function}, which is summarized in \cref{thm:landscape} and is the main technical contribution of this paper. Although MUSIC has been widely used in applications, to our best knowledge, a geometric analysis of the landscape function has not been explored before. Our geometric analysis also implies that classical MUSIC is minimax optimal, provided the fine grid spacing is small enough. This requirement, however, incurs a significant computational cost and is the main drawback of classical MUSIC. 
	
	\subsection{Discussion and consequences}
	\label{sec:implications} 
	
	{\bf Optimality for deterministic or adversarial perturbations.} The minimax error is the accuracy achieved by the best algorithm(s) evaluated on the worst case parameter and perturbation. It is not obvious that an optimal and tractable algorithm exists because by definition, the minimax error is taken over all possible functions from data to parameters, including those that are computationally intractable or non-computable. Many papers on deterministic minimax rates pertain to the super-resolution regime, see \cite{donoho1992superresolution,demanet2015recoverability,li2021stable,batenkov2021super,liu2022mathematical}. Similar techniques are applicable to the well-separated case, but we could not find a direct result.
	
	There are three aspects in which Gradient-MUSIC is optimal, up to universal constants.
	\begin{enumerate}[(a)]
		\item 
		{\it Rate-optimality.} We prove in \cref{lem:minimax} that the minimax frequency and amplitude errors are $\Omega(\|\bfeta\|_p/m^{1+1/p})$ and $\Omega(\|\bfeta\|_p/m^{1/p})$ respectively. These minimax lower bounds match the upper bounds achieved by Gradient-MUSIC, which certifies optimality in the noise level $\|\bfeta\|_p$, number of samples $2m-1$, and number of frequencies $s$. 
		\item 
		{\it Noise tolerance.} Our theory for Gradient-MUSIC requires that $\|\bfeta\|_p\leq cm^{1/p}$ for a small enough universal $c>0$. This noise assumption is necessary regardless of which method is used, because spectral estimation is generally impossible otherwise, see \cref{rem:impossible}. 
		\item 
		{\it Separation condition.} Our theory for Gradient-MUSIC holds under the separation condition $m\Delta\geq 8\pi$. This requirement does not depend on the total number of frequencies $s$, so importantly, it possible for $m$ to be proportional to $s$. When $m\Delta \ll \pi$, the spectral estimation problem requires a much stronger noise assumption regardless of which method is used, see \cite{moitra2015matrixpencil,li2021stable,batenkov2021super}. 
	\end{enumerate}
	
	Of the three optimality features of Gradient-MUSIC, the most striking one is its rate-optimality, which has important implications that will be discussed momentarily. 
	
	To our best knowledge, Gradient-MUSIC is the first provably optimal and computationally efficient algorithm for spectral estimation, under deterministic perturbations and sufficiently well-separated frequencies. Gradient-MUSIC is as good or better than any conceivable method, even those that have access to unlimited computational resources. 
	
	{\bf Benefits of oversampling.} It is already evident from previous analysis that for subspace and convex methods, increasing the number of samples enlarges the class of signals for which stable recovery is possible, see \cite{moitra2015matrixpencil,liao2016music,li2020super,fernandez2013support,azais2015spike}. This is also evident for Gradient-MUSIC because the separation condition $m\Delta \geq 8\pi$ relaxes for larger $m$. However, there is a second and perhaps more interesting benefit of oversampling. 
	
	When examining the effects of increasing the number of samples, it is natural to use the $\ell^\infty$ norm for $\bfeta$. Suppose the signal parameters $(\bfx,\bfa)$ are fixed, while $m$ can be made arbitrarily large and $\|\bfeta\|_\infty\leq \epsilon$ for sufficiently small $\epsilon$. Our main results show that  Gradient-MUSIC recovers the frequencies and amplitudes with error at most $C\epsilon/m$ and $C\epsilon$. Thus, for frequencies that are already well separated, using even more measurements beyond the bare minimum reduces the frequency error, while keeping the amplitude error bounded.

	The benefits of oversampling are also apparent in other $\ell^p$ norms as well. Of notable interest is the $\ell^2$ norm, which can be recast in terms of the standard noise-to-signal ratio (NSR). When $m\Delta\geq 8\pi$, by \cref{prop:AB19},
	$$
	\NSR:=\frac{\|\bfeta\|_2^2}{\|\bfy\|_2^2} \asymp \frac{\|\bfeta\|_2^2}{m}.
	$$
	Our main results show that if $\sqrt{\NSR}$ for a small enough universal constant $c>0$, then Gradient-MUSIC achieves 
	$$
	\text{frequency error} \leq \frac{C \sqrt{\NSR}}{m}, \andspace
	\text{amplitude error} \leq C \sqrt{\NSR}. 
	$$
	When framed this way,  it estimates the frequencies much more accurately than $\sqrt{\NSR}$ and improves in $m$ for fixed $\NSR$. 
	
	{\bf Spectral estimation even for sub-linearly growing noise variance.} For $\bfeta\sim\calN(\bfzero,\bfSigma)$ where $\bfSigma=\diag\big(\{\sigma^2(1+|k|)^{2r}\}_k\big)$ and $r\in (0,1/2)$, the main results show that Gradient-MUSIC can still succeed, which may be rather surprising. The performance of Gradient-MUSIC matches the rate for {\it nonlinear least squares} (NLS) with suitable initialization, which was derived in \cite{ying2025perturbative}. Moreover, both Gradient-MUSIC and NLS fail if $r\geq 1/2$. This comparison strongly suggests that Gradient-MUSIC is also optimal for $r\in (0,1/2)$ and the growth condition $r<1/2$ is necessary. For i.i.d. Gaussian noise $\bfeta\sim\calN(\bfzero,\sigma^2 \bfI)$, Gradient-MUSIC matches the Cram\'er-Rao lower bound in \cite{stoica1989music} up to logarithmic factors, which certifies that the former is provably optimal. It was shown in \cite{ding2024esprit} that ESPRIT is optimal for i.i.d. subgaussian noise.
	
	{\bf Global optimization with the landscape function.} We studied Gradient-MUSIC from the viewpoint of nonconvex optimization in order to eliminate the gridding artifacts of classical MUSIC. It is well known that finding suitable initialization for local optimization methods like gradient descent is challenging when the objective function is nonconvex. A key finding of this paper is that evaluating the 1D landscape function on a grid of width $1/(2m)$ is enough to find suitable initialization {\em simultaneously for all $s$ smallest local minima}. While the constant $1/2$ can be potentially improved, in \cref{rem:mesh}, we demonstrate that the landscape function needs to be evaluated on a grid whose width is no larger than $4\pi/m$, even if there is no noise and $m\Delta\geq 2\pi \beta$ for arbitrarily large $\beta$. 
	
	It is helpful to compare Gradient-MUSIC with NLS, since the latter is the generic nonconvex optimization method. The papers \cite{traonmilin2023basins,traonmilin2024strong} derived lower and upper bounds for the strong basins of attraction for NLS, but for a different setup where certain random measurements are collected. Theoretical results and numerical simulations in \cite{traonmilin2024strong} show that the strong basin of attraction to the global minimum of NLS is a ball in $\R^{2s}$ whose radius is proportional to $\Delta$. The volume of this ball relative to the measure of the parameter space is $O(\Delta^{2s})$. Thus, $\Omega(\Delta^{-2s})$ random guesses are required to find suitable initialization. Without a principled method to find initial guess, NLS is computationally intractable as $s$ becomes large.
	
	{\bf Open questions for super-resolution regime.} This paper proves that Gradient-MUSIC is optimal in the well-separated case and does not consider the super-resolution regime. The minimax lower bound for clustered frequencies in the super-resolution regime was derived in \cite{batenkov2021super}. To our best knowledge, we are not aware of any algorithms that provably achieve this optimal scaling in $m$, $\Delta$ and noise level, widely considered the three most important parameters for this setting. A detailed discussion would be too distracting to include here, so we refer the reader to \cite{batenkov2021super} and references therein for further details.

	\subsection{Literature review}
	\label{sec:review}
	
	In this section and throughout the paper, $\bfeta\in \C^{2m-1}$ denotes some additive perturbation of the clean samples. It may be deterministic or stochastic. We concentrate our discussion solely on the well-separated case $m\Delta \gg \pi$. To simplify the following discussion, the various constants $C>0$ that appear in this section do not depend on $\bfeta$ and $m$, but may depend on other parameters such as $s$ and $\bfa$. 
	
	{\bf Subspace methods.} For the well-separated case and deterministic perturbations, ESPRIT \cite{li2020super} and MPM \cite{moitra2015matrixpencil} recover the frequencies and amplitudes with error at most $C\|\bfeta\|_2/\sqrt m$ and $C\|\bfeta\|_2 \sqrt m$, respectively. Applying H\"older's inequalities, these upper bounds become $C\|\bfeta\|_\infty$ and $Cm \|\bfeta\|_\infty$, respectively. In either case, they do not achieve the minimax lower bounds. Even a sub-optimal analysis of MUSIC has been elusive. Prior work such as \cite{liao2016music} requires an implicit assumption on the convexity of the objective function near the true minima. Unlike the main results of this paper, these previous results say that one should use the smallest $m$ such that the frequencies are sufficiently well separated because increasing $m$ beyond the minimum number requires more computational resources while potentially making the amplitude error larger. 
	
	For i.i.d. subgaussian noise, it was recently shown that ESPRIT is minimax optimal in \cite{ding2024esprit}. Reference \cite{ding2024esprit} does not provide results for other types of random noise or deterministic perturbations. Since there are significant differences between MUSIC and ESPRIT, a fine-grain analysis of one method generally does not carry over to the other. A technical discussion about MUSIC versus ESPRIT can be found in \cref{rem:esprit}. 
	
	Older analysis for subspace methods are typically for the multisnapshot version where the amplitudes vary randomly and independently. For example, \cite{stoica1989music,li1993performance} derive error bounds that are asymptotic (in $m$ and the number of snapshots). While the multisnapshot problem and spectral estimation are related, there are subtle differences. For the multisnapshot problem, under i.i.d. Gaussian noise and independent stochastic amplitudes, classical MUSIC is asymptotically optimal \cite{stoica1989music}, while ESPRIT is provably sub-optimal \cite{rao2002performance}. On the other hand, for spectral estimation and i.i.d. Gaussian noise, this paper and \cite{ding2024esprit} show that (classical \& Gradient-) MUSIC and ESPRIT are optimal, respectively. 
	
	{\bf Convex optimization.} A direct comparison to our work is not easy to give since convex methods may output more than $s$ frequencies; they may create false positives, or estimate a true pair $(x_j,a_j)$ as two separate pairs. Even if these effects are assumed not to occur, \cite{fernandez2013support,azais2015spike} show that (after performing some simplifications) for an absolute constant $C>0$ and any sufficiently small $\|\bfeta\|_2$, the frequency error is at most $C \sqrt{\|\bfeta\|_2} /m$. This falls short of the lower bound of $\Omega(\|\bfeta\|_2/m^{3/2})$ established here. As for the amplitudes, assuming no true $(x_j,a_j)$ gets recovered as two separate ones, \cite{fernandez2013support,azais2015spike} imply that the amplitude error is at most $C\|\bfeta\|_2$, whereas the optimal lower bound is $\Omega(\|\bfeta\|_2/\sqrt m)$. 
	
	For i.i.d. Gaussian noise, \cite{tang2015near} showed that atomic norm minimization is optimal for denoising the measurements. This is a weaker statement than estimating $(\bfx,\bfa)$ optimally, since it can be proved that any minimax optimal method for frequency and amplitude estimation is necessarily optimal for denoising. The frequency and amplitude estimation errors for atomic norm minimization in that paper do not achieve the minimax rates for i.i.d. Gaussian noise.

	{\bf Nonconvex optimization.} We had already discussed NLS earlier. Minimization of a difference of convex objective functions was used in \cite{lou2016point}, while a sliding Frank--Wolfe algorithm was employed in \cite{denoyelle2019sliding}. Neither paper has performance guarantees when there is noise is the measurements. We are not aware of any other nonconvex optimization techniques for spectral estimation with theory that quantifies the frequency and amplitude errors in the presence of noise.

	\subsection{Organization}
	
	The remainder of this paper is organized as follows. \cref{sec:preliminaries} explains the spectral estimation problem more precisely and reviews the classical MUSIC algorithm. \cref{sec:gradMUSIC} explains the methodology and main steps of the Gradient-MUSIC algorithm. 
	
	\cref{sec:noisemodels} contains our main results for two families of deterministic and stochastic noise, which were listed in \cref{tab:mainresults}. \cref{sec:main} contains the main theory derived in this paper. We provide an abstract perspective of subspace methods, a geometric analysis of the landscape function, the core approximation theorem for Gradient-MUSIC, a sparsity detection method, and an amplitude estimation procedure. A few straightforward extensions and variations of the main theory are provided in \cref{sec:extensions}. 
	
	We compare classical and Gradient- MUSIC in \cref{sec:classicalmusic}. We show that classical MUSIC also has optimal approximation guarantees, but Gradient-MUSIC is always more computationally efficient. \cref{sec:numerical} contains computational aspects of this paper: an alternative gradient termination condition, and numerical simulations for stochastic noise. 
	
	The remaining parts of the paper include the proofs and technical details. \cref{sec:landscape} includes all of the results necessary to carry out a geometric analysis of the landscape function. This section is self contained and can be read independent of the rest of this paper, aside from global notation and definitions. \cref{sec:proofthm,sec:prooflemma} contain proofs of theorems and lemmas, respectively.  
	
	\subsection{Notation}

	The torus $\T$ is identified with either $(-\pi,\pi]$, $[0,2\pi)$, or a unit circle, depending on whatever is most convenient at the time. The distance between $u,v\in \T$ is $|u-v|=\min_{n\in\Z} |u-v+2\pi n|$. For $p\in [1,\infty]$, we use $\|\cdot\|_p$ to denote the $p$ norm of a vector and $\|\cdot\|_{L^p(\T)}$ for the $L^p$ norm of a measurable function defined on the torus. We let $\|\cdot\|_p$ denote the $\ell^p$ to $\ell^p$ operator norm of a matrix, and $\|\cdot\|_F$ be the Frobenius norm. The cardinality of a finite set $A$ is denoted $|A|$. Let $B^n_p(\epsilon)\subset\C^n$ denote the closed $\ell^p$ ball in $\C^n$ of radius $\epsilon$ centered at zero. 
	
	We say $f$ is a {\it trigonometric polynomial} of degree at most $n$ if there exist $\{c_k\}_{k=-n}^{n}\subset\C$ such that
	$$
	f(x)=\sum_{k=-n}^{n} c_k e^{ikx}.
	$$
	We say $f$ has degree $n$ if $c_n\not=0$ or $c_{-n}\not=0$. Let $\calT_m$ denote the space of trigonometric polynomials that have degree at most $m-1$. We say a trigonometric polynomial $f$ admits a {\it polynomial sum-of-squares} representation if there exist trigonometric polynomials $f_1,\dots,f_k$ such that $f=|f_1|^2+\cdots + |f_k|^2$. 
	
	As usual, we write $x\lesssim y$ (resp., $x\gtrsim y$) if there is a universal constant $C$ such that $x\leq Cy$ (resp., $x\geq Cy$). We write $x\lesssim_{a,b} y$ if there is a $C$ that potentially depends on $a,b$ such that $x\leq C y$. We write $x\asymp y$ if $x\lesssim  y$ and $x\gtrsim y$ both hold, and $\asymp_{a,b}$ is defined analogously. We generally follow the convention that $C$ and $c$ are universal constants whose values may change from one line to another, and that $C\geq 1$ and $c\leq 1$. We use standard notation $O$, $o$, $\Omega$, and $\Theta$ for asymptotic relations.

	For a vector $\bfa:=\{a_j\}_{j=1}^s\subset\C$, we denote 
	$a_{\min}:=\min_{j=1,\dots,s} |a_j|$ and $a_{\max}:=\max_{j=1,\dots,s} |a_j|.$ For any integer $n\geq 1$, define the set
	\begin{equation}
		\label{eq:defindex}
		I(n):= \left\{ -\frac{n-1}{2}, \, \frac{n+1}{2}, \dots, \frac{n-3}{2}, \, \frac{n-1}{2}\right\}.
	\end{equation}
	Note $I(n)$ consists of $n$ elements with spacing 1 and is symmetric about zero. 
	
	Let $\bbU^{m\times s}$ be the set of all $m\times s$ matrices whose columns form an orthonormal basis for its range. Generally, we will not make a distinction between $\bfU\in \bbU^{m\times s}$ and the subspace spanned by its columns. For $\bfU\in \bbU^{m\times s}$, let $\bfU_\perp\in \bbU^{m\times (m-s)}$ be a matrix whose columns form an orthonormal basis for the orthogonal complement of $\bfU$. For $\bfU, \bfV\in \bbU^{m\times s}$, the {\it sine-theta distance} is 
	\begin{equation}
		\label{eq:sintheta}
		\vartheta(\bfV,\bfW)
		:=\big\|\bfV\bfV^*-\bfW \bfW^*\big\|_2
		=\big\|\bfV_\perp^* \bfW\big\|_2. 
	\end{equation}
	For any $k\geq 0$, let $C^k(\T)$ be the space of $k$ times continuously differentiable functions defined on $\T$. We equip it with a norm,  
	$$
	\|f\|_{C^k(\T)}
	:=\sum_{j=0}^k \big\|f^{(j)}\big\|_{L^\infty(\T)}.
	$$
	
	A complex standard normal random variable $w\sim \calN(0,1)$ has independent real and complex parts that are normally distributed with mean zero and variance of $1/2$. We write $\bfw\sim \calN(\bfzero,\bfI)$ if the entries of the complex vector $\bfw$ are independent $\calN(0,1)$ random variables. We write $\bfx\sim \calN(\bfmu,\bfSigma)$ if $\bfx=\bfSigma^{1/2}\bfw+\bfmu$ for a $\bfw\in \calN(\bfzero,\bfI)$ and positive definite $\bfSigma$.

	\section{Preliminaries}
	\label{sec:preliminaries}
	
	\subsection{Spectral estimation problem formulation}
	\label{sec:problem}
	
	Consider a nonharmonic Fourier sum $h\colon\R\to\C$ with $s$ nonzero components, defined as in \eqref{eq:hfun}. In formula \eqref{eq:hfun}, the frequencies and amplitudes of $h$ are $\bfx:=\{x_j\}_{j=1}^s\subset \T:=\R/2\pi\Z$ and $\bfa:=\{a_j\}_{j=1}^s\subset\C\setminus \{0\}$, respectively. We call $s$ the {\it sparsity}. Each $a_j\not=0$ by assumption, so this function has exactly $s$ frequencies. The $h$ function in \eqref{eq:hfun} is uniquely specified by its parameters $\bfx$ and $\bfa$, which should be grouped as pairs, 
	\begin{equation}
		(\bfx,\bfa)=\{(x_j,a_j)\}_{j=1}^s.
		\label{eq:frequencyamppair}
	\end{equation}
	
	In spectral estimation, one receives the noisy measurements $\tilde \bfy$ given in \eqref{eq:noisymeasurement}, and the goal is recover the frequency and amplitude pairs in \eqref{eq:frequencyamppair}.

	\begin{definition} [Spectral Estimation Problem] \label{def:spectralestimation}
		The {\it spectral estimation problem} is to find a computationally tractable algorithm such that it receives $\tilde \bfy$ and outputs an approximation to $(\bfx,\bfa)$. 
	\end{definition}
	
	\begin{remark}
		There is a a modulation, translation, or rotational invariance. Here, the only important assumption is that samples are integer spaced and they are consecutive. The problem does not change if $\{-m+1,\dots,m-1\}$ were shifted by any real number because any shift creates a harmless phase factor in $\bfy$ that can be absorbed into $\bfa$.  
	\end{remark}
	
	\begin{remark}
		There are several variations of the spectral estimation problem, such as whether $s$ is known beforehand and/or if the amplitudes need to be estimated as well. The core of spectral estimation lies in approximating $\bfx$ assuming $s$ is known, which we refer to as {\it frequency estimation}. We will treat the other two problems of {\it sparsity detection} (estimating $s$ from just data $\tilde\bfy$) in \cref{sec:sparsity} and {\it amplitude estimation} (approximation of $\bfa$) in \cref{sec:amplitudes} separately from the core frequency estimation problem.
	\end{remark}
	
	Spectral estimation exhibits a permutation invariance, whereby the function $h$ is invariant under the transformation $\{(x_j,a_j)\}_{j=1}^s\mapsto \{(x_{\pi(j)},a_{\pi(j)})\}_{j=1}^s$ for any permutation $\pi$ on $\{1,\dots,s\}$. For this reason, given a pair $(\hat\bfx,\hat\bfa)$, where $\bfx$ and $\hat\bfx$ have the same cardinality, it is natural to define the frequency error as the matching distance,
	$$
	\min_{\pi} \max_j |x_j-\hat x_{\pi(j)}|. 
	$$
	Then for the permutation $\pi$ that attains this minimum, the amplitude error is defined as
	$$
	\max_j |a_j-\hat a_{\pi(j)}|. 
	$$
	It is customary to just assume $(\hat\bfx,\hat\bfa)$ has been indexed so that $\bfx$ and $\hat\bfx$ best match, so the optimal permutation is the identity. This is the convention that we use throughout this paper. 
	
	For any integer $n\geq 1$, the nonharmonic {\it Fourier matrix} with nodes $\bfx\subset\T$ is defined as 
	$$
	\bfPhi(n,\bfx) 
	:= \left[ e^{ik x_j} \right]_{k\in I(n),\, j=1,\dots,s} \in \C^{n\times s}.
	$$
	The range of $\bfPhi(n,\bfx)$ and its singular values are invariant under any shift of the index set $I(n)$ defined in \eqref{eq:defindex}. Spectral estimation is a nonlinear inverse problem because the forward map $(\bfx,\bfa)\mapsto \bfy$ is nonlinear. To understand the general stability of spectral estimation, it is prudent to analyze the forward map. Define the minimum separation of $\bfx$ by 
	$$
	\Delta:=\Delta(\bfx)=\min_{j\not=k} \min_{n\in\Z} \, |x_j-x_k+2\pi n|. 
	$$
	Recall the following result that controls its extreme singular values. 
	
	\begin{proposition}[Aubel-B\"olcskei \cite{aubel2019vandermonde}]
		\label{prop:AB19}
		For any $m\geq 1$, $\beta>1$, and $\bfx\subset\T$ such that $\Delta(\bfx)\geq 2\pi \beta/m$, it holds that 
		$$
		\sqrt{m(1-\beta^{-1})} 
		\leq \sigma_{\min}(\bfPhi(m,\bfx)) 
		\leq \sigma_{\max}(\bfPhi(m,\bfx))
		\leq \sqrt{m(1+\beta^{-1})}.
		$$
	\end{proposition}
	
	When $m\Delta\gg 2\pi$, it means that the forward map is well conditioned. When $m\Delta\ll 2\pi$, the condition number of $\bfPhi(m,\bfx)$ behaves radically different. Since our theory does not consider this case, we do not give a description of what happens and instead refer the reader to \cite{barnett2022exponentially,li2024multiscale} and references therein.  
	
	\subsection{Review of the classical MUSIC algorithm}
	\label{sec:music}
	
	Subspace methods were developed in the late 1980's, and some popular methods include MUSIC \cite{schmidt1986multiple} and ESPRIT \cite{kailath1989esprit}. They were originally invented for direction-of-arrival (DoA) estimation. The original versions assume that the signal amplitudes vary over time and one obtains multiple snapshots of data. The spectral estimation problem can be viewed as a single-snapshot version. 
	
	Of course, one obtains considerably more information in the multisnapshot version than in spectral estimation. The original MUSIC and ESPRIT algorithms can be used for spectral estimation as well, albeit with some modifications. This method of modification uses a Hankel matrix, which is also used in Prony's method \cite{Prony} and MPM \cite{hua1990matrixpencil}. While there are some similarities between the single and multiple snapshot problems, they are considerably different, see \cite{li2022stability} for further details. 
	
	To be clear, in this paper, {\it classical MUSIC} refers to the single-snapshot MUSIC algorithm in \cite{liao2016music}, which is faithful to the original multiple-snapshot MUSIC algorithm created in \cite{schmidt1986multiple} for the multisnapshot version.

	This section briefly reviews classical MUSIC and also some mechanism behind subspace methods. There are three major steps: identification of $s$ and approximation of the signal subspace, approximation of the frequencies $\bfx$ with $s$ given, and estimation of the amplitudes $\bfa$ given approximate frequencies. The second procedure is the core step, and is the focus of our discussion.
	
	{\bf Noiseless setting} $\bfeta=\bfzero$. Recall we indexed $\bfy$ by $\{-m+1,\dots,m-1\}$. The square {\it Toeplitz matrix} of $\bfy\in \C^{2m-1}$ is defined as 
	$$
	\bfT
	:=T(\bfy)
	=\begin{bmatrix}
		y_0 & y_{-1} &\cdots & y_{-m+1}\\
		y_1 & y_0 & & y_{-m} \\
		\vdots & &\ddots &\vdots \\
		y_{m-1} & y_{m-2} &\cdots &y_0 
	\end{bmatrix}\in \C^{m\times m}.
	$$
	More generally, this defines a linear operator $T\colon \C^{2m-1}\to \C^{m\times m}$ which maps an arbitrary vector to its associated Toeplitz matrix. Whenever $m\geq s$, we have the factorization 
	\begin{equation}
		\label{eq:toeplitzfactorization}
		T(\bfy) = \bfPhi(m,\bfx) \bfA \bfPhi(m,\bfx)^*, \wherespace \bfA:=\diag(\bfa). 
	\end{equation}
	If additionally $m>s$, the Toeplitz $\bfT$ has rank exactly equal to $s$. Hence, it is rank deficient and its leading $s$ dimensional left singular space is also $\bfPhi(m,\bfx)$. Its range $\bfU$ can be computed through the singular value decomposition of $\bfT$. It is common to refer to $\bfU$ as the ``signal subspace."
	
	\begin{remark}
		In many references that use subspace methods, a Hankel matrix is used instead of Toeplitz. There is no practical difference since they have the same left singular space and identity \eqref{eq:toeplitzfactorization} holds if $T(\bfy)$ and the adjoint $^*$ are replaced with a Hankel operator of $\bfy$ and transpose $^T$, respectively. Some references also work with right singular spaces instead, which again, only results in cosmetic changes.   
	\end{remark}
	
	To compute $\bfx$ from $\bfU$, the classical MUSIC algorithm forms a certain function $q=q_\bfU$, in the following way.   
	
	\begin{definition}[Steering vector]
		\label{def:steering}
		The {\it (normalized) steering vector} is a vector-valued function $\bfphi\colon \T\to \C^m$ defined as
		$$
		\bfphi(t)=\frac 1 {\sqrt m}\left[ e^{i k t} \right]_{k\in I(m)}.
		$$ 
	\end{definition}
	
	\begin{definition}[Landscape function]
		\label{def:nscorrelation}
		For any $m>s$, the {\it landscape function} $q_{\bfW}\colon\T\to \R$ associated with a subspace $\bfW\subset \C^m$ of dimension $s$ is defined as 
		\begin{equation*}
			q_{\bfW}(t)
			:= 1 -\|\bfW^*\bfphi(t)\|_2^2. 
		\end{equation*}
		When $\bfW$ is understood from context, we simply write $q$ instead of $q_{\bfW}$. 
	\end{definition}
	
	The following has been established in prior analysis of MUSIC, see \cite{schmidt1986multiple,liao2016music}. 
	
	\begin{proposition}
		\label{prop:qroots}
		Let $m>s$ and $q:=q_{\bfU}$ be the landscape function associated with $\bfU$, where $\bfU$ is an orthonormal basis for the range of $\bfPhi(m,\bfx)$. We have $q(t)=0$ if and only if $\bfphi(t)\in \bfU$ if and only if $t\in \bfx$.
	\end{proposition}

	Thus, whenever there is no noise, $\bfU$ can be computed from $T(\bfy)$ and the unknown frequencies $\bfx$ can be determined by finding the all zeros of $q$. When $m$ is large, this is not necessarily an easy task, as it requires finding all zeros of $q$, whose degree is potentially much larger than $s$. Nonetheless, this discussion illustrates that the noiseless problem can be reduced to a standard numerical task. 
	
	\smallskip 
	{\bf Noisy case} $\bfeta\not=\bfzero$. The previous strategy has to be significantly modified because we cannot directly compute $\bfU$ anymore. Subspace methods first form a Toeplitz matrix from noisy data,
	\begin{equation} \label{eq:Htilde}
		\tilde \bfT
		:= T(\tilde \bfy)
		\in \C^{m\times m}.
	\end{equation}
	Since $T$ is a linear operator and $\tilde \bfy = \bfy+\bfeta$, we have 
	\begin{equation}
		\label{eq:H1}
		\tilde \bfT 
		:= T(\tilde\bfy)
		=T(\bfy) + T(\bfeta)
		= \bfT + T(\bfeta).
	\end{equation}
	The Toeplitz matrix $\tilde \bfT$ is a perturbation of $\bfT$. Subspace methods compute the leading $s$ dimensional left singular space of $\tilde \bfT$ to approximate $\bfU$.
	
	\begin{definition}[Toeplitz Estimator]
		\label{def:Toeplitzestimator}
		If $m>s$ and $\sigma_s(\tilde \bfT)\not=\sigma_{s+1}(\tilde \bfT)$, then the {\it Toeplitz estimator} is the unique (up to trivial ambiguities) orthonormal basis $\tilde \bfU$ for the leading $s$ dimensional left singular space of $\tilde\bfT$. 
	\end{definition}
	
	It is common to call $\tilde \bfU$ an ``empirical signal subspace." We refrain from using this term as it incorrectly suggest that the Toeplitz estimator is the only way of approximating $\bfU$. The error between $\bfU$ and $\tilde\bfU$ is naturally quantified by the sine-theta distance $\vartheta(\bfU,\tilde\bfU)$ defined in \eqref{eq:sintheta}.
	
	The classical MUSIC algorithm evaluates the landscape function $\tilde q:=q_{\tilde \bfU}$ associated with the Toeplitz estimator $\tilde\bfU$, on a finite set $G\subset\T$ and finds the $s$ smallest (discrete) local minima of $\tilde q$ on $G$. Here, $G$ can be nonuniform, but it needs to be dense enough so that the discrete minima closely approximate the continuous ones. To quantify this, we define the {\it mesh norm} of $G$ as
	\begin{equation}
		\label{eq:meshnorm}
		\mesh(G):= \max_{t\in \T}\min_{u\in G} |t-u|. 
	\end{equation}	
	A smaller mesh norm corresponds to a denser set. By making $\mesh(G)$ small, it reduces numerical error of classical MUSIC, but increases its computational complexity. 
	
	Classical MUSIC requires solving a nonconvex optimization problem where $\tilde q$ is the objective function. There is an important difference between this setup and a standard optimization approach. Since $\bfx$ has cardinality $s$, a typical optimization approach (such as maximum likelihood optimization and nonlinear least squares) would create an objective function of $2s$ variables and look for its global minimum. In contrast, classical MUSIC creates a single variable function $\tilde q$ and finds its $s$ smallest minima. This can be viewed as a dimension reduction step. This is one of several reasons why we defined $q_\bfW$ as the ``landscape function", as opposed to the more generic term ``objective function". A mathematical perspective of $q_\bfW$ is provided in \cref{sec:abstract}.  
	
	\begin{algorithm}
		\caption{Classical MUSIC}
		\begin{algorithmic}
			\Require Noisy Fourier data $\tilde \bfy\in \C^{2m-1}$. \\
			{\bf Parameters:} Number of frequencies $s$, and finite subset $G\subset\T$.
			\begin{enumerate}
				\item 
				Find an orthonormal basis $\tilde \bfU$ for the leading $s$ dimensional left singular space of $\tilde\bfT$. 
				\item 
				Evaluate $\tilde q$ on a grid $G$ and find its $s$ smallest discrete local minima of $\tilde q$ on $G$, denoted $\hat\bfx$. 
			\end{enumerate}
			\Ensure Estimated frequencies $\hat \bfx$. 
		\end{algorithmic}
		\label{alg:MUSIC}
	\end{algorithm}
	
	The classical MUSIC algorithm has several well known theoretical and computational issues. 
	\begin{enumerate}[(a)]\itemsep-2pt
		\item 
		The number of frequencies $s$ must be known, otherwise the computed $\tilde \bfU$ and $\tilde q$ are incorrect. Then the recovered $\hat\bfx$ can be completely different from what is expected. 
		\item 
		Even if the Toeplitz estimator $\tilde \bfU$ is correctly found, it is possible that $\tilde q$ has strictly less than $s$ local minima, which makes the entire procedure undefined. 
		\item 
		Even if $\tilde q$ has at least $s$ many local minima, it is not clear why selecting the ones that minimize the value of $\tilde q$ is preferred over other local minima. 
		\item 
		Despite classical MUSIC being formulated as a grid-free algorithm, one must evaluate $\tilde q$ on a grid $G$ with a small mesh norm $\mesh(G)$ to approximate its local minima. 
	\end{enumerate}

	\section{Gradient-MUSIC algorithm}
	\label{sec:gradMUSIC}
	
	\begin{algorithm}[t]
		\caption{Gradient-MUSIC (given a subspace and sparsity)}
		\begin{algorithmic}
			\Require Subspace $\tilde \bfU$ of dimension $s$. \\
			{\bf Parameters:} Finite set $G\subset\T$, threshold parameter $\alpha$, gradient step size $h$, and number of iterations $n$.
			\begin{enumerate} 
				\item {\bf Step I (Initialization on a Coarse Grid).} 
				Evaluate the landscape function $\tilde q:=q_{\tilde\bfU}$ associated with the subspace $\tilde\bfU$ on the set $G\subset\T$ and find the accepted points
				$$
				A:=A(\alpha)=\{u\in G\colon \tilde q(u)< \alpha \}. 
				$$ 
				Find all $s$ clusters of $A$ and pick representatives $t_{1,0},\dots,t_{s,0}$ for each cluster. 
				\item {\bf Step II (Gradient Descent for Local Optimization).} 
				Run $n$ iterations of gradient descent with step size $h$ and initial guess $t_{j,0}$, namely
				$$
				t_{j,k+1}=t_{j,k}-h \tilde q \, '(t_{j,k}) \foreachspace k = 0,\dots,n-1.
				$$ 
				Let $\hat x_j=t_{j,n}$ and $\hat \bfx=\{\hat x_j\}_{j=1}^s$. 
			\end{enumerate}
			\Ensure Estimated frequencies $\hat \bfx$.
		\end{algorithmic}
		\label{alg:gradMUSIC}
	\end{algorithm}
	
	Motivated by the limitations of classical MUSIC, we propose the Gradient-MUSIC algorithm, based on a nonconvex optimization reformulation of MUSIC. In this section, we numerically demonstrate the geometric shape of the landscape function in Figure \ref{fig:landscape}, and leave the theoretical analysis in \cref{subsec:landscapeanalysis}.
	
	The core problem in spectral estimation is to estimate the unknown frequencies $\bfx$ given $s$. In this section, we fix a particular $\bfx$ and $m$ such that $\Delta(\bfx)\geq 8\pi/m$ and let $\bfU=\range(\bfPhi(m,\bfx))$ be the noise-free signal subspace. Figure \ref{fig:landscape} shows a particular landscape function $q=q_\bfU$. Let $\tilde\bfU$ be the Toeplitz estimator of the signal subspace $\bfU$ from the noisy measurements $\tilde\bfy=\bfy+\bfeta$, where $\bfeta\sim \calN(0,\sigma^2\bfI)$ is one realization with variance $\sigma^2=1$. Let $\tilde q=q_{\tilde\bfU}$ be the landscape function associated with the Toeplitz estimator $\tilde\bfU$. Figure \ref{fig:landscape} also plots $\tilde q$. We see that $q$ and $\tilde q$ are similar, but the error appears to be largest near $\bfx$. 
	
	There are subtleties that we will expand on in \cref{sec:abstract}. For now, it is important to recognize that $\bfU=\range(\bfPhi(m,\bfx))$ for a fixed $\bfx$, so it is not an arbitrary subspace in $\C^m$ and it has some special properties that will discussed later. Also, $\tilde \bfU$ is some perturbation of $\bfU$ chosen through a particular process. While we can think of $\tilde q$ as a perturbation of $q$, it is a particular implicit perturbation of $q$ that ``factors" (in an algebraic sense) through this particular perturbation of the signal subspace.  
	
	\begin{figure}[ht]
		\centering
		\begin{subfigure}{0.45\linewidth}
			\includegraphics[width=\linewidth]{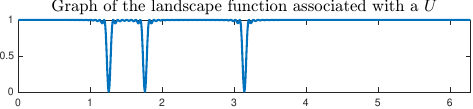}
			\includegraphics[width=\linewidth]{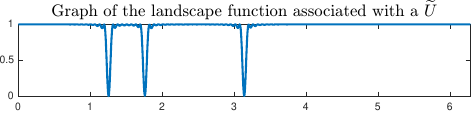}
			\includegraphics[width=\linewidth]{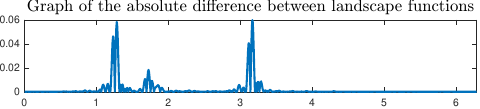}
			\caption{Plots of $q$, $\tilde q$, and $|q-\tilde q|$. }
		\end{subfigure}
		\begin{subfigure}{0.45\linewidth}
			\includegraphics[width=0.85\linewidth]{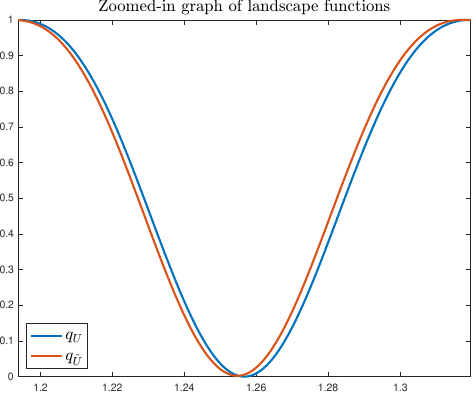}
			\caption{Zoomed in plot.}
		\end{subfigure}
		\caption{Plot of the landscape function $q$ generated from $\bfx=\{2\pi(2/m),\, 2\pi(10/m),\, \pi\}$ and the landscape function $\tilde q$ associated with $\tilde\bfU$. }
		\label{fig:landscape}
	\end{figure}
	
	Figure \ref{fig:landscape} (b) shows the landscape function $\tilde q$ near a $x_j\in \bfx$. It appears that the $s$ smallest local minima $\tilde\bfx=\{\tilde x_j\}_{j=1}^s$ of $\tilde q$ are not far away from $\bfx=\{x_j\}_{j=1}^s$. Figure \ref{fig:landscape} illustrates why the classical MUSIC algorithm seeks to find the smallest $s$ local minima of $\tilde q$. On the other hand, Figure \ref{fig:landscape} suggests an alternative method to search for $\tilde\bfx$ without evaluation on a fine grid.

	We make two key observations of the landscape function: (1) The landscape function $\tilde q$ is large on the region far from $\tilde\bfx$, shown in Figure \ref{fig:landscape} (a). (2) The landscape function $\tilde q$ appears locally convex in a neighborhood of each local minimum $\tilde x_j$, shown in Figure \ref{fig:landscape} (b).
	
	Our observations of the landscape function, which is supported by a theoretical analysis in Section \ref{subsec:landscapeanalysis}, lead to the core methodology of Gradient-MUSIC, which is summarized in \cref{alg:gradMUSIC}. It is a two-stage process in which a coarse grid is laid out to find a suitable initialization in each basin of attraction to $\tilde x_j$ and then a gradient descent is performed with this initialization. 
	
	{\bf Step I (Initialization on a Coarse Grid).} Figure \ref{fig:landscape} suggests that $\tilde q(t)$ is large when $t$ is far from $\tilde\bfx$, so by sampling $\tilde q$ and thresholding, we can find a set of suitable initialization. In view of this discussion, we define the following. 
	
	\begin{definition}[Accepted elements]
		For any finite set $G\subset\T$ and $\alpha>0$, we say that a $u\in G$ is {\it rejected} if $\tilde q(u)\geq \alpha$ and is {\it accepted} if $\tilde q(u)<\alpha$. Let $A:=A(\alpha)\subset G$ denote the set of accepted elements. \label{def:accepted}
	\end{definition}
	
	\begin{figure}[t]
		\centering
		\begin{subfigure}{0.7\linewidth}
			\includegraphics[width=\linewidth]{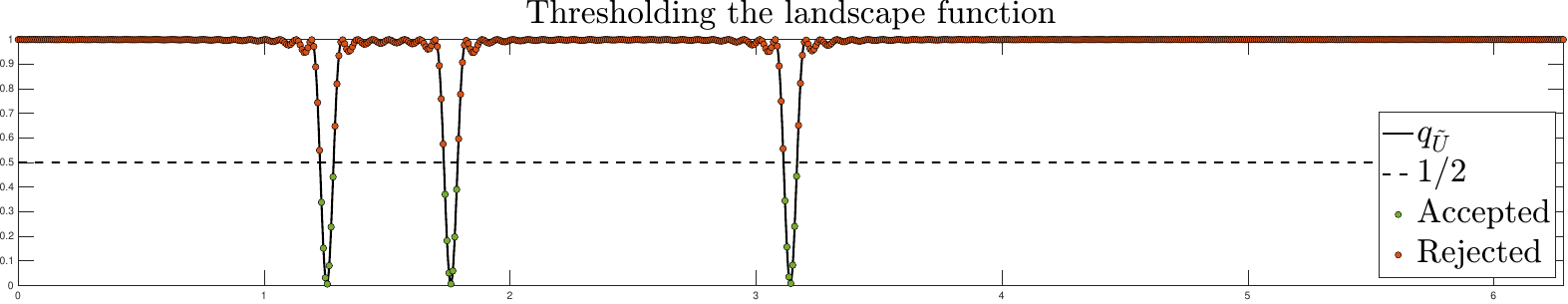}
			\caption{All accepted and rejected elements.}
		\end{subfigure}\\ \medskip 
		\begin{subfigure}{0.4\linewidth}
			\includegraphics[width=\linewidth]{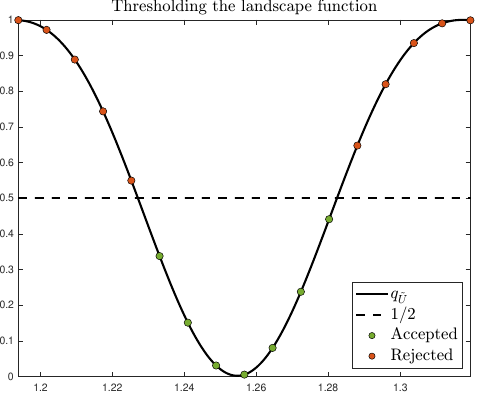}
			\caption{Zoomed in plot.}
		\end{subfigure} 
		\caption{Plot of $\tilde q$ and the set of $(u,\tilde q(u))$, where $u$ are elements of a uniform grid of width $2\pi/(8m)$. The accepted and rejected $(u,q(u))$ are shown in green and red, respectively.}
		\label{fig:landscape2}
	\end{figure}

	Figure \ref{fig:landscape2} shows an example of $\tilde q$ sampled on a uniform grid $G$ of width $2\pi/(8m)$ and the set of accepted and rejected points with parameter $\alpha = 1/2$. Examining Figure \ref{fig:landscape2} (b), it appears that the accepted points lie in a basin of attraction where gradient descent initialized at any of those points would converge to an element in $\tilde\bfx$. 
	
	Importantly, the grid selected in Figure \ref{fig:landscape2} does not need to be very thin; it is thin enough enough to locate at least one suitable initialization per $\tilde x_j$, but is not thin enough that its discrete local minima will always be a high quality approximation of $\tilde x_j$. This is important for computations because evaluation of $\tilde q(t)$ at just a single $t\in \T$ is not cheap and it requires $O(ms)$ operations due to the matrix-vector product in \cref{def:nscorrelation}. 
	
	Examining Figure \ref{fig:landscape2} (a), it appears that there are exactly $s$ many connected components in the set of accepted elements. This is made precise in the following definition.
	
	\begin{definition}[Cluster]\label{def:cluster}
		Given $v,w\in G$ such that $|u-v|<\pi$, the {\it path} in $G$ that connects them is the set of all elements $u\in G$ that lie in the (shorter) arc between $v$ and $w$. We say a subset $V \subset G$ is a {\it cluster} in $G$ if for any $v,w\in V$, the path in $G$ connecting $v$ and $w$ is also contained in $V$. An arbitrary element of a cluster is called a {\it representative}.
	\end{definition}
	
	{\bf Step II (Gradient Descent for Local Optimization).} For the $j$-th cluster, pick an arbitrary representative $t_{j,0}$. Using this as the initial guess, any appropriate local optimization technique will converge to $\tilde x_j$.  {\it Gradient descent} with step size $h$ and initial guess $t_{j,0}$ produces the iterates $\{t_{j,k}\}_{k=0}^\infty$ which are defined recursively as 
	\begin{equation}
		\label{eq:gradient}
		t_{j,k+1}=t_{j,k}-h\tilde q \,'(t_{j,k}) \forallspace k \in \N.
	\end{equation}
	
	The effectiveness of the algorithm hinges on the properties of the landscape function $\tilde q$. A complete geometric analysis of $\tilde q$ is carried out in \cref{thm:landscape}. In fact, this theorem is far more general than the discussion carried out in this section. It applies to any $\bfx$ for which $\Delta(\bfx) \geq 8\pi /m$ and $\tilde\bfU$ that is close enough to $\bfU$, not just the Toeplitz estimator for $\bfU$. For this reason, the $\tilde\bfU$ in Gradient-MUSIC (\cref{alg:gradMUSIC}) is an arbitrary subspace. The performance guarantees of this algorithm are summarized in \cref{thm:maingradientMUSIC}.

	\section{Main results: optimality of Gradient-MUSIC for various noise models} 
	\label{sec:noisemodels}
	
	As explained earlier, there are three main steps in the Gradient-MUSIC algorithm, which we handle in separate parts of this paper.
	\begin{enumerate} \itemsep-2pt
		\item 
		Estimation of the number of frequencies $s$ and initial approximation $\tilde\bfU$ of $\bfU$ from just data $\tilde\bfy$. This is carried out in \cref{sec:sparsity}.
		\item 
		Estimation of the frequencies $\bfx$ given a $\tilde\bfU$ that is close to the true subspace $\bfU$. This is the main step and involves the most analysis, which is handled in \cref{sec:gradientMUSIC}.
		\item 
		Estimation of the amplitudes when a reasonable approximation of $\bfx$ has been found. This is carried out in \cref{sec:amplitudes}.
	\end{enumerate}
	By incorporating those results, we provide a full pipeline version of Gradient-MUSIC in \cref{alg:gradMUSIC2}. The algorithm outputs estimated parameters $(\hat\bfx,\hat\bfa)$ given noisy data $\tilde\bfy$. It requires specification of several parameters, which we give explicit formulas and guidelines for. To explain our main theorems, we introduce the following definition. 
	\begin{definition}
		\label{def:signalmodel}
		For any $c>0$ and $0<r_0\leq r_1$, define the {\it signal class}
		$$
		\calS(c,r_0,r_1) 
		= \left\{ (\bfx,\bfa) \colon \Delta(\bfx) \geq c, \ r_0\leq |a_j|\leq r_1,  \forallspace j=1,\dots,|\bfx|\right\}. 
		$$
	\end{definition}

	\begin{algorithm}[t]
		\caption{Gradient-MUSIC (full pipeline)}
		\begin{algorithmic}
			\Require Noisy Fourier data $\tilde \bfy \in \C^{2m-1}$.\\
			{\bf Parameters:} Threshold parameters $\gamma,\alpha\in [0,1]$, finite set $G\subset\T$, gradient step size $h$, and number of iterations $n$.
			\begin{enumerate}
				\item 
				Form the noisy Toeplitz matrix $\tilde \bfT=T(\tilde\bfy)$ and compute 
				$$
				\hat s:=\max\big\{k \colon {\sigma_k(\tilde \bfT)} \geq \gamma {\sigma_1(\tilde \bfT)} \big\}.
				$$ 
				Find an orthonormal basis $\tilde \bfU$ for the leading $\hat s$ dimensional left singular space of $\tilde\bfT$.  
				\item 
				Use \cref{alg:gradMUSIC} applied to $\tilde\bfU$ to output estimated frequencies $\hat\bfx$.
				\item 
				Let $\hat \bfPhi:=\bfPhi(m,\hat \bfx)$ and compute
				$$
				\hat \bfa := \diag\Big( \hat \bfPhi^+ \tilde \bfT \big(\hat \bfPhi^+ \big)^*\Big)=\{\hat a_j\}_{j=1}^{\hat s}.
				$$
			\end{enumerate}
			\Ensure Estimated parameters $(\hat \bfx,\hat\bfa)$.
		\end{algorithmic}
		\label{alg:gradMUSIC2}
	\end{algorithm}

	\subsection{Deterministic perturbations}
	
	We first provide performance guarantees for \cref{alg:gradMUSIC2} under deterministic perturbations measured in the $\ell^p$ norm. 
	
	\begin{theorem}[Deterministic $\ell^p$ perturbations]
		\label{cor:deterministic}
		Let $m\geq 100$, $r_0>0$, and $p\in [1,\infty]$. Suppose $(\bfx,\bfa)\in \calS\left({8\pi}/m, r_0, 10 r_0 \right)$ and $\bfeta\in \C^{2m-1}$ such that 
		\begin{equation}
			\label{eq:nsrp}
			\frac{\|\bfeta\|_p}{a_{\min} m^{1/p}} \leq \frac{3}{1600}.
		\end{equation}
		Let $(\hat\bfx,\hat\bfa)$ be the output of Gradient-MUSIC (\cref{alg:gradMUSIC2}) with $\gamma = 0.0525$, $\alpha=0.529$, $h=6/m^2$, finite $G\subset\T$ such that $\mesh(G)\leq 1/(2m)$, and 
		\begin{equation}
			\label{eq:niter}
			n\geq \max\left\{ 31, \, 6 \log\left(\frac {15 \, a_{\min} \, m^{1/p}}{\|\bfeta\|_p}\right) \right\}.
		\end{equation}
		Then we have
		\begin{align*}
			\max_{j} |x_j-\hat x_j|
			&\leq 55 \, \frac{\|\bfeta\|_p}{a_{\min}\, m^{1+1/p}}, \\ 
			\max_{j} |a-\hat a_j|
			&\leq 455 \, \frac{a_{\max}}{a_{\min}} \frac{\|\bfeta\|_p}{m^{1/p}}. 
		\end{align*}   
	\end{theorem}

	\begin{remark}[Optimality of noise condition]
		\label{rem:impossible}
		We give a simple example which illustrates why the noise condition \eqref{eq:nsrp} is necessary, up to universal constants. If $(\bfx,\bfa)=(0,r_0)$ and $\eta_k=-r_0$ for each $k=-m+1,\dots,m-1$, then $\|\bfeta\|_p=r_0 (2m-1)^{1/p}$ while $\tilde\bfy=\bfzero$. This shows that spectral estimation is impossible unless $\|\bfeta\|_p/(r_0 m^{1/p})$ is sufficiently small.
	\end{remark}
	
	\cref{cor:deterministic} is proved in \cref{sec:deterministic}. Although the $p\in (1,\infty]$ statements in \cref{cor:deterministic} are immediate consequences of the $p=1$ statement and H\"older's inequality, we stated this theorem for all $p\in [1,\infty]$ since each $p$ has a different interpretation. In the subsequent discussion, we ignore all quantities aside from $m$ and $\bfeta$. Dependence on $a_{\min}$ can be effectively ignored by rescaling the main model equation \eqref{eq:noisymeasurement}, so that $a_{\min}=1$ and $\bfeta/a_{\min}$ is redefined to be $\bfeta$. Additionally, we ignore $a_{\max}/a_{\min}$ since by assumption, $a_{\max}/a_{\min}\leq 10$. 
	
	{\bf Small bounded perturbations}. Suppose $\bfeta$ is deterministic and quantified in the $\ell^\infty$ norm. This is arguably the most natural norm if we fix $(\bfx,\bfa)$, let $m$ increase, and assume the entries of $\bfeta$ are uniformly bounded. We see that assumption \eqref{eq:nsrp} is satisfied whenever $\|\bfeta\|_\infty$ is a small enough universal constant. In this case, the frequency and amplitude errors are $C\|\bfeta\|_\infty/m$ and $C\|\bfeta\|_\infty$ respectively.
	
	{\bf High energy perturbations}. Suppose $\bfeta$ is deterministic and quantified in the $\ell^2$ norm. This is an energy norm, and this assumption is perhaps natural for physical applications. Condition \eqref{eq:nsrp} is satisfied whenever $\|\bfeta\|_2/\sqrt m$ is sufficiently small. If $\|\bfeta\|_2$ is uniformly bound in $m$, then this is a less stringent constraint compared to the one for bounded noise. In this case, the frequency and amplitude errors are $C\|\bfeta\|_2/ m^{3/2}$ and $C\|\bfeta\|_2/\sqrt m$ respectively. 
	
	{\bf Sparse perturbations}. Suppose $\bfeta$ is deterministic and quantified in the $\ell^1$ norm. We refer to it as a sparse perturbation since $\ell^1$ is a surrogate norm for sparsity. Condition \eqref{eq:nsrp} is satisfied whenever $\|\bfeta\|_1/m$ is sufficiently small. If $\|\bfeta\|_1$ is uniformly bounded in $m$, then this condition is even easier to satisfy than previous ones. The frequency and amplitude errors are $C\|\bfeta\|_1/m^2$ and $C\|\bfeta\|_1/m$ respectively. This indicates that both MUSIC algorithms take advantage of sparse perturbations.

	\subsection{Stochastic noise} 
	
	Now we move onto several examples of stochastic noise. For simplicity, we consider a natural example where $\bfeta\sim \calN(\bfzero,\bfSigma)$ is normally distributed with mean zero and diagonal covariance matrix $\bfSigma$. The following theorem is proved in \cref{sec:stochastic}. 
	
	\begin{theorem}[Gaussian noise with diagonal covariance]
		\label{cor:stochastic}
		Let $m\geq 100$, $r_0>0$, and $\bfSigma\in \R^{(2m-1)\times (2m-1)}$ be a diagonal matrix with positive diagonals. There are positive universal constants $C_1$, $C_2$, $C_3$, and $c$ such that for any $(\bfx,\bfa)\in \calS\left({8\pi}/m, r_0, 10r_0 \right)$, noise $\bfeta\sim \calN(\bfzero,\bfSigma)$, and $t>1$, the following hold with probability at least 
		\begin{equation}
			\label{eq:nsrsampling}
			1-2m^{1-t^2}-2m \exp\left( -\frac{ca_{\min}^2  m^2}{\tr(\bfSigma)} \right).
		\end{equation}
		Let $(\hat\bfx,\hat\bfa)$ be the output of Gradient-MUSIC (\cref{alg:gradMUSIC2}) with $\gamma = 0.0525$, $\alpha=0.529$, $h=6/m^2$, any $G\subset \T$ such that $\mesh(G)\leq 1/(2m)$, and any
		\begin{equation}
			\label{eq:niter2}
			n\geq \max\left\{ 31, \, C_1\log\left( \frac{a_{\min} \, m} {t \sqrt{\tr(\bfSigma)\log(m)}} \right) \right\}.
		\end{equation} 
		Then we have  
		\begin{align*}
			\max_{j} |x_j-\hat x_j|
			&\leq \frac{C_2}{a_{\min}} \frac{t \sqrt{\tr(\bfSigma)\log(m)}}{m^2}, \\ 
			\max_{j} |a-\hat a_j|
			&\leq C_3 \, \frac{ a_{\max}}{a_{\min}} \frac{t \sqrt{\tr(\bfSigma)\log(m)}}{m}. 
		\end{align*}
	\end{theorem}
	
	Let us examine the content of \cref{cor:stochastic} for a parametric family of examples. Suppose $\Sigma_{k,k}=\sigma^2 (1+|k|)^{2r}$ for some $\sigma>0$ and $r \in \R$. This encapsulates a family of nonstationary Gaussian noise: decaying ($r <0$), stationary ($r =0$), and growing ($r>0$). For some $C_r>0$ that depends only on $r$, we have
	\begin{equation}\label{eq:Sigmacolor}
		\tr(\bfSigma)
		\leq 
		\begin{cases}
			C_r \sigma^2 m^{2r+1} &\text{if }r >-1/2, \\
			C_r \sigma^2 \log(m) &\text{if }r = -1/2, \\
			C_r &\text{if }r < -1/2.
		\end{cases}
	\end{equation}
	For this family of diagonal covariance matrices, we see that for fixed $\sigma$, condition \eqref{eq:nsrsampling} is at least $1-o(m)$ as $m\to\infty$ if and only if $r \in (-\infty,1/2)$. In the subsequent discussion, we ignore all quantities aside from $m$, $r$ and $\sigma$. For the same reasons as the deterministic case, we can ignore dependence on $a_{\min}$ and $a_{\max}/a_{\min}$. 
	
	{\bf Stationary}. Suppose $r=0$ and $\sigma>0$ is arbitrary. The variance can be arbitrarily large, provided that $m$ is large enough to make condition \eqref{eq:nsrsampling} non-vacuous. The frequency and amplitude errors are $C\sigma m^{-3/2} \sqrt{\log(m)}$ and $C\sigma m^{-1/2} \sqrt{\log(m)}$. Aside from $\log(m)$ factors, both rates are optimal in view of the classical Cram\'er-Rao lower bound, see \cite{stoica1989music}. However, this example does not show that Gradient-MUSIC is maximally noise stable because the spectral estimation problem is still solvable for growing noise.
	
	\begin{remark}[Relaxation to i.i.d. subgaussian]\label{rem:subgaussian}
		For the $r=0$ case, the assumption that $\bfeta\sim \calN(\bfzero,\sigma^2 \bfI)$ is not essential and $\bfeta$ can be assumped to be i.i.d. subgaussian with mean zero and variance $\sigma^2$. Indeed, \cite[Theorem 1]{meckes2007spectral} showed that the Toeplitz matrix of subgaussian $\bfeta$ has expected spectral norm $C\sigma \sqrt{m\log(m)}$ where $C$ also depends on the subgaussian constant. Using this result in the proof of \cref{cor:stochastic}, we would get that the average frequency and amplitude errors are $C\sigma m^{-3/2} \sqrt{\log(m)}$ and $C\sigma m^{-1/2} \sqrt{\log(m)}$, identical to the case for $\bfeta\sim \calN(\bfzero,\sigma^2\bfI)$.
	\end{remark}
	
	{\bf Increasing}. Suppose $r \in (0,1/2)$ and $\sigma>0$. The frequency and amplitude errors are $C\sigma m^{r-3/2} \sqrt{\log(m)}$ and $C\sigma m^{r-1/2} \sqrt{\log(m)}$. This example is perhaps surprising because it shows that Gradient-MUSIC can tolerate noise that grows in strength as more samples are collected. Up to logarithmic factors, both Gradient-MUSIC and nonlinear least squares (see \cite{ying2025perturbative} for an analysis) attain the same rate for $r\in [0,1/2)$ and both fail when $r\geq 1/2$. Since NLS is believed to be optimal, this strongly suggests that Gradient-MUSIC is also optimal and that no method can accurately estimate both the frequencies and amplitudes whenever $r \geq 1/2$. 
	
	{\bf Decreasing}. Suppose $r \in (-\infty,0)$ and $\sigma>0$. For $r\in  (-1/2,1/2)$, the frequency and amplitude errors are $C\sigma m^{r-3/2} \sqrt{\log(m)}$ and $C\sigma m^{r-1/2} \sqrt{\log(m)}$. Compared to the other two situations, Gradient-MUSIC provides better accuracy for both frequency and amplitude estimation. The rate improves as $r$ is made more negative, until it saturates at $r = -1/2$, since making $r$ even more negative just results in $\tr(\bfSigma)$ being bounded by a constant, as seen in \eqref{eq:Sigmacolor}.  
	
	\subsection{Minimax rates for deterministic perturbations} 
	\label{sec:minimax}
	
	This section derives lower bounds for optimal rates of approximation under deterministic perturbations quantified in the $\ell^p$ norm. We define the {\it parameter space},
	\begin{equation}
		\label{eq:parameterspace}
		\calP(s):= \big\{(\bfx,\bfa)\colon |\bfx|=|\bfa|=s \andspace |a_j|>0 \forallspace j=1,\dots,s \big\}. 
	\end{equation}
	The requirement $|a_j|>0$ for all $j$ ensures that there is a bijection between $\calP(s)$ and all possible exponential sums with exactly $s$ distinct frequencies. Due to this bijection, it is equivalent to work with $\calP(s)$ instead. More generally, we view data $\tilde\bfy$ as a function from the parameters space and noise to data. Namely, define $\tilde\bfy\colon \calP(s)\times \calN \to \C^{2m-1}$ by the equation,
	\begin{equation}
		\label{eq:information}
		\tilde\bfy:=\tilde\bfy(\bfx,\bfa,\bfeta) = \bfPhi(2m-1,\bfx) \bfa + \bfeta.
	\end{equation}
	As seen in \cref{rem:impossible}, if the perturbation is allowed to be too large, then $s$ cannot be determined from observations. The question we ask is if the perturbation is forced to be small enough and a method also knows $s$, then what is the best rate of approximation?  
	
	Fix any nonempty subsets $\calS\subset\calP(s)$ and $\calN\subset \C^{2m-1}$. For the set of signals, we will consider the signal class \cref{def:signalmodel}. For the noise set $\calN$, we will consider a natural scenario where $\bfeta$ is contained in $B^{2m-1}_p(\epsilon)$, the closed $\ell^p$ ball in $\C^{2m-1}$ of radius $\epsilon$ centered at zero. Consider the set of all possible data generated from the signal and noise sets,
	$$
	\calY:=\{\tilde\bfy(\bfx,\bfa,\bfeta)\colon (\bfx,\bfa)\in\calS, \bfeta\in\calN\}.
	$$
	A function $\psi: \calY\to\calP(s)$ is called a {\it method}. This definition does not require $\psi$ to be computationally tractable or even numerically computable. As is customary, we let 
	$$
	(\hat\bfx,\hat\bfa)=\psi(\tilde\bfy(\bfx,\bfa,\bfeta))
	$$ 
	denote the output of $\psi$, and we simply write $(\hat\bfx,\hat\bfa)$ instead of $\psi$. Just like the other parts of this paper, we have implicitly assumed that $\bfx$ and $\hat\bfx$ have been indexed to best match each other, which induces an ordering on $\hat\bfa$. 
	
	By looking at the worst case signal and noise, the frequency and amplitude recovery rates for a method $\psi$ (which is identified with its output $(\hat\bfx,\hat\bfa)$) are
	\begin{align*}
		X(\psi,\calS,\calN)&:=\sup_{(\bfx,\bfa)\in \calS} \, \sup_{\bfeta\in\calN} \,  \max_{j=1,\dots,s} |x_j-\hat x_j|,  \\
		A(\psi,\calS,\calN)&:=\sup_{(\bfx,\bfa)\in \calS} \, \sup_{\bfeta\in\calN} \,  \max_{j=1,\dots,s} |a_j-\hat a_j|. 
	\end{align*}
	We need to split the frequency and amplitude errors because as we will see, they have non-equivalent rates of approximation. By taking the infimum over all methods defined on $\calS$ and $\calN$, the minimax rates for frequency and amplitude estimation are defined as 
	\begin{align*}
		X_*(\calS,\calN):=\inf_{\psi} X(\psi,\calS,\calN) \andspace
		A_*(\calS,\calN):=\inf_{\psi} A(\psi,\calS,\calN). 
	\end{align*}
	Since these definitions do not preclude the use of intractable algorithms, it is not clear if there is an efficient optimal algorithm that attains these minimax errors.

	We have the following lower bounds. 
	\begin{lemma}
		\label{lem:minimax}
		Let $s\geq 1$, $m\geq 4s$, $0<r_0\leq r_1$, and $\calS:=\calS(8\pi/m,r_0, r_1)\cap \calP(s)$. For any $p\in [1,\infty]$ and $\epsilon\leq 16\pi r_0 m^{1/p}$, 
		\begin{equation*}
			X_*\left(\calS,B_p^{2m-1}(\epsilon)\right)  \geq \frac{\epsilon}{8 r_0 m^{1+1/p}} \andspace
			A_*\left(\calS,B_p^{2m-1}(\epsilon)\right)  \geq \frac{\epsilon}{8 m^{1/p}}.
		\end{equation*}
	\end{lemma}
	
	The proof of \cref{lem:minimax} is located in \cref{proof:minimax}. Note the assumption $m\geq 4s$ in this lemma is necessary because $\calS(8\pi/m,r_0,r_1)\cap \calP(s)=\emptyset$ otherwise.  
	
	There are clearly some similarities between the lower bounds in \cref{lem:minimax} and the performance guarantee for Gradient-MUSIC in \cref{cor:deterministic}. We say a method $(\hat\bfx,\hat\bfa)$ is {\it optimal} on $\calS\subset\calP(s)$ and $B_p^{2m-1}(\epsilon)$ if there are constants $C_1,C_2>0$ such that for all sufficiently large $m$ and small $\epsilon$, 
	\begin{align*}
		\sup_{(\bfx,\bfa)\in \calS} \, \sup_{\bfeta\in B^{2m-1}_p(\epsilon)} \,  \max_{j=1,\dots,s} |x_j-\hat x_j| 
		&\leq C_1 \, X_*\left(\calS,B^{2m-1}_p(\epsilon)\right), \\
		\sup_{(\bfx,\bfa)\in \calS} \, \sup_{\bfeta\in B^{2m-1}_p(\epsilon)} \,  \max_{j=1,\dots,s} |a_j-\hat a_j| 
		&\leq C_2 \, A_*\left(\calS,B^{2m-1}_p(\epsilon)\right). 
	\end{align*}
	Combining \cref{lem:minimax} and \cref{cor:deterministic}, we obtain the following corollary. 
	
	\begin{corollary}
		\label{cor:minimaxp}
		Let $s\geq 1$, $m\geq \max\{4s,100\}$, $0<r_0\leq r_1$, and $\calS:=\calS(8\pi/m,r_0, 10r_0)\cap \calP(s)$. For any $p\in [1,\infty]$ and $\epsilon\leq r_0 m^{1/p}/400$, 
		\begin{equation*}
			X_*\left(\calS,B_p^{2m-1}(\epsilon)\right) \asymp \frac{\epsilon}{r_0 m^{1+1/p}} \andspace
			A_*\left(\calS,B_p^{2m-1}(\epsilon)\right) \asymp \frac{\epsilon}{m^{1/p}}.
		\end{equation*}
		Moreover, Gradient-MUSIC is optimal on $\calS$ and $B_p^{2m-1}(\epsilon)$.  
	\end{corollary}
	
	\section{General theory and framework}
	\label{sec:main}
	
	\subsection{An abstract perspective of subspace methods}
	\label{sec:abstract}
	
	Let $\bbU^{m\times s}$ be the Grassmannian, the set of all complex subspaces of dimension $s$ in $\C^m$. We often identify $\bfW\in \bbU^{m\times s}$ with a matrix $\bfW$ that has orthonormal columns, and vice versa. Although not immediately clear in its present form, spectral estimation is closely related to the problem of finding special subspaces in $\bbU^{m\times s}$ that we define below. 
	
	\begin{definition}\label{def:Fouriersubspace}
		We say $\bfU$ is a {\it Fourier subspace} in $\C^m$ if for some $\bfx\subset\T$, we have $\bfU=\range(\bfPhi(m,\bfx))$ . In this case, we say $\bfU$ is generated by $\bfx$ and $\dim(\bfU)=|\bfx|$. We let $\bbF^{m\times s}$ denote the collection of all $s$ dimensional Fourier subspaces in $\C^m$. 
	\end{definition}
	
	The starting point is, given a subspace $\bfW\in \bbU^{m\times s}$, how do we tell if $\bfW$ is actually a Fourier subspace? In an abstract sense, the main idea of classical MUSIC is to test $\bfW$ against a special family of vectors and see how $\bfW$ reacts. Recall \cref{def:nscorrelation}, which importantly defines the landscape function associated to {\it any} arbitrary subspace. Next, for clarity, we restate \cref{prop:qroots} in a more general form. 
	
	\begin{proposition}
		\label{lem:qroots2}
		Let $m>s$ and $q:=q_{\bfU}$ be the landscape function associated with a Fourier subspace $\bfU\in \bbU^{m\times s}$. We have $q(t)=0$ if and only if $\bfphi(t)\in \bfU$ if and only if $t\in \bfx$.
	\end{proposition}
	
	Provided that $\bfU$ is a Fourier subspace, the landscape function $q_\bfU$ has exactly $s$ unique roots $\bfx$ and each one is a double root. The roots $\bfx$ uniquely specifies the subspace $\bfU$. If $\bfW\in \bbU^{m\times s}$ is just an arbitrary subspace, the conclusions of \cref{lem:qroots2} do not hold for $q_{\bfW}$. Let us briefly summarize some basic properties of landscape functions. Recall that $\calT_m$ denotes the space of trigonometric polynomials of degree at most $m-1$.
	
	\begin{lemma}
		\label{lem:landscape1}
		Let $m>s$ and $q:=q_{\bfW}$ be the landscape function associated with a subspace $\bfW\in \bbU^{m\times s}$. For all $t\in \T$, it holds that $0\leq q(t)\leq 1$. Furthermore, $q$ is not identically zero, $q\in \calT_m$, and admits a polynomial sum-of-squares representation. For each integer $\ell\geq 0$, we have 
		\begin{equation}
			\label{eq:qnorm1}
			\|q^{(\ell)} \|_{L^\infty(\T)} \leq m^\ell. 
		\end{equation}
	\end{lemma}
	
	The proof of this lemma is in \cref{proof:landscape1}. Let us study the landscape functions in greater detail. Define the sets 
	$$
	\calF_{m,s}
	:=\left\{q_{\bfU}\colon \bfU\in \bbF^{m\times s} \right\} \andspace
	\calU_{m,s}
	:=\left\{q_{\bfW}\colon \bfW\in \bbU^{m\times s} \right\}. 
	$$
	Of course, we have $\calF_{m,s}\subset \calU_{m,s} \subset \calT_m$. The inclusion $\calF_{m,s}\subset \calU_{m,s}$ is strict because $\bbF^{m\times s}$ is a strict subset of $\bbU^{m\times s}$ when $m>s$. The inclusion $\calU_{m,s} \subset \calT_m$ is also strict, since \cref{lem:landscape1} tells us that each $q_\bfW$ is necessarily a polynomial sum-of-squares, and not every trigonometric polynomial can be written this way. 
	
	More abstractly we can think of $q$ as the map $q \colon\bbU^{m\times s} \to \calU_{m,s}$ by $\bfW\mapsto q_{\bfW}$. We present a consequence, which is proved \cref{proof:bijection}. 
	
	\begin{lemma}
		\label{lem:bijection}
		Let $m>s$. There is a bijection between the following three families. 
		\begin{enumerate}[(a)]
			\item 
			All subsets of $\T$ of cardinality $s$.
			\item 
			$\bbF^{m\times s}$, the collection of Fourier subspaces in $\C^m$ of dimension $s$.
			\item 
			$\calF_{m,s}$, the set of landscape functions associated with $\bbF^{m\times s}$.
		\end{enumerate}
	\end{lemma}
	
	As a consequence of \cref{lem:bijection}, we write $\bfx\simeq_m \bfU\subset\C^m$ to reference this bijection between the set $\bfx$ and its generated Fourier subspace $\bfU$. We include $m$ in the subscript of $\simeq_m$ since this is only a bijection once we have specified $m$, the number of rows of $\bfU$. 
	
	Recall the definition of the sine-theta distance in \eqref{eq:sintheta}. The following definition now quantifies what we mean by a good approximation $\tilde\bfU$ to the Fourier subspace $\bfU$ generated by a particular $\bfx$. 
	
	\begin{definition}
		\label{def:subspaceerror}
		Given a Fourier subspace $\bfU\in \bbF^{m\times s}$ identified with $\bfx\subset \T$ and an arbitrary $\tilde\bfU\in \bbU^{m\times s}$, we call $\vartheta:=\vartheta(\bfx,\tilde\bfU):=\vartheta(\bfU,\tilde\bfU)$ the (induced) {\it subspace error} between $\bfx$ and $\tilde\bfU$. 
	\end{definition}

	We briefly summarize our new perspective on classical MUSIC. We have three equivalent ways of viewing this algorithm. The original context of finding $\bfx$. This is equivalent to finding its unique Fourier subspace $\bfU$ that is generated by $\bfx$. On the other hand, we can instead study $q_\bfU\in \calF_{m,s}$. 
	
	From this point of view, we see that the main idea behind classical MUSIC is to study the properties of $q_{\tilde \bfU}$, where $\tilde\bfU$ is a perturbation of $\bfU\simeq_m \bfx$, in order to gain some insight into $\bfx$. However, to carry out this plan, we need to study the stability (i.e., metric) properties of the map $q\colon \bbU^{m\times s} \to \calU_{m,s}$. Naturally, $\bbU^{m\times s}$ is a metric space when equipped with the sine-theta distance. We can equip $\calT_m$ with the $C^k(\T)$ norm for any $k\geq 0$. We have the following lemma, which is proved in \cref{proof:landscape2}.  
	
	\begin{lemma}
		\label{lem:landscape2}
		Let $m>s$. Let $q_{\bfV}$ and $q_{\bfW}$ be the landscape function associated with $\bfV,\bfW \in \bbU^{m\times s}$ and $\vartheta:=\vartheta(\bfV,\bfW)$. For each integer $\ell\geq 0$, we have
		\begin{equation}
			\label{eq:qnorm2}
			\big\|q_{\bfV}^{(\ell)} -q_{\bfW}^{(\ell)}\big\|_{L^\infty(\T)}
			\leq \vartheta m^\ell. 
		\end{equation}
	\end{lemma}
	
	Thus, the map $q\colon \bbU^{m\times s} \to \calU_{m,s}$ is a continuous map from one metric space to another. Indeed, we see that 
	$$
	\|q_{\bfV}-q_{\bfW}\|_{C^k(\T)}
	\leq \vartheta(\bfV,\bfW) (1+m+\cdots+m^k). 
	$$ 
	Figure \ref{fig:qmap} shows this thought process as a diagram.

	\begin{figure}[ht]
		\centering
		\begin{tikzpicture}[scale=0.5]
			\filldraw[fill=black!10!white, draw=black] (0,0) rectangle (12,4);
			\draw[thick] (5,2) ellipse (3 and 1);
			\draw (12,4) node[anchor=north east]{$\mathbb{U}^{m\times s}$}; 
			\draw (8,2) node[anchor=west]{$\mathbb{F}^{m\times s}$}; 
			\draw[thick,->] (14,2) -- (18,2); 
			\draw (16,2) node[anchor=south]{$q$};
			\filldraw[fill=black!10!white, draw=black] (20,0) rectangle (32,4);
			\draw[thick] (25,2) ellipse (3 and 1);
			\draw (32,4) node[anchor=north east]{$\calU_{m,s}$}; 
			\filldraw[fill=black] (2,2) circle (0.1) node[anchor=east]{$\bfU$};
			\filldraw[fill=black] (2.5,3) circle (0.1) node[anchor=east]{$\tilde\bfU$};
			\filldraw[fill=black] (22,2) circle (0.1) node[anchor=east]{$q_{\bfU}$};
			\filldraw[fill=black] (22.5,3) circle (0.1) node[anchor=east]{$q_{\tilde\bfU}$};
			\draw (28,2) node[anchor=west]{$\calF_{m,s}$};
		\end{tikzpicture}
		\caption{Diagram of the map $\bbF^{m\times s}\subset\bbU^{m\times}$ to $\calF_{m,s}\subset \calU_{m,s}$. Arbitrary $\bfU\in \bbF^{m\times s}$ and $\tilde\bfU\in \bbU^{m\times s}$ are mapped to their associated landscape functions $q_{\bfU}\in \calF_{m,s}$ and $q_{\tilde\bfU}\in \calU_{m,s}$. }
		\label{fig:qmap}
	\end{figure}
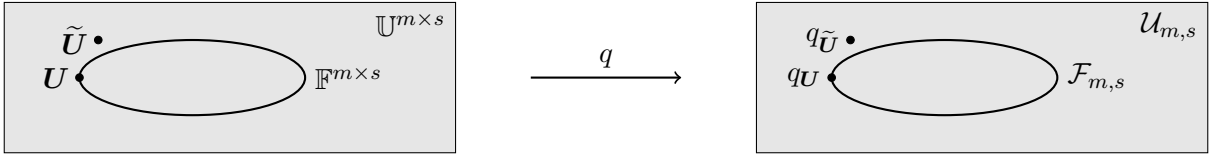
	
	In both classical MUSIC and the full pipeline form of Gradient-MUSIC (\cref{alg:gradMUSIC}), the subspace $\tilde \bfU$ is the Toeplitz estimator, see \cref{def:Toeplitzestimator}. 
	
	\subsection{Sparsity detection and the Toeplitz estimator}
	\label{sec:sparsity}
	
	We turn our attention to the the estimation of $\tilde\bfU$ and $|\bfx|$ from just the data $\tilde \bfy$, which were assumed to be known in both \cref{alg:MUSIC,alg:gradMUSIC}. If $\tilde \bfT$ is a sufficiently small perturbation of the low rank matrix $\bfT$, a natural strategy is to count how many singular values of $\tilde \bfT$ exceed a certain threshold. We introduce the following definition.
	
	\begin{definition}
		\label{def:subspaceerror2}
		For any $(\bfx,\bfa)$ and $\bfeta$, let $s=|\bfx|$ and define the {\it Toeplitz subspace error} 
		\begin{equation*}
			\rho
			:=\rho(\bfx,\bfa,\bfeta)
			:=\frac{2\|T(\bfeta)\|_2}{\sigma_s(T(\bfy))}.
		\end{equation*}
	\end{definition}
	
	A few elementary properties of $\rho$ are summarized in the following lemma, which is proved in \cref{proof:nsrrelationship2}. 
	
	\begin{lemma}
		\label{lem:rho}
		Let $m>s$. For any $(\bfx,\bfa)$ and $\bfeta$, set  $\rho:=\rho(\bfx,\bfa,\bfeta)$. If $\rho\leq 1-1/\sqrt 2$, then $\tilde \bfT$ has a unique leading $s$ dimensional left singular space and the Toeplitz estimator $\tilde\bfU$ for $\bfx\simeq_m\bfU$ is well defined. Additionally, $\vartheta(\bfU,\tilde\bfU) \leq \rho$.  
	\end{lemma}

	\begin{remark}
		\label{rem:vartheta}
		There are some subtle differences between $\vartheta$ and $\rho$. While \cref{lem:rho} shows that $\vartheta\leq \rho$ whenever $\rho$ is sufficiently small, the reverse inequality does not hold in general because it is possible that $\vartheta=0$ yet $\rho\not=0$. For instance, choose a nonzero $\bfeta$ that perturbs the singular values of $\bfT$ but not the subspace $\bfU$. Then we have $\tilde \bfU\tilde \bfU^*=\bfU\bfU^*$, $\vartheta=0$, and $\rho\not=0$. 
	\end{remark}
	
	It is not difficult to prove that if $\rho(\bfx,\bfa,\bfeta)<1$, then 
	\begin{equation}
		\label{eq:thresholdnaive}
		|\bfx| = \max\Big\{k\colon \sigma_k(\tilde \bfT) >\|T(\bfeta)\|_2\Big\}.
	\end{equation}
	A deficiency of this threshold procedure is that $\|T(\bfeta)\|_2$ is unknown and behaves very differently depending on the properties of $\bfeta$, as will be seen shortly. Replacing $\|T(\bfeta)\|_2$ on the right hand side of \eqref{eq:thresholdnaive} with an under- or over- estimate of may lead to a wrong value for the number of frequencies, which may be catastrophic. The following lemma provides an alternative threshold strategy and is proved in \cref{proof:sparsity}. 
	
	\begin{lemma}
		\label{lem:sparsity}
		Let $m\geq 1$, $\beta>1$, and $0<r_0\leq r_1$. For any $(\bfx,\bfa)\in \calS(2\pi \beta/m,r_0,r_1)$ and perturbation $\bfeta$ such that $\rho:=\rho(\bfx,\bfa,\bfeta)\leq 1-1/\sqrt 2$, we have 
		$$
		\frac{\sigma_s(\tilde \bfT)}{\sigma_1(\tilde \bfT)}
		\geq \frac 7{10} \frac {r_0} {r_1} \frac{\beta-1}{\beta+1} \andspace
		\frac{\sigma_{s+1}(\tilde \bfT)}{\sigma_1(\tilde \bfT)}
		\leq \frac{\rho}{2-\rho}. 
		$$
		In particular, if
		$$
		\rho
		\leq \min\left\{ 1- \frac 1 {\sqrt 2}, \, \frac {r_0} {r_1} \frac{\beta-1}{\beta+1} \right\},
		\quad \text{then} \quad 
		|\bfx|=\max\left\{k \colon \frac{\sigma_k(\tilde \bfT)}{\sigma_1(\tilde \bfT)} \geq \frac 7 {10}\frac {r_0} {r_1} \frac{\beta-1}{\beta+1} \right\}. 
		$$
	\end{lemma}
	
	The point of this lemma is that the threshold parameter only depends on the signal class parameters and are chosen without knowledge of $\|T(\bfeta)\|_2$ or $\rho(\bfx,\bfa,\bfeta)$. We do not need extremely accurate estimates for $\beta$ or $r_0/r_1$ to get the correct value for $|\bfx|$. The penalty of using a loose estimate for either parameter is met by a stronger prior assumption on $\rho$. This threshold procedure is used in the first step of the full pipeline version of Gradient-MUSIC (\cref{alg:gradMUSIC2}). 
	
	Finally, we need to upper bound the Toeplitz subspace error $\rho$ which involves computing $\|T(\bfeta)\|_2$ under various assumptions on $\bfeta$. The following lemmas are proved in \cref{proof:noisepnorm} and \cref{sec:toeplitznorm}.
	
	\begin{lemma}
		\label{lem:noisepnorm}
		For any $p\in [1,\infty]$ and $\bfeta\in \C^{2m-1}$, we have $\|T(\bfeta)\|_2\leq 2 m^{1-1/p}\|\bfeta\|_p$.
	\end{lemma} 
	
	\begin{lemma}
		\label{lem:toeplitznorm}
		Let $\bfSigma\in \R^{(2m-1)\times (2m-1)}$ be a diagonal matrix with positive diagonals. If $\bfeta\sim \calN(\bfzero,\bfSigma)$, then the random matrix $T(\bfeta)\in \C^{m\times m}$ satisfies 
		\begin{align*}
			\E \|T(\bfeta)\|_2 
			&\le \textstyle \sqrt{2 \tr(\bfSigma)\log(2m) }, \\
			\P\left\{\|T(\bfeta)\|_2 \ge t \right\} &\le 2m e^{-t^2/(2 \tr(\bfSigma))} \forallspace t\geq  0.
		\end{align*}	
	\end{lemma}
	
	Both lemmas are standard and may have already been discovered. Earlier analysis \cite{liao2016music} proved \cref{lem:noisepnorm} for $p=2$, which then implies the conclusion for all $p\in [2,\infty]$ through H\"older's inequality. The only new content of the lemma is for $p\in [1,2)$, which is done by using Young's convolution inequality. As for \cref{lem:toeplitznorm}, the i.i.d. Gaussian case was also established in \cite{liao2016music}, and here we have a more general version that holds for diagonal covariance matrix. The proof relies on a standard matrix concentration argument \cite[Theorem 4.1.1]{tropp2012user}.
	
	\subsection{Landscape analysis}
	\label{subsec:landscapeanalysis}
	
	Our overall goal is to use Gradient-MUSIC (\cref{alg:gradMUSIC}) to find the $s$ smallest local minima of $\tilde q:=q_{\tilde\bfU}$, where $\tilde\bfU$ is a sufficiently small perturbation of $\bfU$. In order for this approach to succeed, we need to show that such minima exist and they well approximate $\bfx$, and we need enough information about the landscape function $\tilde q$ to get suitable control over gradient descent dynamics. This is accomplished by the following theorem.
	
	\begin{theorem}[Landscape analysis]
		\label{thm:landscape}
		Let $m\geq 100$. For any $\bfx=\{x_j\}_{j=1}^s\subset \T$ and $\tilde\bfU\in \bbU^{m\times s}$ such that $\Delta(\bfx)\geq 8\pi/m$ and $\vartheta:=\vartheta(\bfx,\tilde\bfU)\leq 0.01$, the landscape function $\tilde q:=q_{\tilde\bfU}$ associated with $\tilde\bfU$ has at least $s$ critical points $\tilde \bfx:=\{\tilde x_1,\dots,\tilde x_s\}$ such that the following hold for each $j\in \{1,\dots,s\}$. 
		\begin{enumerate}[{\rm (a)}]
			\item 
			$\displaystyle |\tilde x_j-x_j|
			\leq \frac{7\vartheta}{m}.
			$
			\item 
			$0.0271 \, m^2 \leq \tilde q \, ''(t)\leq 0.269 \, m^2$ $\displaystyle
			\forallspace t\in \left[\tilde x_j-\frac{\pi}{3m}, \, \tilde x_j+\frac{\pi}{3m}\right]$.
			\item 
			In particular, $\tilde x_j$ is local minimum of $\tilde q$ and   
			$
			\tilde q(\tilde x_j)
			\leq \tilde q(x_j)
			\leq \vartheta^2
			\leq 10^{-4}.
			$
			\item 
			$\displaystyle
			\tilde q \, '(t) \geq +0.0306 \, m \forallspace t\in \left[\tilde x_j+\frac{\pi}{3m}, \, \tilde x_j+\frac{4\pi}{3m}\right],\smallskip \\  
			\tilde q \, '(t) \leq -0.0306 \, m \forallspace t\in \left[\tilde x_j-\frac {4\pi} {3m}, \,  \tilde x_j-\frac{\pi }{3m}\right].
			$
			\item 
			$\tilde q(t)\geq 0.529$ for all $\displaystyle t\not\in \bigcup_{j=1}^s \left[\tilde x_j- \frac{4\pi}{3m}, \, \tilde x_j+\frac{4\pi}{3m}\right].
			$
			\item 
			$\displaystyle \tilde q \, '(t) \leq +0.292 \, m^2|t-\tilde x_j| \forallspace t\in \left[\tilde x_j+\frac{\pi}{3m}, \, \tilde x_j+\frac{4\pi}{3m}\right], \smallskip \\
			\tilde q \, '(t) \geq -0.292 \, m^2|t-\tilde x_j| \forallspace t\in \left[\tilde x_j-\frac {4\pi}{3m}, \,  \tilde x_j-\frac{\pi }{3m}\right].
			$
		\end{enumerate}	
	\end{theorem}
	
	\cref{thm:landscape} is proved in \cref{proof:landscape}. Parts (a) and (c) tell us that $\tilde q$ has small local minima $\tilde \bfx$ that are close to $\bfx$. Part (b) tells us $\tilde q$ is convex in a $\pi/(3m)$ neighborhood of $\tilde\bfx$. While it is not necessarily convex in a larger region, parts (d) and (f) tell us that $\tilde q$ is strictly decreasing on $[\tilde x_j-4\pi/(3m),\tilde x_j)$ and strictly increasing on ($\tilde x_j,\tilde x_j+4\pi/(3m)]$. While it is possible that $\tilde q$ has other local minima, they are necessarily at least $4\pi/(3m)$ away $\tilde \bfx$ and part (e) implies that $\tilde q$ is large. This explains the behavior seen in Figure \ref{fig:landscape}.
	
	\begin{remark} \label{rem:generality}
		\cref{thm:landscape} is purely an approximation result because $\tilde\bfU\in \bbU^{m\times s}$ is an arbitrary subspace and we do not assume it is produced by any particular algorithm. While we will primarily use the theorem when $\tilde\bfU$ is the Toeplitz estimator described in \cref{sec:sparsity}, there are several reasons why we could potentially be interested in other $\tilde\bfU$, which we discuss in \cref{sec:differentsubspace}. 
	\end{remark}
	
	The landscape analysis is the primary technical contribution of this paper. The lemmas that are used to prove \cref{thm:landscape} are stated with greater generality than the theorem itself. \cref{sec:landscape} is a self contained section that includes all lemmas used in the landscape analysis and their proofs.
	
	It may be helpful to briefly explain how \cref{thm:landscape} is proved and why it holds. For simplicity, we concentrate our discussion on the properties of $q_\bfU$ where $\bfx\simeq_m \bfU$. A key observation is that for any $x_j\in \bfx$, we have
	\begin{equation}
		\label{eq:localapprox}
		q(t)=1-f_m(t-x_j)+p_j(t),
	\end{equation}
	where $p_j(t)$ is a polynomial that acts as an error term and $f_m$ is a normalized {\it Fej\'er kernel}, 
	\begin{equation}
		\label{eq:fejerkernel}
		f_m(t)
		:=\frac 1 {m^2} \left( \frac{\sin(mt/2)}{\sin(t/2)} \right)^2. 
	\end{equation}
	Provided that the influence of $p_j$ is negligible for $t$ near $x_j$, this implies that $q(t)$ (and also $\tilde q$ through some additional analysis) behaves like $1-f_m(t-x_j)$. This explains why the  graph of $\tilde q$ in Figure \ref{fig:landscape2} (b) looks strikingly similar to a shift of the graph of $1-f_m$. 
	
	The primary difficultly with using formula \eqref{eq:localapprox} is that $f_m$ and $p_j$ both have the same scaling in $m$, e.g., their $\ell$-th derivative scales as $m^\ell$ for all $\ell\geq 0$. In order to effectively use this approximation, one must show that the constants in $f_m^{(\ell)}$ dominate those that appear in $p_j^{(\ell)}$. Note that $f_m$ is a fixed function while $p_j$ depends on $\bfx$, so a loose upper bound for $p_j$ may depend on $s$. We went through great care to control $p_j$ and its derivatives by explicit quantities that do not depend on $s$, otherwise we cannot utilize formula \eqref{eq:localapprox} without placing additional unnatural assumptions. This is done by arguing that $p_j$ can be decomposed into $s-1$ almost orthogonal functions. 
	
	\begin{remark}[Sharpness of \cref{thm:landscape}]
		\label{rem:naturalness}
		Identity \eqref{eq:localapprox} is natural since $p_j=0$ for the special case where $\bfx$ is just a singleton $\bfx=\{x_1\}$, and so $q(t)=1-f_m(t-x_1)$. This immediately implies that several features of \cref{thm:landscape} are sharp. Indeed, the convexity estimate in part (a) and the derivative estimates in parts (d) and (f) can only hold in a $2\pi c_1/m$ and $2\pi c_2/m$ neighborhood of $\tilde \bfx$, for some universal $0<c_1<c_2<1$. Hence, \cref{thm:landscape} shows that the graph of $\tilde q$ has wells that are maximally wide. Finally, $|\tilde q\, ^{(\ell)}|$ must scale like $m^\ell$, which is what we obtained in the theorem. 
	\end{remark}

	Variations of \cref{thm:landscape} can be obtained by following its proof and using the lemmas for a different set of parameters. A script that calculates the resulting constants for variable parameter choices are included with the software accompanying this paper. We mention two regimes of interest.
	
	\begin{remark}
		The constant in \cref{thm:landscape} that we optimized for is the ``4" in condition $\Delta(\bfx)\geq 2\pi (4)/m$. Had we instead assumed that $\Delta(\bfx)\geq 2\pi \beta/m$ for some $\beta>0$, then our proof techniques break down when $\beta\approx 3.45$, even when $\vartheta=0$. This is expected since the conclusions of \cref{thm:landscape} do not hold when $m\Delta\ll 2\pi$, even if $\vartheta=0$, so our proof argument has to yield vacuous statements for small enough $\beta>1$. 
	\end{remark}
	
	\begin{remark}
		There is an opposite extreme that is also relevant. Say $\bfx$ is fixed while $m\to\infty$, which is the usual set up of statistical problems. In this case, we can pick an arbitrarily large $\beta$ and assume $\Delta(\bfx)\geq 2\pi \beta/m$, since this condition will eventually be satisfied for some $m$. All of the numerical constants in \cref{thm:landscape} improve as $\beta$ is made larger and our sufficient condition on $\vartheta$ relaxes. For example, if $\Delta(\bfx)\geq 2\pi(20)/m$, then $\vartheta\leq 0.06$ suffices to get non-vacuous conclusions.   
	\end{remark}

	\subsection{Performance guarantees for Gradient-MUSIC} 
	\label{sec:gradientMUSIC}
	
	We next analyze the Gradient-MUSIC algorithm (\cref{alg:gradMUSIC}). In expository sections, we assume that the assumptions of \cref{thm:landscape} hold and that $\tilde q = q_{\tilde\bfU}$ for an arbitrary $\tilde\bfU\in \bbU^{m\times s}$ such that $\vartheta(\bfx,\tilde\bfU)\leq 0.01$. Recall the notation and definitions introduced in \cref{sec:gradMUSIC}. 
	
	Notice that \cref{thm:landscape} part (b) and (d) suggest that gradient descent (with appropriate parameters) produces iterates that converge to $\tilde x_j$ whenever the initial guess lies in the interval, 
	\begin{equation}
		\label{eq:basin}
		B_j:=\left[\tilde x_j-\frac {4\pi}{3m}, \, \tilde x_j+\frac {4\pi}{3m} \right].
	\end{equation}
	We will prove convergence later, but for now, we concentrate on finding a suitable initialization in each $B_j$. To do this, \cref{thm:landscape} part (e) tells us that $\tilde q$ is large away from $B_j$. We just need to evaluate $\tilde q$ on a reasonable finite $G\subset\T$. We want a coarse $G$ in order to reduce computational costs, but simultaneously, not too coarse otherwise our strategy may fail to find suitable initialization. Recall the definitions of accepted set (\cref{def:accepted}) and cluster (\cref{def:cluster}).

	\begin{lemma}
		\label{lem:grad}
		Suppose the assumptions of \cref{thm:landscape} hold. For any $\alpha \in (\vartheta^2,0.529]$ and finite $G\subset\T$ such that $\mesh(G)< (\alpha-\vartheta^2)/m$, the following hold. 
		\begin{enumerate}[(a)]
			\item 
			$A\subset \bigcup_{j=1}^s B_j$.
			\item 
			$A\cap B_j\not=\emptyset$ for each $j$.
			\item 
			Each $A_j:=A\cap B_j$ is a cluster in $G$.
		\end{enumerate}
	\end{lemma}
	
	This lemma is proved in \cref{proof:grad}. According to this lemma, the accepted elements $A$ can be written as a disjoint union of exactly $s$ clusters. We only need to run gradient descent with one initialization per $A_j$, so pick any $t_{j,0}\in A_j$, which we call a {\it representative}. It does not matter which representative is selected per $A_j$, but a reasonable heuristic is to pick a representative that is in the middle of $A_j$. This explains the behavior seen in Figure \ref{fig:landscape2}. 
	
	To show that our method succeeds, it remains to determine a suitable step size $h$ for gradient descent and prove that the iterates converge to $\tilde x_j$. Convergence does not immediately follow because $\tilde q$ is only provably convex in a strict subset of $B_j$, see \cref{thm:landscape} part (b), and gradient descent does not know if the current iterate lies in a region of convexity or not. Regardless, we will overcome this technical issue by using lower and upper bounds for $\tilde q\,$ given in \cref{thm:landscape} part (d) and (f).
	
	We will let $\hat\bfx$ be the output of Gradient-MUSIC. There are two types of errors that contribute to the frequencies error $|x_j-\hat x_j|$. By triangle inequality, we have
	$$
	|x_j-\hat x_j|\leq |x_j-\tilde x_j|+|\tilde x_j-\hat x_j|.
	$$
	The first term $|x_j-\tilde x_j|$ is a fundamental approximation error that comes from the landscape function itself, which is controlled by \cref{thm:landscape}. Gradient-MUSIC approximates each tilde $\tilde x_j$ by gradient descent with suitable initialization. We can interpret the second term $|\tilde x_j-\hat x_j|$ as the {\it optimization error}. 
	
	The following theorem provides a performance guarantee for Gradient-MUSIC and explicit parameters to use. This makes it a user friendly and direct algorithm.
	
	\begin{theorem}[Gradient-MUSIC performance guarantee]
		\label{thm:maingradientMUSIC}
		Let $m\geq 100$, $\alpha = 0.529$, $G\subset \T$ be finite such that $\mesh(G)\leq 1/(2m)$, $h=6/m^2$, and $n\geq 31$. For any $\bfx=\{x_j\}_{j=1}^s\subset \T$ and $\tilde\bfU\in \bbU^{m\times s}$ such that $\Delta(\bfx)\geq 8\pi/m$ and $\vartheta:=\vartheta(\bfx,\tilde\bfU)\leq 0.01$, Gradient-MUSIC (\cref{alg:gradMUSIC}) outputs $\hat \bfx$ such that the frequency error satisfies
		\begin{align*}
			\max_{j=1,\dots,s} |x_j-\hat x_j|
			&\leq \frac{7\vartheta} m + \frac{77 \pi \, (0.839)^n}m.
		\end{align*}
	\end{theorem}
	
	\cref{thm:maingradientMUSIC} is proved in \cref{proof:maingradientMUSIC}. The frequency error $|x_j-\hat x_j|$ consists of two terms, where $7\vartheta/m$ comes from the landscape analysis and $77 \pi \, (0.839)^n/m$ is the optimization error. 
	
	If $\vartheta=0$, then $\tilde q= q$ and $\tilde x_j=x_j$ for each $j$. The optimization error dominates and we can calculate $x_j$ to any desired precision by tuning the number of iterations. When $\vartheta\not=0$, then we ideally pick $n$ sufficiently large so that the two errors are balanced. We see that $77\pi \, (0.839)^n\lesssim \vartheta$ whenever $n\gtrsim \log(1/\vartheta)$, so the number of iterations is typically not large. It may seem strange that $n\gtrsim \log(1/\vartheta)$ grows as $\vartheta$ decreases. This is unavoidable since if $\vartheta$ is small (e.g., due small noise), then $7\vartheta/m$ is small, so we need to increase the number of gradient iterations to so that the optimization matches the perturbation error. If $n$ is not big enough, then the optimization error dominates.  
	
	\begin{remark}
		\cref{thm:maingradientMUSIC} does not make any assumptions on what $\tilde\bfU\in \bbU^{m\times s}$ is, only that we start with a $\tilde\bfU$ such that $\vartheta(\bfU,\tilde\bfU)\leq 0.01$. To see why this generality may be helpful, see the discussion in \cref{rem:generality}. 
	\end{remark}
	
	\begin{remark}
		There are other valid choices for threshold parameter $\alpha$ and step size $h$, which can be extracted from the proof of \cref{thm:maingradientMUSIC}. Other choices will result in different requirements on $\mesh(G)$ and number of gradient iterations $n$, and will lead to different numerical constants. 
	\end{remark}
	
	\begin{remark}[Necessity of $\mesh(G)\leq 1/(2m)$]
		\label{rem:mesh}
		One may ask if it is possible to weaken the assumption that $\mesh(G)\leq 1/(2m)$ in \cref{thm:maingradientMUSIC}. Without additional assumptions or prior information about $\bfx$, it is not possible to make $m \, \mesh(G)\to \infty$ as $m\to\infty$. This can be seen by considering the case where $\bfx=\{x_1\}$ is a singleton and $\vartheta=0$. Then $\tilde q(t)=q(t)=1-f_m(t-x_1)$, see \cref{rem:naturalness}. The interval around $x_1$ for which gradient iterations converge to $x_1$  cannot be larger than $[x_1-2\pi /m,x_1+2\pi/m]$. Sampling and thresholding $\tilde q$ on $G$ can potentially miss this interval if $m \, \mesh(G)\to \infty$. This reasoning also illustrates that even if $\Delta(\bfx)\geq 2\pi \beta/m$ for larger $\beta$, we still need $\mesh(G)\leq 4\pi/m$ regardless of $\beta$. 
	\end{remark}
	
	\begin{remark}[Technical comparison between MUSIC and ESPRIT] \label{rem:esprit}
		\cref{thm:maingradientMUSIC} shows that provided the assumptions there hold, then Gradient-MUSIC produces $\hat\bfx$ such that $\max_j |x_j-\hat x_j|\lesssim \vartheta/m$. Importantly, this bound holds regardless of what $\bfeta$ is or which computational method is used to produce $\tilde \bfU$ as long as $\vartheta\leq 0.01$. 
		
		It was shown in \cite{li2020super} that ESPRIT produces $\hat\bfx$ such that $\max_j |x_j-\hat x_j|\lesssim_s \vartheta$ under some technical assumptions. We do not know if this can be strengthened to $\max_j |x_j-\hat x_j|\lesssim \vartheta/m$ for any $\tilde\bfU$ such that $\vartheta(\bfU,\tilde\bfU)$ is small enough. 
		
		Let us return to a more concrete discussion. The paper \cite{ding2024esprit} showed that for i.i.d. subgaussian noise $\bfeta$ with mean zero and variance $\sigma^2$ that ESPRIT achieves frequency error at most $C\sigma/m^{3/2}$, ignoring $\log(m)$ factors. This is done through a refined eigenvector analysis that relies on the i.i.d. subgaussian noise assumption. It is unclear whether their proof technique can be generalized. On the other hand, \cref{cor:stochastic} showed that Gradient-MUSIC achieves frequency error at most $C\sigma \sqrt{\log(m)}/m^{3/2}$ for i.i.d. gaussian noise (see also \cref{rem:subgaussian}). This was done by first showing that $\vartheta \lesssim C\sigma \sqrt{\log(m)}/\sqrt m$ for the Toeplitz estimator $\tilde\bfU$ and then using the general bound given by  \cref{thm:maingradientMUSIC}. One advantage of our proof structure is its generality, which allows us to handle deterministic and more general stochastic perturbations in a unified manner.
		
		The proof techniques in this paper and \cite{ding2024esprit} are considerably different and one set of results do not imply the other. 
	\end{remark}

	\subsection{Amplitude estimation}
	\label{sec:amplitudes}
	
	In this section, we study the amplitude recovery problem where we are given estimated frequencies $\hat\bfx$ produced by some method (not necessarily Gradient-MUSIC) and we would like to estimate the amplitudes $\bfa$. We study a {\it quadratic method}, which produces
	\begin{equation}
		\label{eq:hata}
		\hat \bfa = \diag\big( \hat \bfA \big), \wherespace \hat\bfA := \hat \bfPhi^+ \tilde \bfT \big(\hat \bfPhi^+ \big)^* \andspace \hat \bfPhi:= \bfPhi(m,\hat\bfx).
	\end{equation}
	This method used in the third step of \cref{alg:gradMUSIC2}. 
	
	The quadratic recovery method also respects the permutation invariance of the spectral estimation problem. If a permuted set $\pi(\hat \bfx)$ is used in formula \eqref{eq:hata} instead of $\hat\bfx$, then we would get $\pi(\hat\bfa)$ in place of $\hat\bfa$. Hence, there is a non-ambiguous pairing between $\hat \bfx$ and $\hat \bfa$, which allows us to use the notation
	$$
	(\hat \bfx, \hat \bfa):=\big\{ (\hat x_j, \hat a_j)\big\}_{j=1}^s.
	$$
	Consequently, without loss of generality, we assume that $\hat\bfx$ has been sorted to best best match $\bfx$, and this fixes a particular ordering on $\hat\bfa$. 
	
	To see why $\hat\bfa$ is a reasonable method of approximation, recall factorization \eqref{eq:toeplitzfactorization} and so  
	\begin{equation}
		\label{eq:hatAhelp}
		\hat\bfA = \big(\hat \bfPhi^+ \bfPhi\big) \bfA \big(\hat \bfPhi^+ \bfPhi \big)^* + \hat \bfPhi^+ T(\bfeta) \big(\hat \bfPhi^+ \big)^*.
	\end{equation}
	Provided that $\hat \bfx$ is a good approximation of $\bfx$ and $\|T(\bfeta)\|_2$ is sufficiently small, then we expect that $\hat \bfPhi^+ \bfPhi$ approximates $\bfI$ and for the diagonal entries of $\hat \bfA$ to approximate that of $\bfA$. Rigorous analysis is carried out in the following lemma whose proof is located in \cref{proof:amplitudes}. 
	
	\begin{lemma}
		\label{lem:amplitudes}
		For $m\geq s$ and $\beta>2$, suppose $\bfx$ and $\hat \bfx$ are both sets of cardinality $s$ such that $\Delta(\bfx)\geq 2 \pi \beta/m$ and $\|\hat\bfx-\bfx\|_\infty \leq 2\pi \alpha/m$ for some $\alpha \leq \beta/4$. Then 
		$$
		\max_{j=1,\dots,s} |a_j - \hat a_j|
		\leq \left( 1 + \sqrt{\frac{\beta+2}{\beta-2}} \right) \frac{\beta + 2}{\beta-2} \, a_{\max} \, m \, \max_{j=1,\dots,s}|x_j-\hat x_j| + \frac \beta {\beta -1 } \frac{\|T(\bfeta)\|_2} m. 
		$$
	\end{lemma}
	
	A different strategy is to approximate the amplitudes $\bfa$ using a (linear) least squares reconstruction method, i.e. given $\hat\bfx$, define
	\begin{equation}
		\label{eq:leastsquares}
		\bfa^\sharp := \argmin_{\bfu} \big\| \tilde \bfy - \hat \bfPhi(2m-1,\hat\bfx) \bfu \big\|_2. 
	\end{equation}
	Prior analysis of least squares for related spectral estimation problems can be found in \cite{moitra2015matrixpencil,li2020super}.

	\section{Comparison to classical MUSIC}
	\label{sec:classicalmusic} 
	
	In this section, we show our analysis can be applied to the classical MUSIC algorithm, which is summarized in \cref{alg:MUSIC}. At this point, we should clarify what we mean by discrete local minima.
	
	\begin{definition}
		\label{def:discretemin}
		Given any finite subset $G$ of $\T$, we say $u\in G$ is a {\it (strict) discrete local minimum} of $f$ on $G$ if $f(u)<f(u_+)$ and $f(u)<f(u_-)$, where $u_-,u_+\in G$ are the left and right nearest neighbors of $u$ on the discrete set $G$, respectively. 
	\end{definition}  
	
	In particular, if $f$ is strictly convex in an interval $I\subset \T$ and $I\cap G$ has at least three elements, then $f$ has exactly one discrete local minimum in $I\cap G$. 
	
	Let $\hat\bfx$ denote the output of classical MUSIC. There are two types of errors that contribute to the frequency error $|x_j-\hat x_j|$. By triangle inequality, 
	$$
	|x_j-\hat x_j|
	\leq |x_j-\tilde x_j|+|\tilde x_j-\hat x_j|. 
	$$
	The first error $|x_j-\tilde x_j|$ is from the landscape function, which is controlled by \cref{thm:landscape}. The classical MUSIC algorithm approximates each $\tilde x_j$ by searching for the local minima of $\tilde q$ on a discrete set $G$. This process introduces the second term $|\tilde x_j-\hat x_j|$, which we call the {\it discretization error}. Not surprisingly, we will see that it is related to $\mesh(G)$. The following theorem is proved in \cref{proof:mainMUSIC}. 
	
	\begin{theorem}[General performance guarantee for MUSIC]
		\label{thm:mainMUSIC}
		Let $m\geq 100$ and $G\subset \T$ such that $\mesh(G)\leq 2\pi/(3m)$. For any $\bfx=\{x_j\}_{j=1}^s\subset \T$ and $\tilde\bfU\in \bbU^{m\times s}$ such that $\Delta(\bfx)\geq 8\pi/m$ and $\vartheta:=\vartheta(\bfx,\tilde\bfU)\leq 0.01$, classical MUSIC (\cref{alg:MUSIC}) outputs $\hat \bfx$ such that the frequency error satisfies
		\begin{align*}
			\max_{j=1,\dots,s} |x_j-\hat x_j|
			&\leq \frac{7\vartheta}m+\mesh(G). 
		\end{align*}
	\end{theorem}
	
	\cref{thm:mainMUSIC} does not make any assumptions on how $\tilde\bfU\in \bbU^{m\times s}$ is found, which makes it a purely computational approximation result. The frequency error has two terms, where $7\vartheta/m$ comes from the landscape analysis and $\mesh(G)$ term upper bounds the discretization error. 
	
	If $\vartheta=0$, then $\mesh(G)$ dominates. We should pick $G$ depending on a preset tolerance parameter, e.g., $\mesh(G)\leq 10^{-6}/m$. In most applications where $\vartheta\not=0$, to balance the two errors, we should set $\mesh(G)\leq C\vartheta/m$ for some $C>0$ that does not depend on $m$. Then the frequency error is at most $C\vartheta/m$.
	
	\begin{remark}
		\label{rem:mesh2}
		Since classical MUSIC finds $\tilde\bfx$ through a search on $G$, without additional information about $\tilde\bfx$, the discretization error can be as large as $\mesh(G)$. To get the best possible upper bound provided by \cref{thm:mainMUSIC}, the requirement $\mesh(G)\asymp \vartheta/m$ cannot be relaxed. This is the coarsest grid that is allowed by classical MUSIC to get the optimal frequency error bound. 
	\end{remark}

	Using \cref{thm:mainMUSIC}, we can provide the analogue of \cref{alg:gradMUSIC2}, \cref{cor:deterministic}, and \cref{cor:stochastic}, but with classical MUSIC instead. For these theorems, assumptions and conclusions are almost identical except that for $\ell^p$ noise, we need 
	$$
	\mesh(G)\lesssim \frac{\|\bfeta\|_p}{a_{\min} \, m^{1+1/p}},
	$$
	while for $\bfeta\sim \calN(\bfzero,\bfSigma)$, we need
	$$
	\mesh(G)\lesssim \frac{\sqrt{\tr(\bfSigma)\log(m)}}{a_{\min}\, m^2}. 
	$$
	
	Let us compare \cref{thm:maingradientMUSIC} for Gradient-MUSIC and \cref{thm:mainMUSIC} for classical MUSIC. They have the same assumptions on $\bfx$, $\vartheta$, and $m$, but the main difference is their assumptions on $G$. 
	\begin{enumerate}[(a)] \itemsep-2pt
		\item 
		For Gradient-MUSIC, to access the best bound provided by \cref{thm:maingradientMUSIC}, we only need $\mesh(G)\leq 1/(2m)$ independent of $\vartheta$. The most natural and simplest choice is a uniform grid of width $1/(2m)$, so that $|G| \asymp m$. In view of \cref{rem:mesh}, we cannot choose a coarser $G$ without additional information about $\bfx$ and/or sacrificing optimality. 
		\item 
		For classical MUSIC algorithm, to access the best bound provided by \cref{thm:mainMUSIC}, we need $\mesh(G)\lesssim \vartheta/m$. Again choosing a uniform grid, we have $|G|\asymp m\vartheta^{-1}$. In view of \cref{rem:mesh2}, we cannot choose a coarser $G$ without additional information about $\bfx$ and/or sacrificing optimality. 
	\end{enumerate}
	
	This justifies why we refer to $G$ as a ``thin" grid for classical MUSIC and a ``coarse" grid for Gradient-MUSIC. Evaluation of $\tilde q$ at a single input $t\in \T$ is not cheap. It requires $O(ms)$ floating point operations due to the formula,
	$$
	\tilde q(t)=1-\left\|\tilde \bfU^*\bfphi(t)\right\|^2_2.
	$$
	This makes classical MUSIC expensive.
	
	While Gradient-MUSIC saves significantly on evaluation of $\tilde q$, it comes with the trade-off that it requires $n\asymp \log(1/\vartheta)$ iterations of gradient descent for each $x_j$. It turns out that this trade-off is always more favorable for Gradient-MUSIC. Indeed, evaluation of $\tilde q \,'$ also has complexity $O(ms)$ because 
	$$
	\tilde q\,'(t)
	=2 \Re \left( \bfphi(t) \tilde \bfU  \tilde \bfU^* \bfphi'(t)\right).
	$$
	Thus, the total complexity for using $n\asymp \log(1/\vartheta)$ gradient iterations to find $s$ local minima is $O(ms\log(1/\vartheta))$. From here, we see that to get the best results,
	\begin{enumerate}[(a)]\itemsep-2pt
		\item 
		Gradient-MUSIC (\cref{alg:gradMUSIC}) has computational complexity 
		$O\left( m^2s + ms^2 \log(1/\vartheta)\right).$
		\item 
		Classical MUSIC (\cref{alg:MUSIC}) has computational complexity $O(m^2s\vartheta^{-1})$. 
	\end{enumerate}
	
	Thus, we see that Gradient-MUSIC is always more efficient than classical MUSIC. This gap widens as $\vartheta$ decreases, which appears in various instances. For example, for $\bfeta\sim N(\bfzero,\sigma^2 \bfI)$ or uniformly bounded $\ell^2$ noise, we have $\vartheta\lesssim \sigma \sqrt{\log(m)}/\sqrt m$ and $\vartheta\lesssim \|\bfeta\|_2/\sqrt m$, respectively.  
	
	\begin{table}[ht]
		\centering
		\begin{tabular}{|c|c|} \hline
			Function &Time (seconds) \\ \hline
			MATLAB \texttt{svds} &0.1505 \\ \hline
			Gradient-MUSIC &0.5030 \\ \hline 
		\end{tabular} \quad 
		\begin{tabular}{|c|c|} \hline
			Function &Time (seconds) \\ \hline
			MATLAB \texttt{svds} &0.1505 \\ \hline
			Classical MUSIC &105.4096 \\ \hline 
		\end{tabular}
		\caption{Left table shows the runtime for estimating the frequencies using MATLAB's \texttt{svds} function to compute the Toeplitz estimator and Gradient-MUSIC (\cref{alg:gradMUSIC3}). Right table shows the analogous situation but with Classical MUSIC (\cref{alg:MUSIC}) instead.}
		\label{tab:runtime}
	\end{table}
	
	We perform a numerical experiment where we compare their computational complexities for white noise with variance $\sigma^2=0.01$. Gradient-MUSIC is implemented following the discussion in \cref{sec:numericsstochastic}, while classical MUSIC uses a uniform grid of width $0.1 \sigma m^{-3/2}$, which has the correct scaling in $\sigma$ and $m$ for white noise as explained in \cref{sec:classicalmusic}. We also keep track of the time used for the MATLAB's \texttt{svds} function used to compute the Toeplitz estimator. 
	
	Table \ref{tab:runtime} shows their worst-case runtime for a fixed $(\bfx,\bfa)$ over 10 realizations of white noise with variance $\sigma^2=0.01$ and $m=1000$. This was performed on a commercial desktop with a 10 core CPU. We observe that the runtime of Gradient-MUSIC is comparable with that of MATLAB's \texttt{svds}, but classical MUSIC is far slower due to its fine grid search. This gap widens as $m$ increases since classical MUSIC has complexity $O(\sigma^{-1} m^{5/2} s/\sqrt{\log(m)})$, while Gradient-MUSIC has complexity $O(m^2s+m s^2 \log( \sigma^{-1} \sqrt{m/\log(m)}))$. Due to slowness of classical MUSIC, we were unable to perform the same experiment for significantly larger $m$ in a reasonable amount of time.

	\section{Computation and numerical simulations}
	\label{sec:numerical}
	
	A software package that implements Gradient-MUSIC with parallelization and reproduces the figures in this paper can be found here.\footnote{\url{https://github.com/weilinlimath/gradientMUSIC}} 
	
	\subsection{Gradient descent termination condition}
	\label{sec:gdtermination}
	
	We wrote our main theorems for Gradient-MUSIC (\cref{thm:maingradientMUSIC,cor:deterministic,cor:stochastic}) in terms of the number of gradient iterations $n$ since that is clearer from a conceptual point of view. In practice, one usually terminates gradient descent once the number of iterations has exceed some preset number (e.g., 300 or 500) or a termination condition has been reached. A standard termination condition is to stop gradient descent at iteration $n$ once $|\tilde q \,'(t_{j,n})|\leq \epsilon m$ for some selected parameter $\epsilon$. We include $m$ in this condition because $\tilde q\, '$ naturally scales in $m$, so $\epsilon$ is a dimensionless constant. 
	
	We found this condition effective to use in practice, and Gradient-MUSIC with this termination condition is summarized in \cref{alg:gradMUSIC3}. First, it will always terminate for big enough $n$ since $\tilde x_j$ is a local minimum so $\tilde q \,'(\tilde x_j)=0$. Second, this condition is rigorously justified by slightly modifying our theory, as shown in the following result. 
	
	\begin{algorithm}[t]
		\caption{Gradient-MUSIC (a given a subspace, sparsity, and adaptive termination)}
		\begin{algorithmic}
			\Require Subspace $\tilde \bfU$. \\
			{\bf Parameters:} Finite set $G\subset\T$, threshold parameter $\alpha$, gradient step size $h$, termination parameter $\epsilon>0$, and positive integers $N_{\min}\leq N_{\max}$.
			\begin{enumerate} 
				\item 
				Evaluate the landscape function $\tilde q:=q_{\tilde\bfU}$ associated with $\tilde\bfU$ on the set $G\subset\T$ and find the accepted points
				$$
				A:=A(\alpha)=\{u\in G\colon \tilde q(u)< \alpha \}. 
				$$ 
				Find all $s$ clusters of $A$ and pick representatives $t_{1,0},\dots,t_{s,0}$ for each cluster. 
				\item 
				For each $j\in \{1,\dots,s\}$, run gradient descent with step size $h$ and initial guess $t_{j,0}$, namely
				$$
				t_{j,k+1}=t_{j,k}-h \tilde q \, '(t_{j,k}),
				$$
				and stop after $n_j$ iterations, where
				$$
				n_j:=\min \left\{ N_{\max}, \,  \min \left\{k\in\N \colon N_{\min}\leq k \text{ and } |\tilde q \,'(t_{j,k})|\leq \epsilon m\right\} \right\}.
				$$
				Let $\hat x_j=t_{j,n_j}$ and $\hat \bfx=\{\hat x_j\}_{j=1}^s$. 
			\end{enumerate}
			\Ensure Estimated frequencies $\hat \bfx$.
		\end{algorithmic}
		\label{alg:gradMUSIC3}
	\end{algorithm}
	
	\begin{theorem}[Gradient-MUSIC performance guarantee with termination condition]
		\label{thm:maingradMUSIC2}
		Let $m\geq 100$, $\alpha = 0.529$, $G\subset \T$ be finite such that $\mesh(G)\leq 1/(2m)$, $h=6/m^2$, $\epsilon>0$, and $31\leq N_{\min}\leq N_{\max}$. For any $\bfx=\{x_j\}_{j=1}^s\subset \T$ and $\tilde\bfU\in \bbU^{m\times s}$ such that $\Delta(\bfx)\geq 8\pi/m$ and $\vartheta:=\vartheta(\bfx,\tilde\bfU)\leq 0.01$, Gradient-MUSIC (\cref{alg:gradMUSIC3}) outputs $\hat \bfx$ such that the frequency error satisfies
		\begin{align*}
			\max_{j=1,\dots,s} |x_j-\hat x_j|
			&\leq \frac{7\vartheta} m + \frac{37\epsilon}{m}.
		\end{align*}
	\end{theorem}
	
	This theorem is proved in \cref{proof:maingradMUSIC2}. It can be used instead of \cref{thm:maingradientMUSIC} to derive similar results like \cref{cor:deterministic,cor:stochastic} except that $\epsilon$ has to be chosen correctly instead of $n$. We omit theses computations which can be carried out similar to the proofs of those theorems. For $\bfeta\in\ell^p$, we select
	$$
	\epsilon\lesssim \frac{\|\bfeta\|_p}{a_{\min} m^{1/p}}. 
	$$
	For $\bfeta\sim \calN(\bfzero,\bfSigma)$ with diagonal $\bfSigma$, we select
	$$
	\epsilon
	\lesssim \frac{\sqrt{\tr(\bfSigma)\log(m)}}{a_{\min} \, m}. 
	$$

	\subsection{Verification of \cref{cor:stochastic}}
	\label{sec:numericsstochastic}
	
	This section numerically verifies the predictions made by \cref{cor:stochastic} for $\bfeta\sim \calN(\bfzero,\bfSigma)$, where $\Sigma_{k,k}=\sigma^2 (1+|k|)^{2r}$ for $\sigma=0.1$ and $r\in \{-1/4,0,1/4\}$. For this experiment, we fix $(\bfx,\bfa)$, pick 5 uniformly log-spaced integers in $m\in \{10^2,\dots,10^4\}$, and draw 50 realizations of $\bfeta$. 
	
	We compute the error made by Gradient-MUSIC (\cref{alg:gradMUSIC2}) over these trials, except we use the more practical version in \cref{alg:gradMUSIC3} that uses a derivative termination condition, in place of the more theoretically appealing \cref{alg:gradMUSIC}. In this experiment, we pick $\epsilon=0.01/m$, $N_{\min}=31$, and $N_{\max}=300$. According to the discussion following \cref{thm:maingradMUSIC2}, we could have picked a much bigger $\epsilon$ while still achieving the same approximation rate, but we have decided to be overly cautious. This is more realistic since one would select smaller $\epsilon$ than necessary in practice.  
	
	\begin{figure}[ht]
		\centering
		\includegraphics[width=0.4\linewidth]{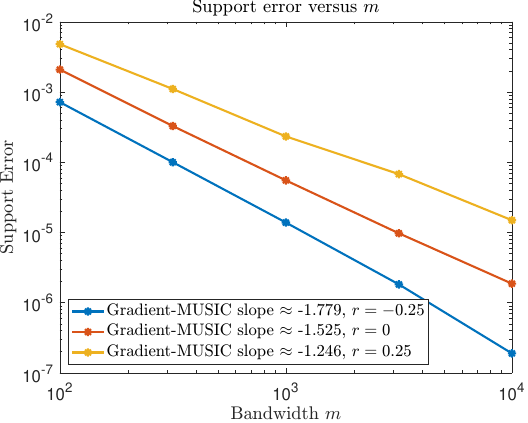} \quad 
		\includegraphics[width=0.4\linewidth]{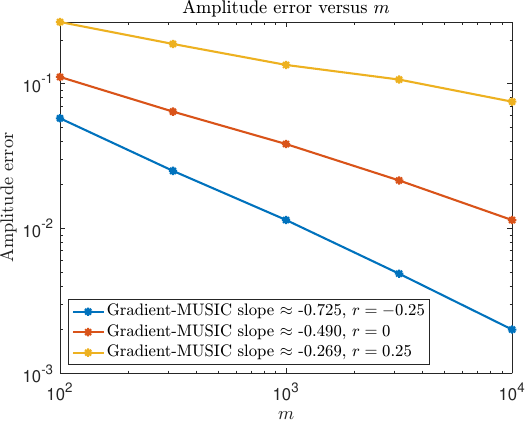}
		\caption{Gradient-MUSIC frequency and amplitude errors for nonstationary independent Guassian noise.}
		\label{fig:stochastic}
	\end{figure}
	
	Figure \ref{fig:stochastic} plots the 90\% percentile frequency and amplitude errors for each $m$ and $r$. The slopes displayed in the figure legends are computed through a least squares fit of the data points. These results are consistent with our theory which predicts that for each $|r|<1/2$ and for a fixed probability of success, the frequency and support errors scale like $m^{-3/2+r}$ and $m^{1/2+r}$ as $m\to\infty$, omitting $\log(m)$ factors.

	\section{Extensions of the theory}
	\label{sec:extensions} 
	
	\subsection{Improvement for real amplitudes}
	
	Suppose the amplitudes $\bfa$ of $h$ in \eqref{eq:hfun} are real, instead of complex. In this section only, assume that the samples are 
	\begin{equation}
		\label{eq:modelrealamp}
		\tilde\bfy = \bfy + \bfeta, \wherespace \bfy=(h(k))_{k=0,\dots,2m-1} \andspace \bfeta\in\C^{2m}.     
	\end{equation}
	We index the samples by $\{0,\dots,2m-1\}$ instead of $\{-m+1,\dots,m-1\}$ since the former is more natural for real amplitudes. We assume we are given $2m$ measurements instead of $2m-1$ to make some resulting expressions cleaner, but of course an additional sample makes no fundamental difference. 
	
	Since $\bfa$ are real, the Fourier transform has natural symmetry so that $h(-k) = \overline{h(k)}$. Thus, we extend $\bfeta$ to a $\bfzeta\in \C^{4m-1}$ such that $\zeta_{-k}=\overline{\eta_k}$ and $\zeta_k =\eta_k$ for each $k=0,\dots,2m-1$. Then we have the new model equation,
	$$
	\tilde\bfy 
	= \bfy + \bfzeta, \wherespace \bfy=(h(k))_{k=-2m+1,\dots,2m-1}.
	$$
	
	This is now the same model as \eqref{eq:modelrealamp} expect we have turned $2m$ noisy measurements to $4m-1$ many. To process $\tilde\bfy$ in the same manner as before, we form the Toeplitz matrix
	$$
	T(\tilde\bfy)
	:=\begin{bmatrix}
		\tilde y_0 & \tilde y_{-1} &\cdots & \tilde y_{-2m+1}\\
		\tilde y_1 & \tilde y_0 & & \tilde y_{-2m+1} \\
		\vdots & &\ddots &\vdots \\
		\tilde y_{2m-1} & \tilde y_{2m-2} &\cdots & \tilde y_0 
	\end{bmatrix}\in \C^{2m\times 2m}.
	$$
	We define $T(\bfy), T(\bfzeta)\in \C^{2m\times 2m}$ in the same manner. Then $T(\bfy)$ enjoys factorization \eqref{eq:toeplitzfactorization} except $\bfPhi(m,\bfx)$ is replaced with $\bfPhi(2m,\bfx)$. So the same setup carries over except we essentially doubled the number of measurements for free by exploiting that $\bfa$ is real. 
	
	We are now ready to apply the same machinery. Notice that $\|\bfzeta\|_p^p \leq 2\|\bfeta\|_p^p$, so this reflection only increases the $\ell^p$ norm by a multiplicative factor of $2^{1/p}\leq 2$. Consequently \cref{lem:noisepnorm} provides us with the bound 
	$$
	\|T(\bfzeta)\|_p\leq 2 (2m)^{1-1/p} \|\bfzeta\|_p
	\leq 8 m^{1-1/p} \|\bfeta\|_p. 
	$$
	If $\bfeta\in \calN(\bfzero,\bfSigma)$, then we cannot ensure that the entries of $\bfzeta$ are independent. However, we do not need this to be case, since our theory just requires an upper bound for $\|T(\bfzeta)\|_2$. The following provides a suitable adaptation of \cref{lem:toeplitznorm}, which is proved in \cref{proof:toeplitznorm2}.
	
	\begin{lemma}
		\label{lem:toeplitznorm2}
		Let $\bfSigma\in \R^{(2m-1)\times (2m-1)}$ be a diagonal matrix with positive diagonals. If $\bfeta\sim \calN(\bfzero,\bfSigma)$ and $\bfzeta\in \C^{4m-1}$ is defined as $\zeta_{-k}=\overline{\eta_k}$ and $\zeta_k =\eta_k$ for each $k\geq 0$, then the random matrix $T(\bfzeta)\in \C^{2m\times 2m}$ satisfies 
		\begin{align*}
			\E \|T(\bfzeta)\|_2 
			&\le \textstyle \sqrt{8 \tr(\bfSigma)\log(4m) }, \\
			\P\left\{\|T(\bfzeta)\|_2 \ge t \right\} &\le 8m e^{-t^2/(8 \tr(\bfSigma))} \forallspace t\geq  0.
		\end{align*}	
	\end{lemma}
	
	With these ingredients at hand, all of our theorems generalize to model \eqref{eq:modelrealamp}, with possibly different constants. The main improvement we would like to highlight is that, instead of assuming $\Delta(\bfx)\geq 8\pi/m$ or $(\bfx,\bfa)\in \calS(8\pi/m,r_0,10r_0)$, they can be relaxed to $\Delta(\bfx)\geq 4\pi/m$ or $(\bfx,\bfa)\in \calS(4\pi/m,r_0,10r_0)$ respectively, since we have effectively doubled the number of measurements without actually acquiring more. 
	
	\subsection{Different subspace estimators, beyond Toeplitz matrices}
	\label{sec:differentsubspace}
	
	Recall that the results in \cref{sec:main} are approximation results that do not make any assumptions on how a $\tilde\bfU\in \bbU^{m\times s}$ is obtained, only that it is a reasonably good approximation to some true Fourier subspace $\bfU\in \bbF^{m\times s}$. This generality has several useful advantages, and we discuss three that we believe are immediately relevant. 
	
	One advantage is compatibility with randomized numerical linear algebra techniques. For example, we do not necessarily have to use the Toeplitz estimator $\tilde\bfU$ in the first place. It is computationally and memory expensive to compute, since one needs to store the large Toeplitz matrix $\tilde\bfT$ in memory and calculate its leading $s$ dimensional left singular space. We can instead use a randomized method \cite{halko2011finding}, which trades some precision for speed.

	A basic illustration of this is to draw a random matrix $\bfG\in \C^{m\times k}$ where $k\ll m$ and the entries of $\bfG$ are iid complex normal random variables. The matrix product $\tilde \bfT \bfG$ requires $O(m^2 k)$ operations,  but $\tilde \bfT \bfG$ requires less memory to store compared to $\tilde\bfT$. A QR factorization of $\tilde \bfT \bfG$ gives rise its column space $\bfQ\in \C^{m \times k}$. We then compute $\bfQ^* \tilde\bfT$ and its leading $s$ dimensional left singular space $\bfB \in \C^{k\times s}$.
	The new estimator $\hat\bfU$ for $\bfU$ is $\bfQ \bfB$, which is not necessarily the same as the Toeplitz estimator $\tilde\bfU$. Bounding the sine-theta distance between $\hat\bfU$ and $\bfU$, we get a sufficient condition on the noise $\bfeta$ such that $\vartheta(\hat\bfU,\bfU)\leq 0.01$ with high probability. Then the rest of the framework goes through except the first step of \cref{alg:gradMUSIC2} uses this random estimator $\hat\bfU$ instead. 
	
	A second advantage is that we can forgo the Toeplitz estimator and its variations altogether. It would be surprising if Toeplitz estimator is optimal for all types of deterministic and stochastic noise. For example, deterministic $\bfeta$ with some algebraic relationship between its entries, or $\bfeta\in\calN(\bfmu,\bfSigma)$ for $\bfmu\not=\bf0$ and non-diagonal $\bfSigma$. Our framework allows for any approximation $\tilde\bfU$ of $\bfU$, provided they are close enough.

	\section{Landscape analysis}
	\label{sec:landscape}

	\subsection{Organization}
	
	Several basic properties of the landscape function were proved in \cref{sec:abstract}. The general setup of this section is that $q:=q_{\bfU}$ for a $\bfx\simeq_m\bfU\in \bbF^{m\times s}$ that satisfies $\Delta(\bfx)\geq 2\pi \beta/m$ for large enough $\beta>1$, while $\tilde q:=q_{\tilde\bfU}\in \bbU^{m\times s}$, where $\vartheta:=\vartheta(\bfx,\tilde\bfU)$ is assumed small enough. Here, $\tilde\bfU$ is an arbitrary subspace that is close enough to $\bfU$ and no further assumptions are made.
	
	Our main goal is to answer two questions. In addition to the properties listed in \cref{prop:qroots}, what other ones does $q_{\bfU}$ possess? What properties does $q_{\tilde\bfU}$ enjoy, beyond the basic perturbation bound given in \cref{lem:landscape2}? We develop several technical tools in the next two subsections. Then we answer these two questions.

	\subsection{Energy estimates for Dirichlet kernels}
	
	Recall the index set $I(m)$ defined in \eqref{eq:defindex}. We define a normalized and modified Dirichlet kernel, 
	\begin{equation}
		\label{eq:dirichletkernel}
		d_m(t)
		:= \frac 1 m \sum_{k\in I(m)} e^{ikt}
		=e^{-i(m-1)/2} \, \frac 1 m \sum_{k=0}^{m-1} e^{ikt}
		= \frac 1 m \frac{\sin(mt/2)}{\sin(t/2)}. 
	\end{equation}
	Note the right hand side has a removable discontinuity at $t=0$. We have $d_m(0)=1$ and $\|d_m\|_{L^\infty(\T)}\leq 1$. Recall that $f_m$ denotes the normalized Fej\'er kernel, defined in \eqref{eq:fejerkernel}. We have
	\begin{equation*}
		f_m(t)
		= (d_m(t))^2
		=\frac 1 {m^2} \left( \frac{\sin(mt/2)}{\sin(t/2)} \right)^2. 
	\end{equation*}
	The Dirichlet and Fej\'er kernels (and related functions) naturally appear in our analysis because
	\begin{equation}
		\label{eq:dirichlet1}
		\bfphi^{(k)}(u)^*\bfphi^{(\ell)}(t)
		=\frac 1 m\sum_{j\in I(m)} (-ij)^k(ij)^\ell e^{i j (t-u)}
		=(-1)^k d_m^{(k+\ell)}(t-u). 
	\end{equation}
	One particular consequence of this is that for all $t\in \T$, we have
	\begin{equation}
		\label{eq:dirichlet2}
		\bfphi(t)^*\bfphi'(t)
		=d_m'(0)
		=0. 
	\end{equation}

	We will derive several estimates. As usual $\bfx$ is a finite set that satisfies certain separation conditions, and for each $\ell\in \{0,1,2\}$, we would like suitable control over a discrete energy term, 
	$$
	\sum_{x\in \bfx} \big|d_m^{(\ell)}(x) \big|^2. 
	$$ 
	Using the trivial upper bound $\|d_m\|_{L^\infty(\T)}\leq 1$ and Bernstein's inequality, we see that this quantity is upper bounded by $m^{2\ell}s$. The $m^{2\ell}$ term is a natural scaling factor, but the $s$ dependence makes this trivial estimate unsuitable for the purposes of this paper. To obtain a refined estimate, we start with the following abstract lemma.

	\begin{lemma}
		\label{lem:energy1}
		Suppose for some $a,b>0$, the set $\bfx\subset \T$ has cardinality $s$, $\Delta(\bfx)\geq b$, and $|x|\geq a$ for all $x\in\bfx$. For any extended real valued function $h$ that is non-negative, even, and  non-increasing away from zero, we have
		$$
		\sum_{x\in \bfx} h(x)
		\leq \sum_{j=0}^{\lfloor s/2\rfloor -1} h(a+bj) + \sum_{j=0}^{\lceil s/2\rceil-1} h(-a-bj)
		$$
	\end{lemma}
	
	\begin{proof}
		For convenience, define the energy of $\bfx$ as
		$$
		E(\bfx):= \sum_{x\in \bfx} h(x).
		$$
		To prove the lemma, we show there is a sequence of transformations of $\bfx$ that do not reduce the energy and terminates after transforming $\bfx$ to 
		$$
		\bfx_*:=\left\{a,a+b,\dots,a+(\lfloor \tfrac s2\rfloor-1)b \right\}\cup \left\{-a,a-b,\dots,a-(\lceil \tfrac s2\rceil-1)b \right\}.
		$$
		Since $h$ is non-negative, even, and non-increasing away from zero, we see that $E$ does not decrease if any element of $\bfx$ is shifted closer to zero while the other ones are fixed. We shift the smallest positive element (which is assumed to be at least $a$) of $\bfx$ to $a$, then the next smallest positive element  (which is necessarily at least $a+b$ due to the separation condition) to $a+b$, etc. We do the same for negative elements of $\bfx$. Letting $s_+$ and $s_-$ be the number of elements in $\bfx$ that are positive and negative respectively, this process transforms $\bfx$ to the set
		$$
		\bfx_1 :=  \left\{a,a+b,\dots,a+(s_+-1)b \right\}\cup \left\{-a,a-b,\dots,a-(s_- -1)b \right\}.
		$$
		From the above considerations, we see that $E(\bfx)\leq E(\bfx_1)$. If $s_+=\lfloor s/2 \rfloor$, then $\bfx_1=\bfx_*$. In which case, we are done since $h$ is even and 
		$$
		E(\bfx)
		\leq E(\bfx_*)
		=\sum_{j=0}^{\lfloor s/2\rfloor -1} h(a+bj) + \sum_{j=0}^{\lceil s/2\rceil-1} h(-a-bj).
		$$
		If $s_+\not=\lfloor s/2 \rfloor$, we first consider the case where $r:=s_+-s_-\geq 1$. Then we reflect the $r$ most positive elements of $\bfx_1$, which does not change the energy since $h$ is symmetric about the origin. This provides us with $\bfx_2$. Then repeating the same argument as before, we shift these $r$ reflected elements closer to zero to get $\bfx_*$ and each transformation does not decrease $E$. This shows that $E(\bfx)\leq E(\bfx_*)$, which completes the proof when $s_+-s_1\geq 1$. The remaining case is analogous.  
	\end{proof}
	
	An important property of this lemma is that while $\sum_{x\in \bfx} h(x)$ depends on $\bfx$, the right hand side $\sum_{j=0}^{s-1} h(a+bj)$ only depends on $a$, $b$ and $s$. Hence, it is a uniform bound over the class of all $\bfx$ that satisfy the assumptions of \cref{lem:energy1}. This flexibility enables us to prove energy estimates for a variety of sets. To simplify the resulting  notation, we define the following extended  real valued functions 
	\begin{align*}
		h_{m,0}(t)
		&:= \left(\frac{1}{m|\sin(t/2)|}\right)^2, \\
		h_{m,1}(t)
		&:= \left(\frac{1}{2m|\sin(t/2)|}+\frac{1}{2m^2|\sin(t/2)|^2}\right)^2, \\
		h_{m,2}(t)
		&:= \left( \frac{1}{4m|\sin(t/2)|}+\frac 1{2m^2|\sin(t/2)|}+\frac{1}{2m^3 |\sin(t/2)|^3} \right)^2, \\
		h_{m,3}(t)
		&:= \left(\frac{1}{8m|\sin(t/2)|}+\frac 3{8m^2|\sin(t/2)|}+\frac{3}{4m^3 |\sin(t/2)|^3} + \frac{5}{4m^4 |\sin(t/2)|^4}\right)^2. 
	\end{align*}
	These functions appear when we compute $d_m^{(\ell)}$ for $\ell\in \{0,\dots,3\}$ and use triangle inequality (details omitted) to see that
	\begin{equation}
		\label{eq:hupper}
		\frac{1}{m^{2\ell}} \big|d_m^{(\ell)}(t)\big|^2
		\leq h_{m,\ell}(t). 
	\end{equation}
	With these functions at hand, we define the quantity, 
	\begin{equation}
		\label{eq:energy}
		E_\ell(m,\alpha,\beta)
		:= 2h_{m,\ell}\left(\frac{2\pi \alpha}{m}\right) + \frac{m}{\pi \beta} \int_{2\pi (\alpha+\beta/2)/m}^\pi h_{m,\ell}(t)\, dt.
	\end{equation}

	\begin{lemma}
		\label{lem:dirichletenergy}
		Let $m\geq m_0\geq 1$, $\beta>1$ and $\alpha>0$. Suppose $\bfx\subset\T$ such that $\Delta(\bfx)\geq 2\pi \beta/m$ and $|x|\geq 2\pi \alpha/m$ for all $x\in \bfx$. Then for each $\ell\in \{0,1,2,3\}$, we have
		$$
		\sum_{x\in \bfx} \big|d_m^{(\ell)}(x)\big|^2
		\leq E_\ell(m_0,\alpha,\beta) \, m^{2 \ell}. 
		$$
	\end{lemma}
	
	\begin{proof}	
		Note that $h_{m,\ell}$ is even, decreasing away from zero, and non-negative. A calculus argument shows that for any $u\in(0,\pi]$, the sequence $\{m\sin(u/m)\}_{m=1}^\infty$ is increasing in $m$. By the assumption that $m\geq m_0$ and that $h_{m,\ell}$ is even, we see that for all $u \in [-\pi,\pi]$,  
		$$
		h_{m,\ell}(u/m)\leq h_{m_0,\ell}(u/m_0).
		$$
		We use this inequality, \eqref{eq:hupper}, and \cref{lem:energy1}, where $a=2\pi \alpha/m$ and $b=2\pi \beta/m$. This yields the inequality
		\begin{align*}
			&\frac{1}{m^{2\ell}} \sum_{x\in \bfx} \big|d_m^{(\ell)}(t)\big|^2
			\leq \sum_{j=0}^{\lfloor s/2\rfloor -1} h_{m_0, \ell}\left(\frac{2\pi (\alpha+\beta j)}{m_0}\right)+\sum_{j=0}^{\lceil s/2\rceil -1} h_{m_0, \ell}\left(-\frac{2\pi (\alpha+\beta j)}{m_0}\right) \\
			&\qquad = 2 h_{m_0, \ell}\left(\frac{2\pi \alpha}{m_0}\right) +\sum_{j=1}^{\lfloor s/2\rfloor -1} h_{m_0, \ell}\left(\frac{2\pi (\alpha+\beta j)}{m_0}\right)+\sum_{j=1}^{\lceil s/2\rceil -1} h_{m_0, \ell}\left(-\frac{2\pi (\alpha+\beta j)}{m_0}\right). 
		\end{align*}
		Next, we interpret the right two sums as a Riemann sum approximation of an integral using the midpoint rule. To set this up, the partition width is $2\pi \beta/m_0$ and the midpoints of this partition are 
		$$
		\left\{ \frac{2\pi (\alpha+\beta j)}{m_0}\right\}_{j=1}^{\lfloor s/2\rfloor -1}\cup \left\{ -\frac{2\pi (\alpha+\beta j)}{m_0}\right\}_{j=1}^{\lceil s/2\rceil-1}.
		$$
		Since $h_{m_0, \ell}$ is non-negative and convex, the midpoint rule underestimates the integral. Thus, we see that
		\begin{align*}
			\frac{1}{m^{2k}} \sum_{x\in \bfx} \left|d_m^{(\ell)}(t)\right|^2
			&\leq 2 h_{m_0, \ell}\left(\frac{2\pi \alpha}{m_0}\right) + \frac{m_0}{2\pi \beta}\int_{2\pi (\alpha+\beta/2)/m_0}^{2\pi-2\pi (\alpha+\beta/2)/m_0} h_{m_0, \ell}(t)\, dt. 
		\end{align*}
		Using that $h_{m_0, \ell}$ is even to manipulate the right integral completes the proof. 
	\end{proof}
	
	\cref{lem:dirichletenergy} is written in a way that suggests $E_\ell(m,\alpha,\beta)$ can be treated as a constant that does not depend on $m$. This is indeed the case. Notice that $h_{m,\ell}(t)$ has a singularity at $t=0$, which determines the behavior of the integral in \eqref{eq:energy}. Using the asymptotic expansion $\sin(t)\sim t$ for small $t$, we see that $h_{m,\ell}(t)$ behaves likes $C/(m^2|t|^2)$. Thus, $E_\ell$ is at most $C/\alpha^2+C/(\beta(\alpha+\beta/2))$ where $C$ does not depend on $m$.  
	
	\subsection{A local approximation result}
	
	To control $q$ and its derivatives, we will derive a useful representation for $\bfU\bfU^*$ where $\bfU\in \bbF^{m\times s}$. This is motivated by the following calculation. 
	
	\begin{remark} \label{rem:singleton}
		When $\bfx=\{x_1\}$, the matrix $\bfPhi$ is just a single column, $\sqrt m\, \bfphi(x_1)$. Selecting $\bfU$ as just $\bfphi(x_1)$ and using formula \eqref{eq:dirichlet1}, we have
		$$
		q(t)
		=1-\bfphi(t)^*\bfU\bfU^*\bfphi(t)
		= 1 - (d_m(t))^2
		= 1 - f_m(t).
		$$
	\end{remark}
	
	Of course \cref{rem:singleton} no longer holds in the nontrivial case where $\bfx$ is not a singleton.  Nonetheless, for any $x_j\in \bfx$, we define the matrix
	\begin{equation}
		\label{eq:Psi}
		\bPsi_j:=\bfPhi(m,\bfx\setminus x_j) \in \C^{m\times (s-1)}.
	\end{equation}
	This is precisely $\bfPhi$ with the column $\sqrt m \, \bfphi(x_j)$ removed. Throughout, we will let $\bfU_j$ be a matrix whose columns form an orthonormal basis for the range of $\bPsi_j$. The reduced singular value decomposition of $\bPsi_j$ is denoted
	\begin{equation}
		\label{eq:Psisvd}
		\bPsi_j=\bfU_j\bfS_j\bfV_j^*.
	\end{equation}
	
	Throughout the landscape analysis, various common expressions involving $\beta$ will appear, such as 
	\begin{equation}
		\label{eq:defAbeta}
		A(\beta):=\frac{\beta}{\beta-1}.
	\end{equation}
	We see that $A(\beta)>1$ and $A(\beta)$ decreases to 1 as $\beta$ increases. This quantity appears whenever we use \cref{prop:AB19}, which provides the inequality $\sigma_{\min}^{-2}(\bfPhi(m,\bfx))\leq A(\beta)/m$ whenever its assumptions hold. To reduce clutter, we define the special function 
	\begin{equation}
		\label{eq:gammasin}
		\gamma(t)=2 \sin\left( \frac 12 \sin^{-1}(t)\right). 
	\end{equation}
	When $t$ is small, this is approximated by the identity function. The following lemma relates the projection operators $\bfU\bfU^*$ and $\bfU_j\bfU_j^*$. 
	
	\begin{lemma}
		\label{lem:qsintheta} 
		Let $m\geq \max\{m_0,s+1\}$ and $\beta>1$ such that $A(\beta) E_0(m_0,\beta,\beta)<1/4$. For any $\bfx=\{x_j\}_{j=1}^s$ such that $\Delta(\bfx)\geq 2\pi \beta/m$ and for any $x_j\in \bfx$, there is a matrix $\bfW_j\in \bbU^{m\times (s-1)}$ such that
		\begin{align*}
			\bfU\bfU^* &= \bfphi(x_j)\bfphi(x_j)^* + \bfW_j\bfW_j^*, \\
			\| \bfW_j-\bfU_j\|_2
			&\leq \gamma\left(\sqrt{A(\beta) E_0(m_0,\beta,\beta)}\right). 
		\end{align*}
		Additionally, for any vector $\bfw$, we have
		\begin{equation*}
			\left\|\bfU_j^* \bfw\right\|_2 
			\leq \sqrt{ \frac{A(\beta)} m}  \, \left\|\bPsi_j^* \bfw\right\|_2.
		\end{equation*}	  
	\end{lemma}
	
	\begin{proof}
		Recall that $\bfU$ is an orthonormal basis for the range of $\bfPhi$, and the later matrix contains $\sqrt m\, \bfphi(x_j)$ as one of its columns. Consider the matrix
		$$
		(\bfI-\bfphi(x_j)\bfphi(x_j)^*) \bPsi_j
		=\bPsi_j-\bfphi(x_j)\bfphi(x_j)^* \bPsi_j,
		$$
		which is necessarily injective, otherwise neither is $\bfPhi$. Let $\bfW_j$ be an orthonormal basis for the range of $(\bfI-\bfphi(x_j)\bfphi(x_j)^*) \bPsi_j$. This now establishes the formula,
		$$
		\bfU\bfU^* = \bfphi(x_j)\bfphi(x_j)^* + \bfW_j\bfW_j^*. 
		$$
		
		To relate $\bfW_j$ and $\bfU_j$, we proceed to view $(I-\bfphi(x_j)\bfphi(x_j)^*) \bPsi_j$ as a perturbation of $\bPsi_j$. We need to do some work before we are able to use Wedin's sine-theta theorem. Since $\bPsi_j=\bfPhi(m,\bfx\setminus x_j)$ and $\Delta(\bfx\setminus x_j)\geq \Delta(\bfx)\geq 2\pi\beta/m$, we can use \cref{prop:AB19} to obtain
		\begin{equation}
			\label{eq:psihelp1}
			\frac 1{\sigma_{\min}(\bPsi_j)}
			\leq \sqrt{\frac{A(\beta)}m}.
		\end{equation}
		Our next step is to control the perturbation size. We use that $\|\bfphi(x_j)\|_2=1$ and identity \eqref{eq:dirichlet1} to see that
		\begin{align*}
			\big\|\bPsi_j-(\bfI-\bfphi(x_j)\bfphi(x_j)^*) \bPsi_j\big\|_2
			&= \big\|\bfphi(x_j)\bfphi(x_j)^* \bPsi_j\big\|_2 \\
			&= \big\|\bfphi(x_j)\big\|_2 \big\|\bPsi_j^*\bfphi(x_j)\big\|_2
			= \sqrt{m \sum_{k\not=j} \big|d_m(x_k-x_j)\big|^2}. 
		\end{align*}
		To upper bound the right hand side, since $\Delta(\bfx)\geq 2\pi \beta/m$ and $x_j\in \bfx$, we have $|x_k-x_j|\geq 2 \pi \beta/m$ for all $k\not=j$. Applying \cref{lem:dirichletenergy} where $\alpha =\beta$ and combining it with the previous inequality, we see that
		\begin{equation}
			\label{eq:psihelp2}
			\|\bPsi_j-(\bfI-\bfphi(x_j)\bfphi(x_j)^*) \bPsi_j\|_2
			\leq \sqrt{m E_0(m_0,\beta,\beta)}.
		\end{equation}
		
		Using assumption $A(\beta) E_0(m_0,\beta,\beta)< 1/4$, \eqref{eq:psihelp2}, and \eqref{eq:psihelp1}, we see that
		$$
		\|\bPsi_j-(\bfI-\bfphi(x_j)\bfphi(x_j)^*) \bPsi_j\|_2
		< \frac 12 \sigma_{\min}(\bPsi_j).
		$$
		This enables us to use Wedin's sine-theta theorem \cite[Chapter V, Theorem 4.4]{stewart1990matrix}, where $\bPsi_j$ acts as the unperturbed matrix whose range is $\bfU_j$ and $(\bfI-\bfphi(x_j)\bfphi(x_j)^*) \bPsi_j$ is the perturbed matrix whose range is $\bfW_j$. Employing inequalities \eqref{eq:psihelp1} and \eqref{eq:psihelp2}, we obtain
		\begin{align*}
			\|\bfU_j\bfU_j^*-\bfW_j\bfW_j^*\|_2 
			\leq \frac{\|\bPsi_j-(\bfI-\bfphi(x_j)\bfphi(x_j)^*) \bPsi_j\|_2}{\sigma_{\min}(\bPsi_j)}
			\leq \sqrt{A(\beta) E_0(m_0,\beta,\beta)}. 
		\end{align*}
		Here, we bounded the sine-theta distance between $\bfU_j$ and $\bfW_j$. We can find a particular choices for $\bfU_j$ and $\bfW_j$ such that $\|\bfU_j-\bfW_j\|_2$ is almost equal to the sine-theta distance, see \cite[Chapter 1, Theorem 5.2]{stewart1990matrix}. The referenced result says that if $t$ is the maximum canonical angle between two subspaces, then particular bases can be chosen for them such that their squared error in spectral norm is $(1-\cos t)^2+\sin^2 t=2-2\cos t = 4 \sin^2(t/2)$. Taking the square root, setting $t=\sin^{-1}(\|\bfU_j\bfU_j^*-\bfW_j\bfW_j^*\|_2)$, and recalling the definition of $\gamma$ in \eqref{eq:gammasin}, we have
		$$
		\|\bfU_j-\bfW_j\|_2
		\leq \gamma \left(\|\bfU_j\bfU_j^*-\bfW_j\bfW_j^*\|_2\right)
		\leq \gamma\left(\sqrt{A(\beta) E_0(m_0,\beta,\beta)}\right).
		$$
		
		For the final part of the proof, recall that $\bfU_j$ is an orthonormal basis for the range of $\bPsi_j$. By \eqref{eq:Psisvd} and \eqref{eq:psihelp1}, we see that for any vector $w$, we have 
		\begin{equation*}
			\left\|\bfU_j^* \bfw\right\|_2 
			\leq \big\|\bfS_j^{-1} \big\|_2 \big\|\bPsi_j^* \bfw\big\|_2 
			\leq \sqrt{ \frac{A(\beta)} m} \, \left\|\bPsi_j^* \bfw \right\|_2.
		\end{equation*}
		This completes the proof. 
	\end{proof}
	
	Provided that the assumptions of \cref{lem:qsintheta} hold, we combine it with the definition of $q:=q_{\bfU}$ and formula \eqref{eq:dirichlet1} to see that, for any $x_j\in \bfx$ and $t\in \T$, 
	\begin{equation}\label{eq:qsintheta}
		\begin{split}
			q(t)&=1-\bfphi(t)^*\bfU\bfU^*\bfphi(t)\\
			&=1-\bfphi(t)^*\bfphi(x_j) \bfphi(x_j)^* \bfphi(t)-\bfphi(t)^*\bfW_j \bfW_j^*\bfphi(t) \\
			&=1-f_m(t-x_j)-\bfphi(t)^*\bfW_j \bfW_j^*\bfphi(t).
		\end{split}
	\end{equation}
	From here, we can use this formula to control $q$ and its derivatives for $t$ sufficiently close to $x_j$. This extends \cref{rem:singleton} to the general case where $\bfx$ is not a singleton. This is also the formula shown in equation \eqref{eq:localapprox} where $h_m(t)=\bfphi(t)^*\bfW_j \bfW_j^*\bfphi(t)$. 
	
	\subsection{Landscape function associated with a Fourier subspace}
	
	The purpose of this subsection is to investigate the behavior of $q=q_\bfU$ where $\bfU\in \bbF^{m\times s}$. We expect this function to have several special properties. Define the quantities
	\begin{equation}
		\begin{split}
			C_0(m,\beta) &:= \frac 1 6 - 2 A(\beta) E_1(m,\beta,\beta) - \frac 1 {6m^2}, \\
			C_1(m,\beta) &:= \frac 1 6 + 2 A(\beta) E_1(m,\beta,\beta). 
		\end{split}\label{eq:cmbeta}
	\end{equation}
	As $m\to\infty$ and $\beta\to\infty$, both quantities converge to 1/6. In particular, $C_0(m,\beta)$ is an increasing function in $\beta$ for fixed $m$. Whenever $m\geq 2$, we can always find $\beta$ sufficiently large so that $C_0(m,\beta)>0$. 
	
	\begin{lemma}
		\label{lem:qsecond}
		Let $m\geq \max\{m_0,s+1\}$ and $\beta>1$. For any $\bfx=\{x_j\}_{j=1}^s$ such that $\Delta(\bfx)\geq 2\pi\beta/m$, let $q:=q_\bfU$ be the landscape function associated with $\bfx\simeq_m \bfU\in \bbF^{m\times s}$. For each $x_j\in \bfx$, we have
		$$
		C_0(m_0,\beta) m^2 \leq q''(x_j) \leq C_1(m_0,\beta) m^2. 
		$$
	\end{lemma}
	
	\begin{proof}
		Fix any $x_j\in \bfx$. Using the definition of $q$, a calculation yields the equation
		$$
		q''(t)
		= 2\bfphi'(t)^* \bfU_\perp \bfU_\perp^* \bfphi'(t) + 2\Re\left(\bfphi(t)^* \bfU_\perp \bfU_\perp^* \bfphi''(t)\right).
		$$
		Since $\bfphi(x_j)$ is lies in the range of $\bfU$, we have $\bfU_\perp^* \bfphi(x_j)=0$. This provides us with the formula
		$$
		q''(x_j)
		= 2\bfphi'(x_j)^* \bfU_\perp \bfU_\perp^* \bfphi'(x_j).
		$$
		Since $\bfU$ forms an orthonormal basis for the range of $\bfPhi$, which has linearly independent columns, we have the explicit formula $\bfU\bfU^*=\bfPhi (\bfPhi^*\bfPhi)^{-1}\bfPhi^*$. Then
		\begin{equation*}
			\begin{split}
				q''(x_j)
				&= 2\bfphi'(x_j)^* \bfU_\perp \bfU_\perp^* \bfphi'(x_j) \\
				&= 2\|\bfphi'(x_j)\|^2_2 - 2\bfphi'(x_j)^* \bfU\bfU^* \bfphi'(x_j) \\
				&= 2\|\bfphi'(x_j)\|^2_2 - 2\bfphi'(x_j)^* \bfPhi (\bfPhi^*\bfPhi)^{-1}\bfPhi^* \bfphi'(x_j). 
			\end{split}
		\end{equation*}
		A calculation shows that $2\|\bfphi'(x_j)\|^2_2=(m^2-1)/6$, see \eqref{eq:phinorm1}. Thus,
		\begin{equation*}
			\left| q''(x_j)-\frac 1 6 (m^2-1)\right|
			\leq 2 \left| \bfphi'(x_j)^* \bfPhi (\bfPhi^*\bfPhi)^{-1}\bfPhi^* \bfphi'(x_j)\right|
			\leq \frac 2{\sigma_{\min}^2(\bfPhi)} \|\bfPhi^* \bfphi'(x_j)\|_2^2. 
		\end{equation*}
		To control this term, note that the columns of $\bfPhi$ are $\sqrt m \, \bfphi(x_j)$ and that $\bfphi(x_j)^*\bfphi'(x_j)=0$. We use \cref{prop:AB19} and formula \eqref{eq:dirichlet1} to obtain 
		\begin{align*}
			\frac 2{\sigma_{\min}^2(\bfPhi)} \|\bfPhi^* \bfphi'(x_j)\|_2^2
			\leq \frac{ 2A(\beta) }  m \|\bfPhi^* \bfphi'(x_j)\|_2^2
			= 2A(\beta) \sum_{k\not=j} |d_m'(x_j-x_k)|^2. 
		\end{align*}
		Due to the separation assumption $\Delta(\bfx)\geq 2\pi \beta/m$ and that $x_j\in \bfx$, we have $|x_k-x_j|\geq 2\pi \beta/m$ for all $k\not=j$. This allows us to use \cref{lem:dirichletenergy} with $\alpha=\beta$, which provides the estimate
		$$
		\sum_{k\not=j} |d_m'(x_j-x_k)|^2
		\leq E_1(m_0,\beta,\beta) \, m^2. 
		$$
		Combining the above yields
		$$
		\left| q''(x_j)-\frac 1 6 (m^2-1)\right|
		\leq 2 A(\beta) E_1(m_0,\beta,\beta) m^2. 
		$$
		Using this inequality, the assumption that $m\geq m_0$, and $q''(x_j)\geq 0$ since $x_j$ is a double root, we obtain the two conclusions of this lemma.  
	\end{proof}
	
	Let us make a few comments about \cref{lem:qsecond}, which is one of the key technical results of this paper. The lower bound for $q''(x_j)$ does not depend on $s$ and holds uniformly over all $x_j\in\bfx$. Its conclusion sharp up to a universal constant, since $\|q''\|_{L^\infty(\T)}\leq  m^2$ from inequality \eqref{eq:qnorm1}. For large enough $\beta$ and $m$, both $C_0(m,\beta)$ and $C_1(m,\beta)$ are roughly 1/6, which naturally appears since $f''_m(0)=(m^2-1)/6$. By \cref{rem:singleton}, the $1/6$ constant cannot be improved. Prior analysis of MUSIC, such as \cite{liao2016music}, recognized that $q''(x_j)$ is important for the algorithm's stability, but did not explicitly calculate $q''(x_j)$. The main theorems in that reference implicitly depend on $q''(x_j)$. 
	
	Define the quantities $B_0=1$, $B_1=1/12$, $B_2=1/80$, and $B_3=1/448$, and note their appearance in \eqref{eq:phinorm1}. For each $\ell\in \{0,1,2,3\}$, we define the constant
	\begin{equation}
		\label{eq:Tconstant}
		T_\ell(m,\alpha,\beta)
		:=\sqrt{A(\beta) E_\ell(m,\alpha,\beta)} + \gamma\left(\sqrt{A(\beta) B_\ell E_0(m,\beta,\beta)}\right).
	\end{equation}
	
	\begin{lemma}
		\label{lem:qglobal} 
		Let $m\geq \max\{m_0,s+1\}$, and $\beta>1$ such that $A(\beta) E_0(m_0,\beta,\beta)<1/4$. For any $\tau\in[0,\beta)$ and $\ell\in \{0,1,2,3\}$, set $T_\ell:=T_\ell(m_0,\beta-\tau,\beta)$. For any $\bfx=\{x_j\}_{j=1}^s$ such that $\Delta(\bfx)\geq 2\pi \beta/m$, let $q:=q_\bfU$ be the landscape function associated with $\bfx\simeq_m \bfU\in \bbF^{m\times s}$. For any $x_j\in \bfx$ and $|t-x_j|\leq 2\pi \tau/m$, we have
		\begin{align*}
			\left|q'(t)
			+f_m'( t-x_j) \right| 
			&\leq 2 \, T_0 T_1 m, \\
			\left| q''(t)+f_m''(t-x_j) \right| 
			&\leq \left(2T_1^2+2T_0T_2\right) m^2, \\
			\left| q'''(t)+f_m'''(t-x_j) \right| 
			&\leq  \left(2 T_0 T_3+6T_1T_2\right) m^3. 
		\end{align*} 
	\end{lemma}
	
	\begin{proof}
		There is nothing to prove if $\bfx$ is a singleton due to \cref{rem:singleton}. From now on, assume that $s>1$. The assumptions of \cref{lem:qsintheta} hold, so let $\bfW_j$  be the matrix defined in that lemma. We take derivatives of formula \eqref{eq:qsintheta} to see that
		\begin{align*}
			q'(t)
			&=-f_m'(t-x_j)-2\Re\left( \bfphi(t)^* \bfW_j \bfW_j^* \bfphi'(t) \right), \\ 
			q''(t)&=-f_m''(t-x_j)-2\bfphi'(t)^*\bfW_j \bfW_j^*\bfphi'(t) -2\Re\left(\bfphi(t)^*\bfW_j \bfW_j^*\bfphi''(t)\right), \\
			q'''(t)&=-f_m'''(t-x_j)-2\Re\left(\bfphi(t)^*\bfW_j \bfW_j^*\bfphi'''(t)\right) - 6\Re\left(\bfphi'(t)^*\bfW_j \bfW_j^*\bfphi''(t)\right).
		\end{align*}
		By Cauchy-Schwarz, we have 
		\begin{align*}
			|q'(t)+f_m'(t-x_j)|
			&\leq 2 \big\|\bfW_j^*\bfphi(t) \big\|_2 \big\|\bfW_j^*\bfphi'(t)\big\|_2, \\
			|q''(t)+f_m''(t-x_j)|
			&\leq 2 \big\|\bfW_j^*\bfphi'(t) \big\|_2^2 + 2 \big\|\bfW_j^*\bfphi(t) \big\|_2 \big\|\bfW_j^*\bfphi''(t)\big\|_2, \\
			|q'''(t)+f_m'''(t-x_j)|
			&\leq 2 \big\|\bfW_j^*\bfphi(t) \big\|_2\big\|\bfW_j^*\bfphi'''(t) \big\|_2 + 6 \big\|\bfW_j^*\bfphi'(t) \big\|_2 \big\|\bfW_j^*\bfphi''(t)\big\|_2.
		\end{align*}
		We need to control these error terms. Using the properties listed in \cref{lem:qsintheta} and the upper bound $\|\bfphi^{(\ell)}\|_2\leq \sqrt{B_\ell} m^\ell$ listed in \eqref{eq:phinorm1}, for any $\ell\in \{0,1,2\}$, we have
		\begin{equation*}
			\begin{split}
				\big\|\bfW_j^* \bfphi^{(\ell)}(t) \big\|_2 
				&\leq \big\|\bfU_j^* \bfphi^{(\ell)}(t) \big\|_2 + \big\|(\bfW_j-\bfU_j)^*\bfphi^{(\ell)}(t) \big\|_2 \\
				&\leq \sqrt{ \frac{A(\beta)} m} \big\|\bPsi_j^* \bfphi^{(\ell)}(t)\big\|_2 + \gamma\left( \sqrt{A(\beta) B_\ell E_0(m_0,\beta,\beta)} \right) m^\ell.
			\end{split}		
		\end{equation*}
		Using the assumptions that $\Delta(\bfx)\geq 2\pi \beta/m$ for $\beta>1$, $\tau\in [0,\beta)$, and that $|t-x_j|\leq 2\pi \tau/m$, we deduce that 
		$
		\left|x_k-t \right|
		\geq {2\pi (\beta-\tau)}/m 
		$
		for all $k\not=j$. We again use identity \eqref{eq:dirichlet1} and apply \cref{lem:dirichletenergy} where $\beta-\tau$ acts as $\alpha$ in the referenced lemma, to deduce that 
		\begin{equation*}
			\label{eq:Psihelp}
			\sqrt{ \frac{A(\beta)} m} \big\|\bPsi_j^* \bfphi^{(\ell)}(t)\big\|_2
			=\sqrt{A(\beta) \sum_{k\not=j} \big|d_m^{(\ell)}(x_k-t)\big|^2}
			\leq \sqrt{A(\beta) E_\ell(m_0,\beta-\tau,\beta)} \,  m^\ell.
		\end{equation*}
		Combining the previous inequalities, we obtain
		\begin{equation*}
			\big\|\bfW_j^* \bfphi^{(\ell)}(t) \big\|_2 
			\leq T_\ell \, m^\ell.
		\end{equation*} 
		Combining everything completes the proof.
	\end{proof}
	
	\cref{lem:qglobal} shows that $q^{(\ell)}(t)$ is pointwise approximated by $f^{(\ell)}(t-x_j)$ whenever $|t-x_j|$ is sufficiently small. The error terms are quadratic in $T_\ell$, e.g., $T_0T_1$, $T_1^2$, $T_0 T_2$, etc. If $\beta-\tau\asymp \beta$, then from the discussion following proof of \cref{lem:energy1}, we see that $T_\ell=O(1/\beta)$ for  each $\ell$ and so all constants in \cref{lem:qglobal} are $O(1/\beta^2)$. The fast decay of the error terms will allow us to pick a $\beta$ not terribly large.    
	
	The next result provides a lower bound for $q(t)$ whenever $t$ is sufficiently far away from $\bfx$. 
	\begin{lemma}
		\label{lem:qfar}
		Let $m\geq \max\{m_0,s+1\}$, and $\beta>1$ such that $A(\beta) E_0(m_0,\beta,\beta)<1/4$. For any $\bfx=\{x_j\}_{j=1}^s$ such that $\Delta(\bfx)\geq 2\pi \beta/m$, let $q:=q_\bfU$ be the landscape function associated with $\bfx\simeq_m \bfU\in \bbF^{m\times s}$. For any $t\in \T$ such that $|t-x_j|\geq \pi/m$ for all $x_j\in \bfx$, we have
		$$
		q(t)
		\geq 
		1- f_m(t-x_j) - (T_0(m_0,\beta/2,\beta))^2. 
		$$
	\end{lemma}
	
	\begin{proof}
		There is nothing to prove if $\bfx$ is a singleton due to \cref{rem:singleton}. From here onward, assume that $s>1$. Let $x_1$ and $x_2$ be the elements in $\bfx$ that are closest to $t$ such that $t\in [x_1,x_2]$. Without loss of generality, we assume that $|t-x_1|\leq |t- x_2|$. Let $\beta_*$ such that $|x_2-x_1|=2\pi\beta_*/m$ and note that $\beta_*\geq \beta$,  otherwise it would contract the assumption that $\Delta(\bfx)\geq 2\pi \beta/m$. We also have $|t-x_2|\geq \pi\beta_*/m$, otherwise $x_2$ would be closer to $t$. 
		
		The assumptions of \cref{lem:qsintheta} are satisfied, so let $\bfW_1$ be the matrix from that lemma. Recalling \eqref{eq:qsintheta}, we use triangle inequality and Cauchy-Schwarz to obtain 
		\begin{align*}
			q(t)
			&\geq 1- f_m(t-x_1) - \|\bfW_1^* \bfphi(t)\|_2^2.
		\end{align*}
		Using the properties of $\bfW_1$, we obtain 
		$$
		\big\|\bfW_1^* \bfphi(t)\big\|_2
		\leq \sqrt{\frac {A(\beta)}m} \, \big\|\bPsi_1^*\bfphi(t)\big\|_2 + \gamma\left(\sqrt{A(\beta) E_0(m,\beta,\beta)}\right). 
		$$
		
		We proceed to control the right hand side. Recall that $\Delta(\bfx)\geq 2\pi \beta/m$ by assumption, while $x_1$ and $x_2$ are the elements closest to $t$ from the left and right respectively. We also have $|x_1-t|\geq \pi/m$ and $|x_2-t|\geq\pi\beta_*/m\geq \pi\beta/m$. This implies $|x_k-t|\geq \pi \beta/m$ for all $k\not=1$.  We use equation \eqref{eq:dirichlet1} and \cref{lem:dirichletenergy} where $\alpha=\beta/2$ to obtain
		\begin{align*}
			\sqrt{\frac {A(\beta)}m} \big\|\bPsi_1^*\bfphi(t)\big\|_2
			= \sqrt{A(\beta)\sum_{k\not=1} \left|d_m(t-x_k)\right|^2}
			&\leq \sqrt{A(\beta)E_0(m_0,\beta/2,\beta)}. 
		\end{align*}
		Combining these inequalities and recalling the definition of $T_0$ completes the proof.
	\end{proof}
	
	\subsection{Landscape function associated with a perturbed Fourier subspace}
	
	This section provides a landscape analysis of $\tilde q:=q_{\tilde \bfU}$ where $\tilde\bfU\in \bbU^{m\times s}$ is a good enough approximation to a $\bfU\in \bbF^{m\times s}$. Let $q:=q_{\bfU}$. Recall \cref{lem:landscape2} which contained some basic inequalities.  
	
	While $\bfx\simeq_m \bfU$ are the only roots of $q$, each one being a double root, there is no reason to believe that $\tilde q$ has any roots whatsoever. However, we will show that $\tilde q$ has critical points $\tilde\bfx$ that are near $\bfx$. In essence, the double roots of $q$ are perturbed to local minima of $\tilde q$. Recall the quantity $T_\ell$ defined in \eqref{eq:Tconstant}. 
	
	\begin{lemma}
		\label{lem:localperturbation}
		Let $m\geq \max\{m_0,s+1\}$, $\beta>1$, and $T_\ell:=T_\ell(m_0,\beta-1/(2\pi),\beta)$ for $\ell\in \{0,1,2,3\}$. For any $\bfx=\{x_j\}_{j=1}^s\simeq_m\bfU$ such that $\Delta(\bfx)\geq 2\pi\beta/m$ and $\tilde\bfU\in \bbU^{m\times s}$, let $\tilde q =q_{\tilde\bfU}$ and $\vartheta:=\vartheta(\bfx,\tilde\bfU)$. Suppose there exists a $r\in (0,1/\vartheta)$ such that 
		\begin{equation}
			\label{eq:rcondition}
			C_0(m_0,\beta) r -1 - \left(\frac 1 {20} + T_0 T_3+3T_1T_2\right) r^2 \vartheta  > 0. 
		\end{equation}
		Then the polynomial $\tilde q$ has $s$ critical points $\tilde \bfx=\{\tilde x_1,\dots,\tilde x_s\}$ such that 
		$$
		\max_{j=1,\dots,s} |\tilde x_j-x_j|\leq \frac{r\vartheta}{m}. 
		$$
	\end{lemma}
	
	\begin{proof}
		Let $q=q_\bfU$ and fix any $x_j\in \bfx$. Using the Taylor remainder formula applied to $q'$ at $x_j$, whenever $|t-x_j|=r \vartheta/m$, there is a $u_{j,t}$ such that $|u_{t,j}-x_j|\leq r \vartheta/m$ and 
		$$
		q'(t)=q'(x_j)+q''(x_j) (t-x_j) + \frac 12 q'''(u_{j,t}) (t-x_j)^2. 
		$$
		Importantly, $x_j$ is a double root of $q$ and so $q'(x_j)=0$. We have $q''(x_j)\geq C_0(m_0,\beta) m^2$ from \cref{lem:qsecond}. For convenience, set $C := (1/20 + T_0 T_3+3T_1T_2)$. To control $|q'''(u_{j,t})|$, note that $|u_{t,j}-x_j|\leq r \vartheta/m\leq 1/m$ since $r\in(0,1/\vartheta)$ by assumption. By \cref{lem:qglobal} with $\tau=1/(2\pi)$ and also the final inequality in \eqref{eq:fejerbounds}, we get
		$$
		|q'''(u_{j,t})|
		\leq |f_m'''(u_{j,t}-x_j)| + \left(2 T_0 T_3+6T_1T_2\right) m^3
		\leq 2C m^3. 
		$$
		Thus, we see that
		\begin{align}
			q'\left(x_j+\frac{r \vartheta}{m}\right)
			&\geq +C_0(m_0,\beta) r \vartheta m - C r^2 \vartheta^2m, \label{eq:qhelp1} \\
			q'\left(x_j-\frac{r \vartheta}{m}\right)
			&\leq -C_0(m_0,\beta) r \vartheta m+ C r^2 \vartheta^2 m. \label{eq:qhelp2} 
		\end{align}
		
		We proceed to examine $\tilde q$. Using inequalities \eqref{eq:qnorm2} and \eqref{eq:qhelp1}, we have
		\begin{align*}
			\tilde q\, '\left(x_j+\frac{r \vartheta}{m}\right)
			&\geq q'\left(x_j+\frac{r \vartheta}{m}\right)-\left|\tilde q \, '\left(x_j+\frac{r \vartheta}{m}\right) - q'\left(x_j+\frac{r \vartheta}{m}\right)\right| \\
			&\geq q'\left(x_j+\frac{r \vartheta}{m}\right)-\|\tilde q \, '-q'\|_{L^\infty(\T)} \\
			&\geq \left( C_0(m_0,\beta) r - C r^2 \vartheta - 1\right) \vartheta m. 
		\end{align*}
		Repeating a similar argument with \eqref{eq:qhelp2} instead establishes 
		$$
		\tilde q \, '\left(x_j-\frac{r \vartheta}{m}\right)
		\leq - \left( C_0(m_0,\beta) r - C r^2 \vartheta - 1\right) \vartheta  m.
		$$
		These inequalities show that \eqref{eq:rcondition} implies 
		$$
		\tilde q \, '\left(x_j+\frac{r \vartheta}{m}\right)>0 \andspace 
		\tilde q\, '\left(x_j-\frac{r \vartheta}{m}\right)<0.
		$$
		By the intermediate value theorem, there exists a $\tilde x_j\in [x_j-r\vartheta/m,x_j+r\vartheta/m]$ such that $\tilde q \, '(\tilde x_j)=0$. This completes the proof. 	
	\end{proof}
	
	Although \cref{lem:localperturbation} only establishes that $\tilde x_j$ is a critical point, we will eventually show that it is actually a local minima, provided the various parameters are correctly chosen. The $1/m$ factor in the conclusion of \cref{lem:localperturbation} comes from the $m^2$ dependence in \cref{lem:qsecond}. Had we gotten a weaker bound for $q''(x_j)$, even something such as $q''(x_j)\geq c_s m^2$ where $c_s\to 0$ as $s\to\infty$, the parameter $r$ would necessarily increase in $s$. 
	
	In order to investigate the size of $\tilde q(\tilde x_j)$, we first investigate the size of $\tilde q(x_j)$. On the one hand, using \eqref{eq:qnorm2}, we obtain
	$$
	\tilde q(x_j)
	\leq |q(x_j)|+|\tilde q(x_j)-q(x_j)|
	\leq 0 + \|\tilde q - q\|_{L^\infty(\T)}
	\leq \vartheta.
	$$
	This inequality is standard and has appeared in previous analysis of MUSIC. It turns out that it is extremely loose, and there are hidden cancellations in $\tilde q-q$ that lead to a significantly better estimate. 
	
	\begin{lemma}
		\label{lem:qlocal}
		Let $m\geq \max\{m_0,s+1\}$ and $\beta>1$. For any $\bfx=\{x_j\}_{j=1}^s\simeq_m\bfU$ such that $\Delta(\bfx)\geq 2\pi\beta/m$ and $\tilde\bfU\in \bbU^{m\times s}$, let $\tilde q =q_{\tilde\bfU}$ and $\vartheta:=\vartheta(\bfx,\tilde\bfU)$. Then $\tilde q(x_j)\leq \vartheta^2$ for all $x_j\in \bfx$.
	\end{lemma}
	
	\begin{proof}
		Using the definition of $\tilde q$, for each $x_j$, we have 
		$$
		\tilde q(x_j)= \bfphi(x_j)^* \tilde \bfU_\perp \tilde \bfU_\perp^* \bfphi(x_j).
		$$
		Since $\bfphi(x_j)$ is a unit norm vector that belongs to the range of $\bfPhi$, it is in the range of $\bfU$ and there is a unit norm $\bfc_j$ such that $\bfphi(x_j)=\bfU \bfc_j$. Then
		$$
		|\tilde q(x_j)|
		= \big\| \tilde \bfU_\perp^* \bfphi(x_j)\big\|_2^2
		= \big\|\tilde \bfU_\perp^* \bfU \bfc_j \big\|_2^2
		\leq \big\|\tilde \bfU_\perp^* \bfU \big\|_2^2
		= \vartheta^2, 
		$$
		where for the last equality, we used that $\big\|\tilde \bfU_\perp \bfU\big\|_2$ is an alternative expression for the sine-theta distance.   
	\end{proof}
	
	Our final task is to prove an analogue of \cref{lem:qglobal} for $\tilde q$ instead of $q$. We will need lower and upper bounds for $\tilde q\, ^{(\ell)}(t)$, for $t$ close to $\tilde x_j$, a critical point in \cref{lem:localperturbation}. This is done through a two step approximation. In the following lemma, $u$ plays the role of $\tilde x_j-x_j$.
	
	\begin{lemma}
		\label{lem:qtildeglobal} 
		Let $m\geq \max\{m_0,s+1\}$, and $\beta>1$ such that $A(\beta) E_0(m_0,\beta,\beta)<1/4$. Fix any $\tau\in [0,\beta)$ and set $T_\ell:=T_\ell(m_0,\beta-\tau,\beta)$ for $\ell\in \{0,1,2\}$. For any $\bfx=\{x_j\}_{j=1}^s$ such that $\Delta(\bfx)\geq 2\pi\beta/m$ and $\tilde\bfU\in \bbU^{m\times s}$, let $\tilde q =q_{\tilde\bfU}$ and $\vartheta:=\vartheta(\bfx,\tilde\bfU)$. For any $x_j\in \bfx$ and $t\in \T$ such that $|t-x_j|\leq 2\pi \tau/m$, we have
		\begin{align*}
			\left|\tilde q\,'(t)
			+f_m'(t-u-x_j) \right| 
			&\leq \left(2 \, T_0 T_1 + \vartheta + \frac 1 6 m |u| \right) m, \\
			\left| \tilde q \, ''(t)+f_m''(t-u-x_j) \right| 
			&\leq \left(2T_1^2+2T_0T_2+\vartheta+\frac 1 {10} m |u| \right) m^2. 
		\end{align*} 
	\end{lemma}
	
	\begin{proof}
		Let $q=q_{\bfU}$ where $\bfx\simeq_m\bfU$. For each $\ell\in \{1,2\}$, we start with the inequality 
		\begin{align*}
			&|\tilde q\, ^{(\ell)}(t)+f_m^{(\ell)}(t-u-x_j)| \\
			&\leq \big|q\,^{(\ell)}(t)+f_m^{(\ell)}(t-x_j)\big| 
			+\big|\tilde q\, ^{(\ell)}(t)-q^{(\ell)}(t)\big| + |f_m^{(\ell)}(t-u-x_j)-f_m^{(\ell)}(t-x_j)|. 
		\end{align*}
		Since $|t-x_j|\leq 2\pi \tau/m$ by assumption, the first term can be controlled by \cref{lem:qglobal}. The second term can be upper bounded using \eqref{eq:qnorm2}. For the final term, we use the mean value theorem and the upper bounds in \eqref{eq:fejerbounds}. 
	\end{proof}
	
	\subsection{Some formulas and inequalities}
	
	Recall the Fourier series representation the Fej\'er kernel, 
	$$
	f_m(t) = \frac 1 m \sum_{k=-m+1}^{m-1} \left( 1-\frac{|k|}m \right) e^{ikt}.
	$$
	Taking derivatives and then applying triangle inequality, we see that
	\begin{equation}
		\label{eq:fejerbounds}
		\|f_m\|_{L^\infty(\T)}
		\leq 1, \quad
		\|f_m'\|_{L^\infty(\T)}
		\leq \frac {m}{3}, \quad 
		\|f_m''\|_{L^\infty(\T)}
		\leq \frac {m^2}{6}, \quad 
		\|f_m'''\|_{L^\infty(\T)}
		\leq \frac {m^3}{10}. 
	\end{equation}
	
	For convenience, we define $B_0:=1$, $B_1=1/12$, $B_2=1/80$, and $B_3=1/448$. For all $t\in \T$ and integer $m\geq 1$, we have
	\begin{equation}
		\begin{split}
			\|\bfphi(t)\|_2^2
			&=\frac 1 m \sum_{k\in I(m)} 1^2 
			= B_0, \\
			\|\bfphi'(t)\|^2_2
			&= \frac 1 m \sum_{k\in I(m)} k^2
			= \frac 1 {12} (m^2-1)
			\leq B_1 m^2, \\
			\|\bfphi''(t)\|^2_2
			&= \frac 1m \sum_{k\in I(m)} k^4 
			=\frac 1 {240} (3m^4-10m^2+7)
			\leq B_2m^4, \\
			\|\bfphi'''(t)\|^2_2
			&= \frac 1m \sum_{k\in I(m)} k^6 
			=\frac 1 {1344} (3m^6-21m^4+49m^2-31)
			\leq B_3 m^6.  
		\end{split}\label{eq:phinorm1}
	\end{equation}

	The previous two lemmas are expressed in terms of the Fej\'er kernel. Its behavior is analyzed in the following lemma.  
	
	\begin{lemma}
		\label{lem:fejer1}
		Let $m\geq m_0\geq 2$. For any $t\in \T$ such that $t\geq 2\pi\tau_0/m$  for some $\tau_0\in (0,1]$, we have
		\begin{equation*}
			f_m(t)
			\leq \max\left\{\frac 1 4, \, \frac{1}{m_0^2\sin^2(\pi\tau_0/m_0)} \right\}. 
		\end{equation*}
		For any $t\in [0,\pi]$, we have 
		\begin{align*}
			\sin\left( \frac {mt}{2}\right) \left(\frac 1 3 \left(1-m_0^{-2}\right) - \frac 1 {120} m^2 t^2 \right) m 
			&\leq -f_m'(t)
			\leq \frac 1 6 m^2 t,  \\
			\frac 1 6 \left(1-m_0^{-2}\right) m^2 - \frac{1}{30}m^4t^2
			&\leq -f_m''(t)
			\leq \frac 1 6.
		\end{align*}
	\end{lemma}
	
	\begin{proof}
		Using the inequality $|\sin(t/2)|\geq |t|/\pi$ which holds for $|t|\leq \pi$, we see that
		$$
		f_m(t)
		=\frac{\sin^2(mt/2)}{m^2 \sin^2(t/2)}
		\leq \frac{\pi^2}{m^2 |t|^2}.
		$$
		This shows that $f_m(t)\leq 1/4$ for all $|t|\geq \pi/m$. This finishes the estimate if $\tau_0\in [1/2, 1]$. For the case where $t\in [2\pi\tau_0/m,\pi/m]$ for $\tau_0\in (0,1/2)$, we use that $f_m$ is decreasing on $[0,2\pi/m]$ and also note that $\{m\sin(u/m)\}_{m=1}^\infty$ is non-decreasing in $m$ for any fixed $u\in [0,\pi]$. Since $m\geq m_0$ by assumption, this now implies 
		$$
		\max_{t \in [2\pi \tau_0/m,\pi/m]} f_m(t)
		= f_m\left(\frac {2\pi \tau_0} m\right)
		= \frac{\sin^2(\pi \tau_0)}{m^2\sin^2(\pi \tau_0/m)}
		\leq \frac{1}{m_0^2\sin^2(\pi\tau_0/m_0)}.
		$$
		This proves the first inequality. 
		
		We now concentrate on the first derivative estimates. A calculation and manipulation of trigonometric functions establishes the formula
		\begin{align*}
			-\frac 1 m \,  f_m'(t)
			&:= \frac{\sin(m t/2)}{2m^3\sin^3(t/2)} \, g_m(t), \\
			g_m(t)
			&:=  (m+1) \sin\left( \frac{m-1}2 t \right) - (m-1) \sin\left( \frac{m+1}2 t \right). 
		\end{align*}
		Using the observation that for $t\geq 0$, the power series of $\sin(t)$ truncated to degree $2k+1$ overestimates $\sin(t)$ when $k$ is even and underestimates it when $k$ is odd, we obtain
		\begin{align*}
			g_m(t)
			\geq \frac 1 {12} \left(1-m^{-2}\right) m^3t^3 - \frac 1 {480} \left(1-m^{-4}\right) m^5t^5 . 
		\end{align*} 
		Using this inequality, that $m\geq m_0$, and $\sin(t/2)\leq t/2$, we obtain 
		\begin{align*}
			-\frac 1 m \,  f_m'(t)
			&\geq \sin\left( \frac{m t}2\right) \left(\frac 1 3 \left(1-m_0^{-2}\right) - \frac 1 {120} m^2t^2 \right). 
		\end{align*}
		This proves the claimed lower bound for $-f'_m$. 
		
		A direct computation shows that the power series expansion of $f_m$ has the form
		\begin{equation*}
			f_m(t)
			=1-\frac 1 {12} (m^2-1) t^2 + \frac{1}{720} (2m^4-5m^2+3)t^4 - \cdots 
		\end{equation*}
		Differentiating each term at a time, we readily obtain power series expansion for $-f_m'$. In particular, we see that the $t^3$ term in $f_m'$ is negative, which implies
		$$
		f_m'(t)
		\leq \frac 1 6 (m^2-1)t
		\leq \frac  1 6 m^2 t.  
		$$
		This proves the claimed upper bound for $f_m'$. 
		
		Finally, we compute the power series expansion of $-f''_m$, which has a negative $t^2$ and positive $t^4$ term. Thus, we have 
		\begin{align*}
			-f_m''(t)
			&\leq \frac 1 6 \left(m^2-1\right)\leq \frac 1 6 m^2, \\
			-f_m''(t)
			&\geq \frac 1 6 \left(m^2-1\right) - \frac{1}{60} \left(2m^4-5m^2+3\right) t^2.
		\end{align*}
		Using that $m\geq m_0\geq 2$ completes the proof.  
	\end{proof}
	
	\subsection{Numerical evaluation of constants}
	\label{rem:numericalconstants}
	
	The lemmas in this section involve many complicated constants that depend on $E_\ell(m,\alpha,\beta)$, which is defined in equation \eqref{eq:energy}. We could have given simpler, but looser, upper bounds for $E_\ell(m,\alpha,\beta)$. However, this would result in stronger conditions on $\beta$ and $\vartheta$ compared to the ones found in the main theorems. For this reason, we have decided to present more accurate and complicated estimates, whose constants would need to be numerically computed. 
	
	It is imperative to note each energy constant $E_\ell(m,\alpha,\beta)$ can be numerically evaluated to machine precision given choices for $m$, $\alpha$, and $\beta$. The only potential issue with numerical evaluation of the integral in \eqref{eq:energy} is the singularity of $h_{m, \ell}$ at zero. However, the integrals are taken over the region $[2\pi(\alpha+\beta/2)/m,\pi]$ for fixed $\beta>1$ and $\alpha>0$ while $|h_{m, \ell}|$ is uniformly bounded in $m$ in this domain. Thus, any standard numerical integration method can be used to calculate $E_\ell$ up to arbitrary precision.
	
	The same reasoning extends to the other constants that depend on $E_\ell(m,\alpha,\beta)$, which include $C_\ell(m_0,\beta)$ and $T_\ell(m,\alpha,\beta)$ defined in \eqref{eq:cmbeta} and \eqref{eq:Tconstant} respectively. Again, we need to check that these expressions can be accurately computed. Both $C_0(m_0,\beta)$ and $C_1(m_0,\beta)$ contain a $E_1(m,\beta,\beta)$ term which is not problematic. In \cref{lem:qglobal}, the constant $T_\ell(m_0,\beta-\tau_1,\beta)$ is also fine since we will always use it for $\tau_1\leq 1$ and $\beta>1$. Likewise, \cref{lem:qfar} only has a $T_0(m_0,\beta/2,\beta)$ term.

	\section{Proofs of theorems}
	\label{sec:proofthm}

	\subsection{Proof of \cref{cor:deterministic}}
	\label{sec:deterministic}
	
	Let $s=|\bfx|$. Recall the quantity $\rho:=\rho(\bfx,\bfa,\bfeta)$ from \cref{def:subspaceerror2}. We first derive an upper bound for $\rho$. By factorization \eqref{eq:toeplitzfactorization} and \cref{prop:AB19} with $\beta=4$, we have 
	\begin{equation}
		\label{eq:Tyinequality}
		\sigma_s(\bfT)
		\geq a_{\min} \, \sigma_s^2(\bfPhi) 
		\geq \frac 3 4 a_{\min} \, m
	\end{equation}
	Using \cref{lem:noisepnorm} and assumption \eqref{eq:nsrp}, we have
	\begin{equation}
		\label{eq:rhohelp1}
		\rho = \frac{2\|T(\bfeta)\|_2}{\sigma_s(\bfT)} \leq \frac{16 m^{1-1/p} \|\bfeta\|_p}{3 a_{\min}\, m}
		\leq \min\left\{ \frac 1{100}, \,  \frac{16 \|\bfeta\|_p}{3 a_{\min}\, m^{1/p}}\right\}. 
	\end{equation}
	
	By \cref{lem:rho}, the Toeplitz estimator $\tilde\bfU$ is well-defined. To see that it is correctly computed, recall we assumed that $(\bfx,\bfa)\in \calS(8\pi/m,r_0, 10r_0)$. Using \cref{lem:sparsity} with $\beta=4$ and $r_1=10r_0$, and that $\gamma=21/400=0.0525$ and $\rho\leq 0.01$, we see that the first step of \cref{alg:gradMUSIC2} correctly identifies $s$. This shows that $\tilde\bfU$ is computed correctly. 
	
	Next, we upper bound $\vartheta:=\vartheta(\bfx,\tilde\bfU)$, see \cref{def:subspaceerror}. Also by \cref{lem:rho}, we have $\vartheta\leq \rho\leq 0.01$, so the assumptions of \cref{thm:mainMUSIC} hold. Using the theorem, inequality \eqref{eq:rhohelp1}, and number of iterations \eqref{eq:niter}, we see that 
	\begin{align*}
		\max_j |x_j-\hat x_j|
		&\leq \frac{7\vartheta} m+ \frac{77 \pi (0.839)^n}{m}
		\leq \frac{112 \|\bfeta\|_p}{3 a_{\min}\, m^{1+1/p}} + \frac{77 \pi (0.839)^n}{m}
		\leq \frac{55 \, \|\bfeta\|_p}{a_{\min}\, m^{1+1/p}}.
	\end{align*}
	
	For the amplitude error, assumption \eqref{eq:nsrp} and  the previous inequality imply that the frequency error is at most $2\pi \alpha/m$ where $\alpha=33/(640\pi)$. We use \cref{lem:amplitudes} for this choice of $\alpha$ and $\beta =4$. Combining this with \cref{lem:noisepnorm} and the frequency error bound that we just established, we have that
	\begin{align*}
		\max_j |a_j-\hat a_j|
		\leq \left( 1 + \sqrt{\frac{\beta+2}{\beta-2}} \right) \frac{\beta + 2}{\beta-2} \,  55 \frac{a_{\max}}{a_{\min}} \frac{\|\bfeta\|_p}{m^{1/p}} + \frac \beta {\beta -1} \frac{2\|\bfeta\|_p}{m^{1/p}}
		\leq 455 \frac{a_{\max}}{a_{\min}} \frac{\|\bfeta\|_p}{m^{1/p}}. 
	\end{align*}
	
	\subsection{Proof of \cref{cor:stochastic}}
	\label{sec:stochastic}
	
	For convenience, set $s=|\bfx|$. For any $t>1$, consider the events
	\begin{equation}
		\label{eq:events}
		\calA := \left\{ \frac{\|T(\bfeta)\|_2}{a_{\min} \, m} 
		\leq \frac{3}{800} \right\}, \andspace
		\calB_t := \left\{ \|T(\bfeta)\|_2 
		\leq t \sqrt{2 \tr(\bfSigma) \log(2m)} \right\}. 
	\end{equation}
	By \cref{lem:toeplitznorm}, there is a $c>0$ such that
	\begin{equation*}
		\P(\calA^c)
		\leq 2m \exp\left( -\frac{c a_{\min}^2  m^2}{\tr(\bfSigma)} \right),
		\andspace
		\P(\calB_t^c)
		\leq 2m^{1-t^2}.
		\label{eq:nsrstochastic}
	\end{equation*}
	Due to the union bound, the event $\calA\cap \calB_t$ occurs with probability at least \eqref{eq:nsrsampling}, which we assume holds for all subsequent parts of this proof. 
	
	We first derive an upper bound for $\rho:=\rho(\bfx,\bfa,\bfeta)$ given in \cref{def:subspaceerror2}. Using inequality \eqref{eq:Tyinequality} and \eqref{eq:events}, we have
	\begin{equation}
		\label{eq:rhohelp2}
		\rho = \frac{2\|T(\bfeta)\|_2}{\sigma_s(\bfT)}
		\leq \frac{8\|T(\bfeta)\|_2}{3a_{\min} \, m}
		\leq \min\left\{ \frac{1}{100}, \, \frac{2 t \sqrt{2\tr(\bfSigma)\log(2m)}}{a_{\min}\, m} \right\}. 
	\end{equation}
	
	By \cref{lem:rho}, the Toeplitz estimator $\tilde\bfU$ is well-defined. Recall we assumed that $(\bfx,\bfa)\in \calS(8\pi/m,r_0,10r_0)$. Using \cref{lem:sparsity} with $\beta=4$ and $r_1=10r_0$, and that $\gamma=0.0525$ and $\rho\leq 0.01$, we see that the number of frequencies is correctly detected in the first step of \cref{alg:gradMUSIC2}. This shows that $\tilde U$ is computed correctly. 
	
	Next, we upper bound $\vartheta:=\vartheta(\bfx,\tilde\bfU)$, see \cref{def:subspaceerror}. Using \cref{lem:rho} again, we have $\vartheta\leq \rho$. In particular, $\vartheta\leq 0.01$ so the assumptions of \cref{thm:maingradientMUSIC} hold. Using this theorem, inequality \eqref{eq:rhohelp2}, and bound for number of gradient iterations \eqref{eq:niter2}, we see that 
	\begin{align*}
		\max_j |x_j-\hat x_j|
		&\leq \frac{7\vartheta} m + \frac{77\pi (0.839)^n}{m}
		\lesssim \frac{1}{a_{\min}} \frac{t \sqrt{\tr(\bfSigma)\log(m)}}{m^2}.
	\end{align*}
	
	Note that \eqref{eq:nsrsampling} implicitly assumes that $a_{\min}^2 m^2/\tr(\bfSigma)$ can be made a sufficiently large constant of our choice. Hence, we can assume that the frequency error is no larger than $2\pi\alpha/m$ for some $\alpha\leq 2$. For the amplitude error, we use \cref{lem:amplitudes} for $\beta =4$. Combining this lemma with the frequency error bound and inequality \eqref{eq:rhohelp2}, we see that
	\begin{align*}
		\max_j |a_j-\hat a_j|
		&\lesssim \frac{a_{\max}}{a_{\min}} \frac{t \sqrt{\tr(\bfSigma)\log(m)}}m + \frac{t \sqrt{\tr(\bfSigma)\log(m)}}{a_{\min}\, m} \\
		&\lesssim \frac{a_{\max}}{a_{\min}} \frac{t \sqrt{\tr(\bfSigma)\log(m)}}{m}. 
	\end{align*}

	\subsection{Proof of \cref{thm:landscape} }
	\label{proof:landscape}
	
	For this proof, define the constant $c_1:=0.01$ so that $\vartheta\leq c_1$. We assumed that $\Delta(\bfx)\geq 8\pi/m=2\pi (4)/m$ and $m\geq 100$. At various points in this proof, we apply lemmas in \cref{sec:landscape} with $\beta=4$ and $m_0=100$. The energy constants $E_\ell$ in equation \eqref{eq:energy} and other quantities that depend on $E_\ell$ can be numerically computed up to machine precision, see the discussion in \cref{rem:numericalconstants}. Also recall the definition of $T_\ell(m,\alpha,\beta)$ in \eqref{eq:Tconstant}. Different choices of parameters will be selected in $T_\ell$ for each step of the proof. 
	
	Various constants that appear in this proof are computed numerically. We display the first three nonzero digits, rounded up or down, depending on whether that step proves an upper or lower bound, respectively. A script that computes these constants and verifies the lemmas' assumptions are included in the numerical software accompanying this paper. 
	\begin{enumerate}[(a)]
		\item 
		We start by using \cref{lem:localperturbation} where $r=7$. We numerically check that the condition in \eqref{eq:rcondition} is fulfilled. According to this lemma, $\tilde q$ has $s$ critical points $\tilde\bfx=\{\tilde x_1,\dots,\tilde x_s\}$ such that for each $j\in \{1,\dots,s\}$, we have
		\begin{equation}
			\label{eq:maxxpert}
			|\tilde x_j-x_j|
			\leq \frac{7\vartheta}{m}
			\leq \frac{7c_1}{m}. 
		\end{equation}
		\item 
		Fix any $t$ such that $|t-\tilde x_j|\leq \pi/(3m)$. We will use \cref{lem:qtildeglobal} where $u=\tilde x_j-x_j$. By inequality \eqref{eq:maxxpert}, we have $|u|\leq 7 c_1/m$ and $|t-x_j|\leq |t-\tilde x_j|+|x_j-\tilde x_j|\leq \pi/(3m)+7c_1/m$. So we use the referenced lemma with parameter $\tau = 1/6+7c_1/(2\pi)$. Hence, we obtain
		\begin{align*}
			\left| \tilde q \, ''(t) + f_m''(t-\tilde x_j) \right|
			&\leq \left(2T_1^2+2T_0T_2+\frac{17}{10} \vartheta\right) m^2. 
		\end{align*}
		Using the lower bound for $-f_m''$ in \cref{lem:fejer1}, we obtain the lower bound
		\begin{align*}
			\tilde q \, ''(t) 
			&\geq \left(\frac 1 6 \, (1-100^{-2}) - \frac{1}{30}m^2 (t-\tilde x_j)^2 - 2 T_1^2 - 2T_0T_2 - \frac{17}{10}c_1 \right)m^2. 
		\end{align*}
		The right hand side is a decreasing function of $(t-\tilde x_j)^2$, so its minimum is attained at its endpoints, so when $|t-\tilde x_j|=\pi/(3m)$. This shows that 
		$$
		\tilde q \, ''(t)
		\geq 0.0271 \, m^2. 
		$$ 
		For the upper bound on $\tilde q \, ''$, we use the upper bound for $-f_m''$ in \eqref{eq:fejerbounds} instead to see that
		\begin{align*}
			\tilde q \, ''(t) 
			\leq \left(\frac 1 6 + 2 T_1^2 + 2T_0T_2 + \frac{17}{10}c_1 \right) m^2 \leq 0.269 \, m^2. 
		\end{align*}
		
		\item 
		Part (b) immediately tells us that the critical point $\tilde x_j$ is also a local minimum of $\tilde q$. By \cref{lem:qlocal}, we have 
		$\tilde q(x_j)\leq \vartheta^2.$ Note that $x_j \in [\tilde x_j- \pi/(3m), \, \tilde x_j+\pi/(3m)]$ as well. Since $\tilde q$ is convex in this interval, $\tilde x_j$ is its only local minimum, so we have $\tilde  q(\tilde x_j)\leq \tilde q(x_j)$.
		\item 
		Since the Fej\'er kernel is even, we only need to prove the desired estimate for the case when $t-\tilde x_j \in [\pi/(3m), 4\pi/(3m)]$. We will use \cref{lem:qtildeglobal} where $u=\tilde x_j-x_j$. By inequality \eqref{eq:maxxpert}, we see that $|u|=|x_j-\tilde x_j|\leq 7c_1/m$ and $|t-x_j|\leq |t-\tilde x_j|+|\tilde x_j-x_j|\leq 4\pi/(3m)+7c_1/m$. Then we use the referenced lemma with parameter $\tau = 2/3+7/(2\pi)$ to obtain
		\begin{equation*}
			\left|\tilde q \, '(t) + f_m'(t-\tilde x_j)\right|
			\leq \left(2 \, T_0 T_1 + \frac {13} 6 \vartheta \right) m. 
		\end{equation*}
		Using the lower bound for $-f_m'$ in \cref{lem:fejer1}, we obtain the lower bound,
		\begin{align*}
			\tilde q \, '(t)
			&\geq \left(\sin\left(\frac {m(t-\tilde x_j)}2\right) \left(\frac 1 3 \left(1-100^{-2}\right) - \frac 1 {120} m^2 (t-\tilde x_j)^2 \right) - 2 T_0 T_1 - \frac{13 } 6c_1 \right) m.
		\end{align*}
		We next argue that the right side is positive whenever $t-\tilde x_j\in [\pi/(3m),4\pi/(3m)]$. A calculus argument shows that the function in $t-\tilde x_j$ is concave, so we just need to evaluate this estimate at the endpoints. Doing so, we see that 
		$$
		\tilde q \, '(t)
		\geq 0.0306 \, m. 
		$$
		
		\item 
		Fix any $t\in\T$ with $|t-\tilde x_j|\geq 4\pi/(3m)$ for all $j\in \{1,\dots,s\}$. Using inequality \eqref{eq:maxxpert}, we see that $|t-x_j|\geq |t-\tilde x_j|+|x_j+\tilde x_j|\geq 4\pi/(3m)-7c_1/m$ for all $j\in \{1,\dots,s\}$. Additionally, inequality \eqref{eq:qnorm2} tells us that $\tilde q(t)\geq q(t) - \vartheta\geq q(t)-c_1$. We use \cref{lem:qfar} and \cref{lem:fejer1} with parameter $\tau_0=2/3-7/(2\pi)$ to obtain 
		\begin{align*}
			\tilde q(t)
			&\geq 1 - \max\left\{\frac 1 4, \, \frac{1}{m_0^2\sin^2(\pi \tau_0/m_0)} \right\} - (T_0(m_0,\beta/2,\beta))^2 - c_1
			\geq 0.529. 
		\end{align*}
		
		\item 
		Since the Fej\'er kernel is even, we only need to prove the desired estimate for $t\in [\tilde x_j+\pi/(3m),\tilde x_j+4\pi/(3m)]$. This is a continuation of the argument in part (d), except we use the upper bound for $-f_m'$ \cref{lem:fejer1} instead. Doing so, we get
		\begin{align*}
			\tilde q \, '(t)
			\leq \left( \frac 1 6 m (t-\tilde x_j) + 2 T_0 T_1 + \frac{13}6 \vartheta \right) m.
		\end{align*}
		Next, using that $m(t-\tilde x_j)\geq \pi/3$, we see that 
		$$
		\tilde q \, '(t)
		\leq \left( \frac 1 6 + \frac 6 \pi T_0 T_1 + \frac{13} {2\pi} \vartheta \right) m^2 (t-\tilde x_j)
		\leq 0.292 \, m^2 |t-\tilde x_j|.  
		$$
	\end{enumerate}
	
	\subsection{Proof of \cref{thm:maingradientMUSIC}}
	\label{proof:maingradientMUSIC}
	
	For convenience, set $s:=|\bfx|$. The assumptions of \cref{thm:landscape} hold, so $\tilde q$ has $s$ local minima $\tilde \bfx$ satisfying the properties listed there. The assumptions of \cref{lem:grad} also hold since we assumed $\alpha=0.529$, $\vartheta\leq 0.01$, and $\mesh(G)\leq 1/(2m)$. By the referenced lemma, $A$ is the union of exactly $s$ nonempty disjoint clusters in $G$. Pick any representative $t_{j,0}\in A_j=A\cap B_j$, where the interval $B_j$ is defined in equation \eqref{eq:basin}. Let $\{t_{j,k}\}_{k=0}^\infty$ be the iterates produced by gradient descent \eqref{eq:gradient} with initial point $t_{j,0}$ and step size $h=6/m^2$. 
	
	For the analysis of gradient descent, we first make an additional assumption and return back to the general case later. Assume additionally that 
	\begin{equation*}
		t_{j,0}\in I_j:=\left[ \tilde x_j - \frac{\pi}{3m}, \, \tilde x_j + \frac{\pi}{3m}\right]. 	
	\end{equation*}
	By \cref{thm:landscape} part (b), the function $\tilde q$ is strictly convex on $I_j$ and enjoys the bounds 
	$$
	0.0271 \, m^2 \leq \tilde q \, ''(t)\leq 0.269 \, m^2 \forallspace t\in I_j.
	$$
	We apply a standard result for gradient descent on a smooth convex landscape, such as \cite[Thoerem 2.1.15]{nesterov2018lectures}, where $\mu=0.0271 \, m^2$ and $L=0.269 \, m^2$ in the referenced theorem. Then gradient descent with step size no greater than $2/(\mu+L)\leq  6.754/m^2$ (note we selected $h=6/m^2$) and initial guess $t_{j,0}\in I_j$ produces iterates $\{t_{j,k}\}_{k=0}^\infty$ that satisfy the inequality 
	\begin{equation*}
		|t_{j,k}-\tilde x_j|
		\leq \left(1-\frac{2h\mu L}{\mu+L}\right)^{k/2} |t_{j,0}-\tilde x_j|
		\leq (0.839)^k \, |t_{j,0}-\tilde x_j|.
	\end{equation*}
	In particular, since $t_{j,0}\in I_j$, we have
	\begin{equation}
		\label{eq:expconvergence}
		|t_{j,k}-\tilde x_j|
		\leq \frac{\pi (0.839)^k}{3m}.
	\end{equation}
	This shows that gradient descent converges to $\tilde x_j$ exponentially provided that $t_{j,0}\in I_j$. 
	
	It remains to consider the complement case where $t_{j,0}\in B_j\setminus I_j.$ We first prove that if $t_{j,k}\in B_j\setminus I_j$ for any $k\geq 1$, then $t_{j,k+1}\in B_j$. Since \cref{thm:landscape} part (d) and (f) provide anti-symmetric inequalities with opposite signs, we assume without loss of generality that $t_{j,k}\in B_j\setminus I_j$ with $t_{j,k}>\tilde x_j$. Using the referenced theorem part (d) and (f) and that $h=6/m^2$, we have
	\begin{equation}
		\label{eq:gradhelp1}
		\frac{0.1836}{m} \leq h \tilde q \, '(t_{j,k})\leq 1.752 \, (t_{j,k}-\tilde x_j). 
	\end{equation}
	Inserting \eqref{eq:gradhelp1} into the definition of gradient descent \eqref{eq:gradient}, and using that $t_{j,k}\in B_j\setminus I_j$ and $t_{j,k}>\tilde x_j$, we see that 
	\begin{align*}
		t_{k+1}-\tilde x_j 
		&\leq t_{j,k}-\tilde x_j - \frac{0.1836}m \\
		&\leq t_{j,k}-\tilde x_j - \frac{3(0.1836)}{4\pi} (t_{j,k}-\tilde x_j)
		\leq 0.956 \, (t_{j,k}-\tilde x_j), \\
		t_{k+1}-\tilde x_j 
		&\geq t_{j,k}-\tilde x_j - 1.752 \, (t_{j,k}-\tilde x_j)
		= -0.752 \, (t_{j,k}-\tilde x_j). 
	\end{align*}
	In particular, these inequalities imply that
	\begin{equation}
		\label{eq:iterates}
		|t_{j, k+1}-\tilde x_j|
		\leq 0.956 \, |t_{j,k}-\tilde x_j| \ifspace t_{j,k}\in B_j\setminus I_j.
	\end{equation}
	This proves that $t_{j,k+1}\in B_j$ if $t_{j,k}\in B_j\setminus I_j$, and completes the proof of this claim. 
	
	We next claim that it takes at most $31$ iterations to reach $I_j$ from $B_j\setminus I_j$. By the previous claim, all these iterates have to be in $B_j$. Suppose for the purpose of deriving a contradiction that $t_{j,0},\dots,t_{j,31}\in B_j\setminus I_j$. Iterating inequality \eqref{eq:iterates}, we see that
	$$
	|t_{j,31}-\tilde x_j|
	\leq (0.956)^{31} \, |t_{j,0}-\tilde x_j|
	\leq \frac{(0.956)^{31} 4\pi}{3 m}
	< \frac{\pi}{3m}. 
	$$
	This contradicts the original assumption that $t_{j,31}\in B_j\setminus I_j$. This completes the proof of this claim. 
	
	Now we are ready to complete the proof. Let $\hat x_j=t_{j,n}$ where $n\geq 31$. If $t_{j,0}\in B_j$, it requires at most 30 iterations to arrive in $I_j$. Using inequality \eqref{eq:expconvergence},  we have
	\begin{equation*}
		|\hat x_j-\tilde x_j|
		\leq \frac{\pi (0.839)^{n-31}} {3m}
		\leq \frac{77 \pi (0.839)^n}{m}.
	\end{equation*}
	Using \cref{thm:landscape} part (a) and triangle inequality, we have 
	\begin{align*}
		|\hat x_j-x_j|
		&\leq |\hat x_j-\tilde x_j|+|\tilde x_j-x_j| 
		\leq \frac{7\vartheta} m + \frac{77 \pi (0.839)^n}m.
	\end{align*}

	\subsection{Proof of \cref{thm:mainMUSIC}}
	\label{proof:mainMUSIC}
	
	For convenience, let $s:=|\bfx|$. The assumptions of \cref{thm:landscape} hold, so $\tilde q$ has $s$ many local minima $\tilde \bfx$ satisfying the properties listed there. Since $\mesh(G)\leq 2\pi/(3m)$ by assumption, consecutive elements in $G$ are at most $2\mesh(G)\leq 4\pi/(3m)$ apart. Since the interval
	\begin{equation}
		\label{eq:interval1}
		\left[\tilde x_j- \frac{4\pi}{3m}, \, \tilde x_j+\frac{4\pi}{3m}\right],
	\end{equation}
	has length $8/(3m)$, we see that $G$ has at least three elements in this interval. \cref{thm:landscape} part (b) and (d) imply that $\tilde q$ is decreasing and then increasing in the intervals
	$$
	\left[\tilde x_j- \frac{4\pi}{3m}, \, \tilde x_j\right] \andspace \left[\tilde x_j, \, \tilde x_j+\frac{4\pi}{3m}\right].
	$$ 
	Thus, $\tilde q$ evaluated on $G$ has exactly one discrete local minima contained in \eqref{eq:interval1}, which we denote by $\hat x_j$. See \cref{def:discretemin} for a definition of discrete local minimum. Note that $\hat x_j$ is selected as $\tilde x_j$ if $\tilde x_j\in G$, or $\hat x_j$ is selected as either the left or right closest neighbor to $\tilde x_j$ on $G$. Regardless of which element is chosen, we have the property that 
	\begin{equation}
		\label{eq:griderror}
		\max_{j=1,\dots,s} |\tilde x_j-\hat x_j|\leq \mesh(G). 
	\end{equation}
	
	We next enumerate the discrete local minima of $\tilde q$ on $G$ as 
	$$
	\hat x_1,\dots,\hat x_s, u_1,\dots, u_r.
	$$
	Here, we use the convention that if $r=0$, then $\hat x_1,\dots,\hat x_s$ are the only discrete local minima of $\tilde q$ on the grid. We will show that MUSIC always picks $\hat x_1,\dots,\hat x_s$. There is nothing to prove if $r=0$, so assume $r\geq 1$. We already established that $\hat x_j$ is the unique discrete local minimum  of $\tilde q$ in the interval \eqref{eq:interval1}. Consequently, for each $k\in \{1,\dots,r\}$, we have
	$$
	u_k\not\in \bigcup_{j=1}^s \left[\tilde x_j- \frac{4\pi}{3m}, \, \tilde x_j+\frac{4\pi}{3m}\right].
	$$
	This enables us to use \cref{thm:landscape} part (e) to see that $\tilde q(u_k) \geq 0.529$ for each $k$. On the other hand, \cref{thm:landscape} part (c) tells us that
	$\tilde q(\tilde x_j)
	\leq \vartheta^2\leq 10^{-4}
	$ for each $j$. 
	Using inequalities \eqref{eq:griderror} and \eqref{eq:qnorm1} together with the mean value theorem, we see that for all $j$,
	\begin{align*}
		\tilde q(\hat x_j)
		&\leq \tilde q(\tilde x_j) + |\tilde q(\hat x_j)-\tilde q(\tilde x_j)| \\
		&\leq \tilde q(\tilde x_j) + \|\tilde q\,'\|_{L^\infty(\T)}|\hat x_j-\tilde x_j| 
		\leq \vartheta^2 + \vartheta
		<0.529
		\leq \min_{k=1,\dots,r} \tilde q(u_k). 
	\end{align*}
	Thus, we have shown that $\hat x_1,\dots,\hat x_s$ are the $s$ smallest discrete local minima of $\tilde q$ on $G$. 
	
	Returning back to the selected local minima, using \cref{thm:landscape} part (a) and inequality \eqref{eq:griderror}, we see that the frequency error is
	$$
	|x_j-\hat x_j|
	\leq |x_j-\tilde x_j| + |\tilde x_j-\hat x_j|
	\leq \frac{7\vartheta}{m}+\mesh(G). 
	$$
	This proves the frequency error bound. 
	
	\subsection{Proof of \cref{thm:maingradMUSIC2}}
	\label{proof:maingradMUSIC2}
	
	In the proof of \cref{thm:maingradientMUSIC}, we showed that after at most 31 iterations of gradient descent, regardless of which representative $t_{j,0}$ is chosen, $t_{j,31}\in [\tilde x_j-4\pi/(3m),\tilde x_j+4\pi/(3m)]$. Since $n_j\geq N_{\min}\geq 31$ and $\hat x_j=t_{j,n_j}$ by definition, we have $|\hat x_j-\tilde x_j|\leq 4\pi/(3m)$. By the mean value theorem and that $\tilde x_j$ is a local minimum of $\tilde q$, there is a $\xi_j$ between $\tilde x_j$ and $\hat x_j$ such that 
	$$
	\tilde q\, '(\hat x_j)
	= \tilde q\, '(\tilde x_j) + \tilde q\, ''(\xi_j)(\hat x_j-\tilde x_j)
	=\tilde q\, ''(\xi_j)(\hat x_j-\tilde x_j).
	$$
	Since $|\hat x_j-\tilde x_j|\leq \pi/(3m)$, we can use \cref{thm:landscape} part (b) to control $\tilde q\, ''(\xi_j)$. Also using that $|\tilde q \, '(\hat x_j)|\leq \epsilon m$ by definition of $n_j$, we have 
	$$
	|\hat x_j-\tilde x_j| 
	=\frac{|\tilde q \, '(\hat x_j)|}{|\tilde q \, ''(\xi_j)|}
	=\frac{|\tilde q \, '(t_{j,n_j})|}{|\tilde q \, ''(\xi_j)|}
	\leq \frac{37\epsilon}{m}. 
	$$
	Using \cref{thm:landscape} part (a) and triangle inequality completes the proof. 
	
	\section{Proofs of lemmas}
	\label{sec:prooflemma}
	
	\subsection{Proof of \cref{lem:minimax}}
	\label{proof:minimax}
	The proof requires an abstract argument that holds for any $\calS$ and $\calN$. Similar techniques were used in \cite{batenkov2021super,li2021stable}. Suppose there are distinct pairs $(\bfx,\bfa),(\bfx',\bfa')\in \calS$, and $\bfeta,\bfeta' \in \calN$ such that 
	\begin{equation}
		\label{eq:data}
		\tilde\bfy(\bfx,\bfa,\bfeta)
		=\tilde\bfy(\bfx',\bfa',\bfeta').
	\end{equation}
	Here, we have two sets of parameters $(\bfx,\bfa)$ and $(\bfx',\bfa')$ that generate the same noisy data. Let $\psi=(\hat\bfx,\hat\bfa)$ be an arbitrary method given data \eqref{eq:data}. By definition of $X(\psi,\calS,\calN)$, we have
	\begin{align*}
		X(\psi,\calS,\calN)
		\geq \max\left\{ \|\hat \bfx-\bfx\|_\infty, \|\hat \bfx-\bfx'\|_\infty\right\}. 
	\end{align*}
	(Here, $\hat\bfx$ is sorted to best match $\bfx$ in the first expression, while it is sorted to best match $\bfx'$ in the second expression.) Using the triangle inequality and that $\psi$ is arbitrary, we have 
	\begin{equation}
		\label{eq:Fhelp}
		X_*(\calS,\calN)
		\geq \frac 1 2 \|\bfx-\bfx'\|_\infty. 
	\end{equation}
	Repeating the same argument for the amplitudes instead yields
	\begin{equation}
		\label{eq:Ahelp}
		A_*(\calS,\calN)
		\geq \frac 1 2 \|\bfa-\bfa'\|_\infty.
	\end{equation}
	
	With these inequalities at hand, we are ready to prove this lemma. Consider $(\bfx,\bfa)\in \calS$ where $x_1=0$ and $a_1=r_0$ while $|x_j|\geq 8\pi/m$ for each $j\not=1$; if $s=1$, then $(\bfx,\bfa)=(0,r_0)$. Pick $\delta:=\epsilon/(4r_0 m^{1+1/p})$, and by our assumption on $\epsilon$, we have $\delta\leq 4\pi/m$. 
	
	Set $\bfeta$ such that $\eta_k=-r_0+(r_0+\epsilon/(4m^{1/p}))e^{ik\delta}$ for each $k$. Using that $|1-e^{ik\delta}|\leq k \delta$ and choice of $\delta$, we see that
	\begin{align*}
		\|\bfeta\|_p
		&\leq (2m-1)^{1/p} \|\bfeta\|_\infty \\
		&\leq (2m-1)^{1/p} \left( r_0 \max_k |1-e^{ik\delta}| + \frac{\epsilon}{4m^{1/p}} \right) \\
		&\leq 2 r_0 \delta m^{1+1/p} + \frac{1}{2} \epsilon
		= \epsilon.
	\end{align*}
	Consider the alternative parameters $(\bfx',\bfa')=\{(\delta,r_0+\epsilon/(4m^{1/p}))\}\cup \{(x_j,a_j)\}_{j=2}^s$; again $(\bfx',\bfa')=(\delta,r_0+\epsilon/(4m^{1/p}))$ if $s=1$. Then we see that $\tilde\bfy(\bfx,\bfa,\bfeta)=\tilde\bfy(\bfx',\bfa',\bfzero)$, which verifies relationship \eqref{eq:data}. Since $\delta\leq 4\pi/m$ as well, the current ordering of $\bfx$ and $\bfx'$ are the ones that minimize their matching distance, and so the order of $\bfa$ and $\bfa'$ is the one used in the amplitude error too. 
	
	By construction, $\bfx$ and $\bfx'$ have the same entries except $x_1=0$ and $x_1'=\delta$. Using \eqref{eq:Fhelp} and choice of $\delta$, we see that
	$$
	X_*(\calS,\calN)
	\geq \frac 12 \|\bfx-\bfx'\|_\infty
	=\frac 12 |x_1-x_1'|
	=\frac{\epsilon}{8r_0 m^{1+1/p}}. 
	$$
	For the amplitudes, note that $\bfa$ and $\bfa'$ have the same entries except $a_1=r_0$ and $a_1'=r_0+\epsilon/(4m^{1/p})$. Using \eqref{eq:Ahelp} and choice of $\delta$, we see that
	$$
	A_*(\calS,\calN)
	\geq \frac 12 \|\bfa-\bfa'\|_\infty
	=\frac 12 |a_1-a_1'|
	=\frac{\epsilon}{8 m^{1/p}}. 
	$$
	
	\subsection{Proof of \cref{lem:landscape1}}
	\label{proof:landscape1}
	
	It follows immediately by definition that $q$ maps to $[0,1]$. Since $\bfW\in \bbU^{m\times s}$ with $m>s$, it is rank deficient. This implies $q$ is not identically zero. Let $\bfW_\perp$ be an orthonormal basis for the orthogonal complement of $\bfW$. Define $m-s$ trigonometric polynomials $q_1,\dots,q_{m-s}$ such that the Fourier coefficients of $q_k$ are supported in $\{0,\dots,m-1\}$ and are the entries of the $k$-th column of $\bfW_\perp$. By Parseval and that the columns of $\bfW_\perp$ are orthonormal, we see that $q_1,\dots,q_{m-s}$ are orthonormal as well. Thus, we have
	\begin{equation}
		\label{eq:qformulas3} 
		q(t)
		=\bfphi(t)^*\bfW_\perp \bfW_\perp^* \bfphi(t)
		=\sum_{k=1}^{m-s} |q_k(t)|^2.
	\end{equation}
	This is a polynomial sum-of-squares representation. Notice that $|q_k|^2=q_k \overline{q_k}$ has Fourier coefficients supported in $\{-m+1,\dots,m-1\}$, which shows $q\in \calT_m$. By Bernstein's inequality for trigonometric polynomials together with $\|q\|_{L^\infty(\T)}\leq 1$, we see that
	$$
	\|q^{(\ell)} \|_{L^\infty(\T)} \leq (m-1)^\ell \|q\|_{L^\infty(\T)}
	\leq m^\ell. 
	$$
	This completes the proof. 
	
	\subsection{Proof of \cref{lem:bijection}}
	\label{proof:bijection}
	
	Let us prove that there is a bijection between (a) and (b). The map from sets of cardinality $s$ to Fourier subspaces is surjective by definition of $\bbF^{m\times s}$. It remains to show that it is injective. Suppose for the sake of deriving a contradiction that there are distinct $\bfx$ and $\bfx'$ sets of cardinality $s$ that generate the same Fourier subspaces $\bfU$. By \cref{lem:qroots2}, its associated landscape function $q_\bfU$ has roots only at $\bfx$, yet it only has roots at $\bfx'$ as well, which is a contradiction.  
	
	Now we prove that there is a bijection between (b) and (c). By definition, $q$ is a surjective map from $\bbF^{m\times s}$ to $\calF_{m,s}$. It remains to show that it is injective. Suppose for the sake of deriving a contradiction that there are distinct $\bfU,\bfV\in \bbF^{m\times s}$ such that $q_{\bfU}=q_{\bfV}$. Since there is a bijection between (a) and (b), we identify $\bfU,\bfV$ with distinct sets $\bfx$ and $\bfy$ of cardinality $s$. According to \cref{lem:qroots2}, the landscape function $q_{\bfU}=q_{\bfV}$ vanishes only at $\bfx$, and only at $\bfy$, which is a contradiction.
	
	\subsection{Proof of \cref{lem:landscape2}}
	\label{proof:landscape2}
	
	By definition of landscape function and subspace error, we have
	\begin{equation*}
		\big\|q_{\bfV} -q_{\bfW}\big\|_{L^\infty(\T)}
		= \sup_{t\in \T} \, \big|\bfphi(t)^* \big(\bfV\bfV^*- \bfW \bfW^* \big) \bfphi(t) \big|
		\leq \big\|\bfV\bfV^*-\bfW\bfW^* \big\|_2
		= \vartheta. 
	\end{equation*}
	This proves the lemma for $\ell=0$. By \cref{lem:landscape1}, $q_{\bfV}-q_{\bfW}\in \calT_m$. By Bernstein's and the previous inequality, we have
	\begin{equation*}
		\big\|q_{\bfV}^{(\ell)} -q_{\bfW}^{(\ell)}\big\|_{L^\infty(\T)}
		= \big\| (q_{\bfV}-q_{\bfW})^{(\ell)}\big\|_{L^\infty(\T)}
		\leq (m-1)^\ell \big\|q_{\bfV}-q_{\bfW}\big\|_{L^\infty(\T)}
		\leq \vartheta m^\ell. 
	\end{equation*}
	This completes the proof.

	\subsection{Proof of \cref{lem:grad}}
	\label{proof:grad}
	
	\begin{enumerate}[(a)]
		\item 
		If $u\not \in B_j$ for all $j\in \{1,\dots,s\}$, then by \cref{thm:landscape} part (e), we have $\tilde q(u)\geq 0.529 \geq \alpha$, so $u$ is rejected. This proves that 
		$$
		\left(\bigcup_{j=1}^s B_j\right)^c \subset A^c,
		$$
		which is equivalent to the first statement. 
		\item 
		By \cref{thm:landscape} part (c), mean value theorem, and \cref{lem:landscape1}, for all $t\in \T$ such that $|t-\tilde x_j|< (\alpha-\vartheta^2)/m$, we have
		$$
		\tilde q(t) 
		\leq \tilde q(\tilde x_j) + |\tilde q(t)-\tilde q(\tilde x_j)|
		\leq \vartheta^2 + |t-\tilde x_j|m
		< \alpha. 
		$$
		This shows that
		$$
		\tilde q(t) < \alpha \forallspace t\in \bigcup_{j=1}^s \left( \tilde x_j- \frac {\alpha-\vartheta^2}{m} , \, \tilde x_j+ \frac {\alpha-\vartheta^2}{m} \right) \subset \bigcup_{j=1}^s B_j.
		$$
		Each open interval in this union has length $2(\alpha-\vartheta^2)/m$. On the other hand, the largest gap between consecutive elements of $G$ is at most $2\mesh(G)< 2(\alpha-\vartheta^2)/m$. Hence, for each $j$, we have
		$$
		G\cap \left( \tilde x_j- \frac {\alpha-\vartheta^2}{m} , \, \tilde x_j+ \frac {\alpha-\vartheta^2}{m} \right)
		\not=\emptyset. 
		$$
		This implies $A\cap B_j$ is nonempty.
		\item 
		Suppose for the purpose of deriving a contradiction that $A_j$ is not a cluster in $G$. Then there are distinct $v,v'\in A_j$ and $u\in G$ such that $u\not\in A_j$. Since $u$ is rejected while $v,v'$ are accepted, we have 
		\begin{equation*}
			\tilde q(v)<\alpha, \quad \tilde  q(v')<\alpha, \andspace \tilde q(u)\geq \alpha.
		\end{equation*}
		Since the arc connecting $v,v'$ also lies in $B_j$, we have that $u\in B_j$ as well. However, \cref{thm:landscape} part (b) and (d) tell us that $\tilde q$ is strictly decreasing on $[\tilde x_j-4\pi /(3m), \tilde x_j)$ and then strictly increasing on $(\tilde x_j,\tilde x_j+4\pi/(3m)]$. This yields the desired contradiction. 
	\end{enumerate}
	
	\subsection{Proof of \cref{lem:rho}}
	\label{proof:nsrrelationship2}
	
	Since $\tilde \bfT=\bfT+T(\bfeta)$, by Weyl's inequality, for each $k\in \{1,\dots,m\}$, we have
	\begin{align*}
		\big|\sigma_k(\tilde \bfT)-\sigma_k(\bfT)\big|
		\leq \|T(\bfeta)\|_2.  
	\end{align*}
	In particular, this implies
	$
	\sigma_s(\tilde \bfT)
	\geq \sigma_s(\bfT) - \|T(\bfeta)\|_2.
	$
	Since $\sigma_{s+1}(\bfT)=0$ from factorization \eqref{eq:toeplitzfactorization}, we also see that	
	$\sigma_{s+1}(\tilde \bfT)\leq \|T(\bfeta)\|_2$. Then we have
	$$
	\sigma_s(\tilde \bfT)
	-\sigma_{s+1}(\tilde \bfT)
	\geq \sigma_s(\bfT) - 2\|T(\bfeta)\|_2 
	=(1-\rho) \sigma_s(\bfT). 
	$$
	The right side is positive because by the assumption that $\rho\leq 1-1/\sqrt 2$ and $\sigma_s(\bfT)>0$. This proves that $\tilde \bfT$ has a uniquely defined $s$-dimensional leading left singular space. Let $\tilde \bfU$ be an orthonormal basis for this space. 
	
	Now we are in position to apply Wedin's sine-theta theorem \cite[Chapter V, Theorem 4.4]{stewart1990matrix}, which provides the bound
	\begin{align*}
		\vartheta(\bfU,\tilde\bfU)
		=\big\|\tilde \bfU \tilde \bfU^* - \bfU\bfU^*\big\|_2
		\leq \frac{\sqrt 2 \, \|T(\bfeta)\|_2}{\sigma_s(\tilde \bfT)}.
	\end{align*}
	Using the previous inequalities and that $\rho\leq 1-1/\sqrt 2$ yields
	\begin{align*}
		\vartheta(\bfU,\tilde\bfU)
		\leq \frac{\sqrt 2 \,}{(1-\rho)} \frac{ \|T(\bfeta)\|_2}{\sigma_s(\bfT)}
		\leq 2 \frac{ \|T(\bfeta)\|_2}{\sigma_s(\bfT)}
		= \rho.
	\end{align*}

	\subsection{Proof of \cref{lem:sparsity}}
	\label{proof:sparsity}
	
	By definition of $\rho$, we have $\|T(\bfeta)\|_2=\rho \sigma_s(\bfT)/2$. By Weyl's inequality (also see the proof of \cref{lem:rho}), we have
	\begin{align*}
		\sigma_1(\tilde \bfT) 
		&\leq \sigma_1(\bfT) + \|T(\bfeta)\|_2 
		=\sigma_1(\bfT) + \frac 1 2 \rho \sigma_s(\bfT)
		\leq \left( 1+ \frac 12 \rho\right) \sigma_1(\bfT), \\
		\sigma_s(\tilde \bfT) 
		&\geq \sigma_s(\bfT)-\|T(\bfeta)\|_2 
		=\left( 1-\frac 12 \rho\right) \sigma_s(\bfT), \\
		\sigma_{s+1}(\tilde \bfT) 
		&\leq \|T(\bfeta)\|_2
		=\frac 12 \rho \sigma_s(\bfT).    
	\end{align*}
	Note that \cref{prop:AB19} together with the assumption that $(\bfx,\bfa)\in \calS(2\pi\beta/m,r_0,r_1)$ and factorization \eqref{eq:toeplitzfactorization} imply 
	$$
	r_0 m (1-1/\beta) 
	\leq a_{\min} \sigma_s^2(\bfPhi)
	\leq \sigma_s(\bfT)
	\leq \sigma_1(\bfT)
	\leq a_{\max} \sigma_1^2(\bfPhi)
	\leq r_1 m (1+1/\beta). 
	$$
	On one hand, since $\rho\leq 1-1/\sqrt 2$, we have
	$$
	\frac{\sigma_s(\tilde \bfT)}{\sigma_1(\tilde \bfT)}
	\geq \frac{2-\rho}{2+\rho} \frac{\sigma_s(\bfT)}{\sigma_1(\bfT)}
	\geq \frac{2-\rho}{2+\rho} \frac {r_0} {r_1} \frac{\beta-1}{\beta+1}
	\geq \frac 7{10} \frac {r_0} {r_1} \frac{\beta-1}{\beta+1}. 
	$$    
	On the other hand, using both upper bounds on $\rho$, we have  
	$$
	\frac{\sigma_{s+1}(\tilde \bfT)}{\sigma_1(\tilde \bfT)}
	\leq \frac{\sigma_{s+1}(\tilde \bfT)}{\sigma_s(\tilde \bfT)}
	\leq \frac{\rho}{2-\rho}
	\leq \frac{\rho}{1+1/\sqrt 2}
	< \frac 7{10} \frac {r_0} {r_1} \frac{\beta-1}{\beta+1}. 
	$$
	This completes the proof.

	\subsection{Proof of \cref{lem:noisepnorm}}
	\label{proof:noisepnorm}
	
	We extend $\bfeta\in\C^{2m-1}$ to a sequence $\eta = \{\eta_k\}_{k\in \Z}$ via zero extension. Fix any vector $\bfu\in \C^{m}$ such that $\|\bfu\|_2=1$, which we also extend to a sequence $u =  \{u_k\}_{k\in \Z}$ via zero extension. Let $*$ be the convolution operator on $\Z$. For any $j\in \{-m,\dots,m-1\}$, we have
	$$
	(T(\bfeta) \bfu)_j 
	= \sum_{k\in\Z} \eta_{j-k} u_k
	= (\eta* u)_j. 
	$$
	By Young's convolution inequality (where the group here is $\Z$), we see that
	$$
	\|T(\bfeta) \bfu\|_2
	=\|\eta*u\|_{\ell^2}
	\leq \|\eta\|_{\ell^1} \|u\|_{\ell^2}
	= \|\bfeta\|_1 \|\bfu\|_2
	= \|\bfeta\|_1. 
	$$
	Since $\bfu$ was arbitrary, this shows that $\|T(\bfeta)\|_2\leq \|\bfeta\|_1$. Using H\"older's inequality, we see that
	$$
	\|\bfeta\|_1
	\leq (2m-1)^{1-1/p} \|\bfeta\|_p
	\leq 2 m^{1-1/p} \|\bfeta\|_p. 
	$$
	This completes the proof.

	\subsection{Proof of Lemma \ref{lem:toeplitznorm}}
	\label{sec:toeplitznorm}

	The proof is an application of a matrix concentration inequality \cite[Theorem 4.1.1]{tropp2015introduction}. Following the notation of the reference, let $\nu(\cdot)$ denote the variance statistic of a sum of random matrices. For each $\ell\in \{-m+1,\dots,m-1\}$, let $\bfA_\ell\in \R^{m\times m}$ such that it has all zeros except it has ones on the $\ell$-th diagonal; $(A_\ell)_{j,k} = 1$ if and only if $k-j=\ell$ and $(A_{\ell})_{j,k}=0$ otherwise. 
	
	We let $\{s_\ell^2\}_{\ell=-m+1}^{m-1}$ denote the diagonal entries of $\bfSigma$. Let $\bfw\sim\calN(\bfzero,\bfI)$ so that $\bfeta=\bfSigma^{1/2}\bfw$ and $\eta_\ell=s_\ell w_\ell$ for each $\ell$. Since $\bfT$ is constant on each diagonal, we have  
	$$
	T(\bfeta)
	= \sum_{\ell=-m+1}^{m-1} \eta_\ell \bfA_\ell
	= \sum_{\ell=-m+1}^{m-1} w_\ell (s_\ell \bfA_\ell). 
	$$
	
	We need to compute $\bfA_\ell \bfA_\ell^*$ and $\bfA_\ell^* \bfA_\ell$ in order to apply the referenced matrix concentration inequality. For $\ell \in \{0,\dots, m-1\}$, we see that $(\bfA_\ell \bfA_\ell^*)_{k,k} = 1$ for each $k\in \{0,\dots,m-1-\ell\}$ and all other entries are zero; for $\ell \in \{-m+1,\dots,0\}$, we see that $(\bfA_\ell \bfA_\ell^*)_{k,k} = 1$ for each $k\in \{m-1+\ell,\dots,m-1\}$ and all other entries are zero. From here, we see that $\sum_{\ell=-m+1}^{m-1}  s_\ell^2 \bfA_\ell \bfA_\ell^*$ is a diagonal matrix and whose $\ell$-th diagonal entry is $s^2_\ell+s^2_{\ell+1}+\ldots+s^2_{\ell+m-1}$. For $\bfA^*_\ell \bfA_\ell$, we get the same conclusion except with reverse indexing since $\bfA_\ell^*=\bfA_{-\ell}$.  
	
	Thus, the matrix variance statistic of $T(\bfeta)$ is
	\begin{align*}
		\nu(T(\bfeta)) 
		:=\left\|\sum_{\ell=-m+1}^{m-1}  s_\ell^2 \bfA_\ell \bfA_\ell^* \right\|_2 
		=\max_{k=-m+1,\ldots, 0} \left(s^2_k+s^2_{k+1}+\ldots+s^2_{k+m-1} \right)
		\leq \tr(\bfSigma). 
	\end{align*}
	Applying the referenced matrix concentration inequality completes the proof.
	
	\subsection{Proof of \cref{lem:amplitudes}}
	\label{proof:amplitudes}
	
	Simple calculations with \eqref{eq:hatAhelp} yield
	\begin{align*}
		\hat\bfA -\bfA 
		&= \big(\hat \bfPhi^+ \bfPhi \big) \bfA \big(\hat \bfPhi^+ \bfPhi - \bfI\big)^* + \big( \hat \bfPhi^+ \bfPhi - \bfI\big) \bfA + \hat \bfPhi^+ T(\bfeta) (\hat \bfPhi^+)^*. 
	\end{align*}
	From here, we  get the inequality
	\begin{equation}
		\label{eq:hatAhelp2}
		\big\|\hat\bfA -\bfA\big\|_2 
		\leq  \big(1+\big\|\hat \bfPhi^+ \bfPhi \big\|_2\big) \big\|\hat \bfPhi^+ \bfPhi - \bfI\big\|_2 \|\bfA\|_2 + \big\| \hat \bfPhi^+ \big\|_2^2 \, \|T(\bfeta)\|_2.
	\end{equation}
	It suffices to control all of the terms that appear here. We have 
	$$
	\big\|\hat \bfPhi^+ \bfPhi - \bfI \big\|_2 
	=\big\| \big(\hat \bfPhi^* \hat\bfPhi \big)^{-1} \hat\bfPhi^* \bfPhi - \bfI \big\|_2
	\leq \big\| (\hat \bfPhi^* \hat\bfPhi)^{-1} \big\|_2 \big\| \hat\bfPhi \big\|_2 \big\|\bfPhi - \hat \bfPhi \big\|_2.
	$$
	By our assumptions on $\bfx$ and $\hat\bfx$, we have $\Delta(\hat \bfx) \geq 2\pi (\beta-2\alpha)/m\geq \pi \beta/m$. From here, we use \cref{prop:AB19} to see that 
	\begin{align*}
		\big\|\hat \bfPhi^+ \bfPhi \big\|_2
		&\leq \frac{\sigma_{\max}(\bfPhi)}{\sigma_{\min}(\hat \bfPhi)} 
		\leq \sqrt{\frac{1+2/\beta}{1-2/\beta}}
		=\sqrt{\frac{\beta+2}{\beta-2}}, \\
		\big\| (\hat \bfPhi^* \hat\bfPhi)^{-1} \big\|_2 \big\| \hat\bfPhi \big\|_2
		&\leq \frac{\sigma_{\max}(\hat \bfPhi)}{\sigma_{\min}^2(\hat \bfPhi)}
		\leq \frac{\sqrt{1+2/\beta}}{1-2/\beta} \frac 1 {\sqrt m}.
	\end{align*}
	It remains to control $\big\|\bfPhi - \hat \bfPhi \big\|_2$, which we would like to upper bound independent of $s$. Recall that the columns of $\bfPhi$ are $\sqrt m \, \bfphi(x_j)$ for $j=1,\dots,s$. For any $t\in [0,1]$ and $j=1,\dots,s$, define 
	$
	x_{t,j}=x_j+t(\hat x_j-x_j). 
	$
	Also let $\bfx_t := \{x_{t,j}\}_{j=1}^s$ and $\bfM := \diag(0,i,\dots,i(m-1))$. A calculation shows that 
	$$
	\bfphi'(x_{t,j}) =\bfM \bfphi(x_{t,j}).
	$$
	Importantly, this formula holds for each $t$ and $j$, with the same matrix $\bfM$. By the fundamental theorem of calculus,
	$$
	\bfphi(\hat x_j)-\bfphi(x_j)
	= \int_0^1 \frac d {dt} (\bfphi(x_{t,j})) \, dt
	= \int_0^1 \bfM \bfphi(x_{t,j}) (\hat x_j-x_j)\, dt.	
	$$
	This formula now implies 
	\begin{equation}
		\label{eq:intformula}
		\hat\bfPhi-\bfPhi
		= \int_0^1 \bfM \bfPhi(m,\bfx_{t}) \diag(\hat\bfx-\bfx) \, dt.
	\end{equation}
	By a direct computation, we see that 
	$$
	\Delta(\bfx_t)
	\geq \Delta(\bfx) - 2t \|\hat \bfx-\bfx\|_\infty
	\geq \frac{2\pi (\beta-2\alpha)}m
	\geq \frac{\pi \beta}m. 	
	$$
	Note this bound is uniform in $t$. Utilizing \eqref{eq:intformula}, \cref{prop:AB19}, and that $\|\bfM\|_2\leq m$, we have
	\begin{equation*}
		\left\|\bfPhi-\hat\bfPhi \right\|_2
		\leq \|\bfM\|_2 \left(\sup_{t\in [0,1]} \|\bfPhi(m,\bfx_{t})\|_2 \right) \|\diag(\hat\bfx-\bfx)\|_2
		\leq \sqrt{1+ 2 /\beta} \, m^{3/2} \,  \|\hat\bfx-\bfx\|_\infty. 
	\end{equation*}
	Combining all of the above and inserting into \eqref{eq:hatAhelp2}, we see that 
	$$
	\big\| \hat\bfA -\bfA \big\|_2 
	\leq \left( 1 + \sqrt{\frac{\beta+2}{\beta-2}} \right) \frac{\beta + 2}{\beta-2} \, a_{\max} \, m  \|\hat\bfx-\bfx\|_\infty + \frac \beta {\beta -1 } \frac{\|T(\bfeta)\|_2} m. 
	$$

	The proof is complete once we note that for each $j=1,\dots,s$, we have
	$$
	|a_j-\hat a_j|
	= \big\|\bfe_j^* \big(\bfA-\hat \bfA \big) \bfe_j \big\|_2
	\leq \big\|\bfA -\hat\bfA \big\|_2. 
	$$
	
	\subsection{Proof of Lemma \ref{lem:toeplitznorm2}}
	\label{proof:toeplitznorm2}
	
	The proof is an application of a matrix concentration inequality \cite[Theorem 4.1.1]{tropp2015introduction} and a modification of the proof of \cref{lem:toeplitznorm}. 
	
	Let $\{s_\ell^2\}_{\ell=0}^{2m-1}$ denote the diagonal entries of $\bfSigma$. Let $\bfeta\in\calN(\bfzero,\bfSigma)$ so that $\bfeta=\bfSigma^{1/2}(\bfu+i\bfv)/\sqrt 2$, where $\bfu,\bfv$ are independent standard normal random vectors. Note that for each $\ell\in\{1,\dots,2m-1\}$,
	$$
	\zeta_\ell = \eta_\ell=\frac{1}{\sqrt 2} s_\ell u_\ell + i \frac{1}{\sqrt 2} s_\ell v_\ell, 
	\andspace \zeta_{-\ell} = \overline{\eta_\ell}=\frac{1}{\sqrt 2} s_\ell u_\ell - i \frac{1}{\sqrt 2} s_\ell v_\ell.
	$$
	
	For each $\ell\in\{1,\dots,2m-1\}$, we define the matrices $\bfA_\ell,\bfB_\ell\in \R^{2m\times 2m}$ in the following way. $\bfA_\ell$ has all zeros except it has ones on the $\ell$-th and $-\ell$-th diagonal; $(A_\ell)_{j,k} = 1$ if and only if $|k-j|=\ell$ and $(A_{\ell})_{j,k}=0$ otherwise. Let $\bfB_\ell\in \R^{2m\times 2m}$ such that it has all zeros except it has ones on the $\ell$-th and minus ones $-\ell$-th diagonal. To simplify notation, let $\bfA_0=\bfB_0=\bfI_{2m}$. Thus, we see that 
	$$
	T(\bfzeta)
	= \sum_{\ell=0}^{2m-1} u_\ell \left(\frac 1{\sqrt 2} s_\ell \bfA_\ell \right) + i \sum_{\ell=0}^{2m-1} v_\ell \left( \frac 1{\sqrt 2} s_\ell \bfB_\ell\right). 
	$$
	
	Call the first and second sums $\bfA$ and $i \bfB$ respectively. Since $\bfu$ and $\bfv$ are independent, they are a sum of independent and real random matrices. To control $\|T(\bfzeta)\|_2$, we control the spectral norm of the two sums individually via matrix concentration and then use the union bound. Notice that $\bfA_\ell \bfA_\ell^*$, $\bfA_\ell^* \bfA_\ell$, $\bfB_\ell \bfB_\ell^*$, and $\bfB_\ell \bfB_\ell^*$ are all diagonal and their diagonals are bounded by 2 in absolute value. Thus, we have
	\begin{align*}
		\max\{\nu(\bfA),\nu(\bfB)\} 
		\leq \sum_{\ell=0}^{2m-1} s_\ell^2
		= \tr(\bfSigma). 
	\end{align*}
	Applying the referenced matrix concentration inequality establishes that 
	$$
	\E \|\bfA\|_2 
	\leq \sqrt{2 \tr(\bfSigma) \log(4m)} 
	\andspace \P(\|\bfA\|_2\geq t)
	\leq 4m e^{-t^2/(2\tr(\bfSigma))}.
	$$
	We get the same upper bounds for $\bfB$ instead of $\bfA$. Using that $\bfu$ and $\bfv$ are independent, $\|T(\bfzeta)\|_2\leq \|\bfA\|_2+\|\bfB\|_2$, and the union bound completes the proof.

	\section*{Acknowledgments}
	
	W. Li is supported by NSF-DMS Award \#2309602 and a Cycle 55 PSC-CUNY award. W. Liao is supported by NSF-DMS Award \#2145167 and DOE Award \#SC0024348. The authors thank Dmitry Batenkov for helpful discussions and for pointing out some references. The authors thank anonymous reviewers for their feedback and suggestions. In particular, the authors are extremely grateful for a generous suggestion which lead to an improvement in \cref{lem:amplitudes}.
	
	\section*{Data Availability Statement}
	
	No new data were generated or analyzed in support of this review. 
	
	\bibliographystyle{plain}
	\bibliography{MUSICoptimalitybib}

\end{document}